\documentclass[11pt]{article}

\usepackage{amsmath} 
\usepackage{amsthm} 
\usepackage{amssymb}	
\usepackage{graphicx} 
\usepackage{multicol} 
\usepackage{multirow}
\usepackage{color}
\usepackage[dvips,letterpaper,margin=1in,bottom=1in]{geometry}

\usepackage[utf8]{inputenc}
\usepackage[english]{babel}

\usepackage{diagbox}
\usepackage{mathtools}

\newtheorem{theorem}{Theorem}[section]

\newtheorem{lemma}[theorem]{Lemma}
\newtheorem{corollary}[theorem]{Corollary}
\newtheorem{proposition}[theorem]{Proposition}
\newtheorem{definition}[theorem]{Definition}
\newtheorem{observation}[theorem]{Observation}

\newtheorem{remark}[theorem]{Remark}
\newtheorem{problem}{Problem}

\DeclarePairedDelimiter\rbra{\lparen}{\rparen}
\DeclarePairedDelimiter\sbra{\lbrack}{\rbrack}
\DeclarePairedDelimiter\cbra{\{}{\}}
\DeclarePairedDelimiter\abs{\lvert}{\rvert}

\DeclarePairedDelimiter\ceil{\lceil}{\rceil}
\DeclarePairedDelimiter\floor{\lfloor}{\rfloor}

\newcommand{\substr}[2] {\sbra*{#1 .. #2}}

\DeclareMathOperator*{\argmax}{arg\,max}

\newcommand{\set}[2] {\left\{\, #1 \colon #2 \,\right\}}

\newcommand{\polylog} {\operatorname{polylog}}

\usepackage{latexsym}
\usepackage{CJK}

\usepackage{enumerate}

\usepackage{algorithm}
\usepackage{algpseudocode}

\usepackage{stmaryrd}
\usepackage{booktabs}

\usepackage[pagebackref,colorlinks]{hyperref}
\newcommand{\footremember}[2]{%
    \footnote{#2}
    \newcounter{#1}
    \setcounter{#1}{\value{footnote}}%
}

\usepackage{cite}
\usepackage{bbm}

\usepackage{tablefootnote}
\usepackage{threeparttable}
\usepackage{scrextend}

\usepackage{tikz}
\usepackage{subcaption}

\begin{document}

    \title{Double-Ended Palindromic Trees in Linear Time\footnote{This paper is the full version of \cite{WYZ25}.}}
        \author{
            Qisheng Wang \footremember{1}{Qisheng Wang is with the School of Informatics, University of Edinburgh, Edinburgh, United Kingdom (e-mail: \href{mailto:QishengWang1994@gmail.com}{\nolinkurl{QishengWang1994@gmail.com}}). Part of the work of Qisheng Wang was done when the author was with the Graduate School of Mathematics, Nagoya University, Japan.}
            \and Ming Yang \footremember{2}{Ming Yang is with China Telecom Cloud Technology Co., Ltd., Beijing, China.}
            \and Xinrui Zhu \footremember{3}{Xinrui Zhu is with XVERSE Technology Inc., Shenzhen, China.}
        }
        \date{}
        \maketitle

    \begin{abstract}
        The palindromic tree (a.k.a.\ eertree) is a data structure that provides access to all palindromic substrings of a string. 
        In this paper, we propose a dynamic version of eertree, called double-ended eertree, which supports online operations on the stored string, including double-ended queue operations, counting distinct palindromic substrings, and finding the longest palindromic prefix/suffix.
        At the heart of our construction, we identify a new class of substring occurrences, called surfaces, that are palindromic substring occurrences that are neither prefixes nor suffixes of any other palindromic substring occurrences, which is of independent interest. 
        Surfaces characterize the link structure of all palindromic substrings in the eertree, thereby allowing a linear-time implementation of double-ended eertrees through a linear-time maintenance of surfaces.
    \end{abstract}

    \textbf{Keywords: Palindromes, eertrees, double-ended data structures, string algorithms.}

    \newpage

    \tableofcontents
    \newpage

    \section{Introduction}
    
    \paragraph{Palindromes.}
    
    Palindromes, which read the same backward as forward, are an interesting and important concept in many fields, e.g., linguistics \cite{Ber73}, mathematics \cite{HS98,BHS04}, physics \cite{HKS95}, biology \cite{Lea94,LLR08,GHD+03}, and computer science \cite{Gus97}. Especially, combinatorial properties of palindrome complexity, i.e., the number of palindromes, of finite and infinite strings have been extensively studied in \cite{Dro95,All97,DP99,Baa99,Dam00,DZ00,ABCD03,BHNR04,AAK10,RS16}. 
    
    Palindromes as strings were first studied from an algorithmic perspective by Slisenko \cite{Sli73} (cf. \cite{Gal78}). Knuth, Morris and Pratt \cite{KMP77} developed the well-known string matching algorithm, and applied it in recognizing concatenations of even-length palindromes. 
    Soon after, this result was improved in \cite{GS78} to
    concatenations of non-trivial palindromes.
    Manacher \cite{Man75} invented a linear-time algorithm to find all palindromic prefixes of a string. Later, it was found in \cite{ABG95} that Manacher's algorithm can be used to find all maximal palindromic substrings, and an alternative algorithm for the same problem was proposed in \cite{Jeu94} based on suffix trees \cite{Wei73,McC76,Ukk95,Far97}. Several parallel algorithms to find palindromic substrings were also designed \cite{CR91,ABG95,BG95}. Finding the longest palindromic substring was studied in the streaming model \cite{BEMTSA14,GMSU19}, after single-character substitution \cite{FNI+21}, and with extensions to finding top-$k$ longest palindromes in substrings \cite{MMSH23}. Recently, it was shown in \cite{CPR22} that the longest palindromic substring can be found in sublinear time. In addition, a quantum algorithm for this problem was found in \cite{LGS22} with quadratic speedup over classical algorithms. Very recently, computing maximal generalized palindromes was studied in \cite{FMN+22}.
    
    Since it was pointed out in \cite{DJP01} that a string $s$ of length $\abs*{s}$ has at most $\abs*{s} + 1$ distinct palindromic substrings (including the empty string), 
    many algorithms concerning palindromes have been successively proposed. Groult, Prieur and Richomme \cite{GPR10} found a linear-time algorithm to count the number of distinct palindromic substrings of a string, and used it to check the palindromic richness of a string. Here, a string $s$ is called palindromic rich if it contains the maximum possible $\abs*{s} + 1$ distinct palindromic substrings \cite{GJWZ09} (which was further studied in \cite{RR09,BDLGZ09,Ves14,GSS16,Ruk17}). 
    Palindrome pattern matching was studied in \cite{IIT13}, where two strings are matched if they have the same positions of all their maximal palindromic substrings \cite{IIBT10}.
    Another line of research focused on concatenations of palindromes. 
    The palindromic length of a string is the minimal $k$ such that it is a concatenation of $k$ palindromes \cite{Rav03}. 
    The problem of recognizing concatenations of exactly $k$ palindromes was explicitly stated in \cite{GS78}, with an $O\rbra{kn}$-time algorithm presented in \cite{KRS15} and later an improved $O\rbra{n}$-time algorithm presented in \cite{RS20}.
    In \cite{FGKK14,ISI+14,RS18}, they proposed several $O\rbra{n \log n}$-time algorithms for computing the palindromic length, which was later settled in \cite{BKRS17} with an $O\rbra{n}$-time algorithm.
    
    Recently, Rubinchik and Shur \cite{RS18} proposed a linear-size data structure called eertree (also known as the palindromic tree), which stores all distinct palindromic substrings of a string and can be constructed online in linear time. The size of the eertree is much smaller than the length $n$ of the string on average, because 
    the expected number of palindromic substrings is $\Theta\rbra{\sqrt{n}}$ \cite{RS16}.
    Using this powerful data structure, they enumerated strings that are palindromic rich of length up to $60$ (cf. sequence A216264 in OEIS \cite{Sha13}), and reproduced a different algorithm from \cite{FGKK14,ISI+14} to compute the palindromic length of a string. Later, an online algorithm to count palindromes in substrings was designed in \cite{RS17} based on eertrees. Furthermore, Mieno, Watanabe, Nakashima, et al. \cite{MWN+22} developed a type of eertree for a sliding window. 
    
    \paragraph{Double-ended data structures.}
    
    A double-ended queue (abbreviated to deque, cf. \cite{Tar83,Knu97}) is an abstract data structure consisting of a list of elements on which four kinds of operations can be performed:
    \begin{itemize}
        \item $\texttt{push\_back}(c)$: Insert an element $c$ at the back of the deque. 
        \item $\texttt{push\_front}(c)$: Insert an element $c$ at the front of the deque.
        \item $\texttt{pop\_back}()$: Remove an element from the back of the deque.
        \item $\texttt{pop\_front}()$: Remove an element from the front of the deque.
    \end{itemize}
    If only operations allowed are \texttt{push\_back} and \texttt{pop\_back} (or \texttt{push\_front} and \texttt{pop\_front}), then the deque becomes a stack. If only \texttt{push\_back} and \texttt{pop\_front} (or \texttt{push\_front} and \texttt{pop\_back}) are allowed, then the deque becomes a queue. 
    
    As a linear data structure, deques are widely used in practical and theoretical computer science with their implementations supported in most high-level programming languages, e.g., C/C++, Java, Python, etc. 
    Deques were investigated in different computational models, including Turing machines \cite{Kos79}, RAMs (random access machines) \cite{Kos94} and functional programming \cite{Hoo82,CG93,Oka95,BT95,KT99}. In addition to the basic deque operations, a natural extension of deque is to maintain more useful information related to the elements stored in it. For example, mindeque (a.k.a. deque with heap order) \cite{GT86} is a kind of extended deque that supports the \texttt{find\_min} operation, which finds the minimal element in the deque. Furthermore, a more powerful extension of mindeque, called catenable mindeque, was developed in \cite{BST95}, which supports catenation of two mindeques.
    
    \subsection{Main results} \label{sec:main-results}
    
    In this paper, we study an extension of deque, called double-ended eertree, which processes a string $s$ (initialized to empty) with four kinds of basic operations supported:
   \begin{itemize}
        \item $\texttt{push\_back}(c)$: Insert a character $c$ at the back of the string. That is, set $s \gets sc$.
        \item $\texttt{push\_front}(c)$: Insert a character $c$ at the front of the string. That is, set $s \gets cs$.
        \item $\texttt{pop\_back}()$: Remove a character from the back of the string. That is, set $s \gets s\substr{1}{\abs*{s}-1}$, provided that $s$ is not empty.
        \item $\texttt{pop\_front}()$: Remove a character from the front of the string. That is, set $s \gets s\substr{2}{\abs*{s}}$, provided that $s$ is not empty.
    \end{itemize}
    In addition to the basic operations above, we can also make online queries to the current string, including but not limited to the following:
    \begin{itemize}
        \item $\texttt{count}()$: Find the number of distinct palindromic substrings of the current string. 
        \item $\texttt{longest\_prepal}()$: Find the longest palindromic prefix of the current string, and check whether it is unique in the current string. 
        \item $\texttt{longest\_sufpal}()$: Find the longest palindromic suffix of the current string, and check whether it is unique in the current string. 
    \end{itemize}

\newcommand{\node} {\operatorname{node}}
\newcommand{\prenode} {\mathit{prenode}}
\newcommand{\sufnode} {\mathit{sufnode}}
\newcommand{\prepal} {\mathit{prepal}}
\newcommand{\sufpal} {\mathit{sufpal}}
\newcommand{\prelen} {\mathit{prelen}}
\newcommand{\suflen} {\mathit{suflen}}
\newcommand{\eertree} {\mathsf{EERTREE}}
\newcommand{\even} {\mathsf{even}}
\newcommand{\odd} {\mathsf{odd}}
\renewcommand{\next} {\operatorname{next}}
\newcommand{\link} {\operatorname{link}}
\newcommand{\nullptr} {\mathsf{null}}
\newcommand{\str} {\operatorname{str}}
\newcommand{\len} {\operatorname{len}}
\newcommand{\linkcnt} {\operatorname{linkcnt}}
\newcommand{\occur} {\operatorname{occur}}
\newcommand{\presurf} {\mathit{presurf}}
\newcommand{\sufsurf} {\mathit{sufsurf}}

\newcommand{\precnt} {\mathit{precnt}}
\newcommand{\sufcnt} {\mathit{sufcnt}}
\newcommand{\subtree} {\mathit{subtree}}
\newcommand{\cnt} {\mathit{cnt}}

    For illustration, we consider the double-ended eertree of the string $s = abaaaba$, denoted $\eertree\rbra{s}$. 
    In Figure \ref{fig:push_back}, we present the structure of the eertree after performing $\texttt{push\_back}\rbra{a}$ on $\eertree\rbra{s}$; in Figure \ref{fig:pop_back}, we present the structure of the eertree after performing $\texttt{pop\_back}\rbra{}$ on $\eertree\rbra{s}$.

    \begin{figure}
    \centering
\tikzset{every picture/.style={line width=0.75pt}} 

\begin{tikzpicture}[x=0.75pt,y=0.75pt,yscale=-1,xscale=1]

\draw   (104,116.07) .. controls (104,92.06) and (123.33,72.59) .. (147.17,72.59) .. controls (171.01,72.59) and (190.34,92.06) .. (190.34,116.07) .. controls (190.34,140.09) and (171.01,159.55) .. (147.17,159.55) .. controls (123.33,159.55) and (104,140.09) .. (104,116.07) -- cycle ;
\draw   (368.09,111.48) .. controls (368.09,87.47) and (387.42,68) .. (411.26,68) .. controls (435.1,68) and (454.43,87.47) .. (454.43,111.48) .. controls (454.43,135.49) and (435.1,154.96) .. (411.26,154.96) .. controls (387.42,154.96) and (368.09,135.49) .. (368.09,111.48) -- cycle ;
\draw   (469.66,269.54) .. controls (469.66,245.52) and (488.99,226.05) .. (512.83,226.05) .. controls (536.67,226.05) and (556,245.52) .. (556,269.54) .. controls (556,293.55) and (536.67,313.02) .. (512.83,313.02) .. controls (488.99,313.02) and (469.66,293.55) .. (469.66,269.54) -- cycle ;
\draw   (266.52,267.5) .. controls (266.52,243.49) and (285.84,224.02) .. (309.69,224.02) .. controls (333.53,224.02) and (352.85,243.49) .. (352.85,267.5) .. controls (352.85,291.51) and (333.53,310.98) .. (309.69,310.98) .. controls (285.84,310.98) and (266.52,291.51) .. (266.52,267.5) -- cycle ;
\draw   (266.52,408.17) .. controls (266.52,384.16) and (285.84,364.69) .. (309.69,364.69) .. controls (333.53,364.69) and (352.85,384.16) .. (352.85,408.17) .. controls (352.85,432.19) and (333.53,451.65) .. (309.69,451.65) .. controls (285.84,451.65) and (266.52,432.19) .. (266.52,408.17) -- cycle ;
\draw   (469.66,412.77) .. controls (469.66,388.75) and (488.99,369.29) .. (512.83,369.29) .. controls (536.67,369.29) and (556,388.75) .. (556,412.77) .. controls (556,436.78) and (536.67,456.25) .. (512.83,456.25) .. controls (488.99,456.25) and (469.66,436.78) .. (469.66,412.77) -- cycle ;
\draw  [color={rgb, 255:red, 0; green, 0; blue, 0 }  ,draw opacity=1 ][fill={rgb, 255:red, 74; green, 144; blue, 226 }  ,fill opacity=1 ] (469.66,553.44) .. controls (469.66,529.43) and (488.99,509.96) .. (512.83,509.96) .. controls (536.67,509.96) and (556,529.43) .. (556,553.44) .. controls (556,577.45) and (536.67,596.92) .. (512.83,596.92) .. controls (488.99,596.92) and (469.66,577.45) .. (469.66,553.44) -- cycle ;
\draw   (266.52,689.52) .. controls (266.52,665.51) and (285.84,646.04) .. (309.69,646.04) .. controls (333.53,646.04) and (352.85,665.51) .. (352.85,689.52) .. controls (352.85,713.53) and (333.53,733) .. (309.69,733) .. controls (285.84,733) and (266.52,713.53) .. (266.52,689.52) -- cycle ;
\draw   (266.52,548.85) .. controls (266.52,524.83) and (285.84,505.37) .. (309.69,505.37) .. controls (333.53,505.37) and (352.85,524.83) .. (352.85,548.85) .. controls (352.85,572.86) and (333.53,592.33) .. (309.69,592.33) .. controls (285.84,592.33) and (266.52,572.86) .. (266.52,548.85) -- cycle ;
\draw   (104,269.54) .. controls (104,245.52) and (123.33,226.05) .. (147.17,226.05) .. controls (171.01,226.05) and (190.34,245.52) .. (190.34,269.54) .. controls (190.34,293.55) and (171.01,313.02) .. (147.17,313.02) .. controls (123.33,313.02) and (104,293.55) .. (104,269.54) -- cycle ;
\draw    (411.26,154.96) -- (311.34,222.89) ;
\draw [shift={(309.69,224.02)}, rotate = 325.79] [fill={rgb, 255:red, 0; green, 0; blue, 0 }  ][line width=0.08]  [draw opacity=0] (12,-3) -- (0,0) -- (12,3) -- cycle    ;
\draw    (411.26,154.96) -- (506.14,224.87) ;
\draw [shift={(507.75,226.05)}, rotate = 216.38] [fill={rgb, 255:red, 0; green, 0; blue, 0 }  ][line width=0.08]  [draw opacity=0] (12,-3) -- (0,0) -- (12,3) -- cycle    ;
\draw    (309.69,310.98) -- (309.69,362.69) ;
\draw [shift={(309.69,364.69)}, rotate = 270] [fill={rgb, 255:red, 0; green, 0; blue, 0 }  ][line width=0.08]  [draw opacity=0] (12,-3) -- (0,0) -- (12,3) -- cycle    ;
\draw    (512.83,313.02) -- (512.83,367.29) ;
\draw [shift={(512.83,369.29)}, rotate = 270] [fill={rgb, 255:red, 0; green, 0; blue, 0 }  ][line width=0.08]  [draw opacity=0] (12,-3) -- (0,0) -- (12,3) -- cycle    ;
\draw    (309.69,451.65) -- (309.69,503.37) ;
\draw [shift={(309.69,505.37)}, rotate = 270] [fill={rgb, 255:red, 0; green, 0; blue, 0 }  ][line width=0.08]  [draw opacity=0] (12,-3) -- (0,0) -- (12,3) -- cycle    ;
\draw    (309.69,592.33) -- (309.69,644.04) ;
\draw [shift={(309.69,646.04)}, rotate = 270] [fill={rgb, 255:red, 0; green, 0; blue, 0 }  ][line width=0.08]  [draw opacity=0] (12,-3) -- (0,0) -- (12,3) -- cycle    ;
\draw [color={rgb, 255:red, 74; green, 144; blue, 226 }  ,draw opacity=1 ]   (512.83,456.25) -- (512.83,507.96) ;
\draw [shift={(512.83,509.96)}, rotate = 270] [fill={rgb, 255:red, 74; green, 144; blue, 226 }  ,fill opacity=1 ][line width=0.08]  [draw opacity=0] (12,-3) -- (0,0) -- (12,3) -- cycle    ;
\draw    (147.17,159.55) -- (147.17,224.05) ;
\draw [shift={(147.17,226.05)}, rotate = 270] [fill={rgb, 255:red, 0; green, 0; blue, 0 }  ][line width=0.08]  [draw opacity=0] (12,-3) -- (0,0) -- (12,3) -- cycle    ;
\draw [color={rgb, 255:red, 74; green, 144; blue, 226 }  ,draw opacity=1 ] [dash pattern={on 4.5pt off 4.5pt}]  (469.66,553.44) .. controls (404.91,404.32) and (171.29,542.95) .. (147.17,313.02) ;
\draw [shift={(147.17,313.02)}, rotate = 84.01] [fill={rgb, 255:red, 74; green, 144; blue, 226 }  ,fill opacity=1 ][line width=0.08]  [draw opacity=0] (10.72,-5.15) -- (0,0) -- (10.72,5.15) -- (7.12,0) -- cycle    ;
\draw  [dash pattern={on 4.5pt off 4.5pt}]  (472.2,395.89) .. controls (418.43,372.96) and (370.81,348.56) .. (341.48,300.41) ;
\draw [shift={(340.16,298.19)}, rotate = 59.52] [fill={rgb, 255:red, 0; green, 0; blue, 0 }  ][line width=0.08]  [draw opacity=0] (10.72,-5.15) -- (0,0) -- (10.72,5.15) -- (7.12,0) -- cycle    ;
\draw  [dash pattern={on 4.5pt off 4.5pt}]  (293.18,227.08) .. controls (281.92,154.02) and (243.57,114.46) .. (192.67,115.98) ;
\draw [shift={(190.34,116.07)}, rotate = 356.91] [fill={rgb, 255:red, 0; green, 0; blue, 0 }  ][line width=0.08]  [draw opacity=0] (10.72,-5.15) -- (0,0) -- (10.72,5.15) -- (7.12,0) -- cycle    ;
\draw  [dash pattern={on 4.5pt off 4.5pt}]  (477.28,243.7) .. controls (435.8,190.53) and (286.06,112.29) .. (193.14,115.94) ;
\draw [shift={(190.34,116.07)}, rotate = 356.68] [fill={rgb, 255:red, 0; green, 0; blue, 0 }  ][line width=0.08]  [draw opacity=0] (10.72,-5.15) -- (0,0) -- (10.72,5.15) -- (7.12,0) -- cycle    ;
\draw  [dash pattern={on 4.5pt off 4.5pt}]  (190.34,98.17) .. controls (229.11,70.21) and (314.69,64.82) .. (368.21,92.78) ;
\draw [shift={(370.63,94.08)}, rotate = 208.88] [fill={rgb, 255:red, 0; green, 0; blue, 0 }  ][line width=0.08]  [draw opacity=0] (10.72,-5.15) -- (0,0) -- (10.72,5.15) -- (7.12,0) -- cycle    ;
\draw  [dash pattern={on 4.5pt off 4.5pt}]  (451.89,89.99) .. controls (496.91,62.27) and (513.6,138.79) .. (454.63,125.2) ;
\draw [shift={(451.89,124.51)}, rotate = 14.96] [fill={rgb, 255:red, 0; green, 0; blue, 0 }  ][line width=0.08]  [draw opacity=0] (10.72,-5.15) -- (0,0) -- (10.72,5.15) -- (7.12,0) -- cycle    ;
\draw  [dash pattern={on 4.5pt off 4.5pt}]  (352.85,689.52) .. controls (428.27,632.29) and (397.92,444.45) .. (467.52,413.65) ;
\draw [shift={(469.66,412.77)}, rotate = 158.93] [fill={rgb, 255:red, 0; green, 0; blue, 0 }  ][line width=0.08]  [draw opacity=0] (10.72,-5.15) -- (0,0) -- (10.72,5.15) -- (7.12,0) -- cycle    ;
\draw  [dash pattern={on 4.5pt off 4.5pt}]  (352.85,533.76) .. controls (395.59,475.01) and (397.26,300.94) .. (467.51,270.42) ;
\draw [shift={(469.66,269.54)}, rotate = 158.93] [fill={rgb, 255:red, 0; green, 0; blue, 0 }  ][line width=0.08]  [draw opacity=0] (10.72,-5.15) -- (0,0) -- (10.72,5.15) -- (7.12,0) -- cycle    ;
\draw  [dash pattern={on 4.5pt off 4.5pt}]  (266.52,408.17) .. controls (244.01,370.88) and (222.73,343.02) .. (179.62,301.86) ;
\draw [shift={(177.64,299.97)}, rotate = 43.52] [fill={rgb, 255:red, 0; green, 0; blue, 0 }  ][line width=0.08]  [draw opacity=0] (10.72,-5.15) -- (0,0) -- (10.72,5.15) -- (7.12,0) -- cycle    ;
\draw  [dash pattern={on 4.5pt off 4.5pt}]  (177.64,236.03) .. controls (204.6,213.82) and (255.2,217.8) .. (278.1,233.05) ;
\draw [shift={(280.48,234.75)}, rotate = 217.6] [fill={rgb, 255:red, 0; green, 0; blue, 0 }  ][line width=0.08]  [draw opacity=0] (10.72,-5.15) -- (0,0) -- (10.72,5.15) -- (7.12,0) -- cycle    ;

\draw (131.44,109.93) node [anchor=north west][inner sep=0.75pt]   [align=left] {$\mathsf{even}$};
\draw (398.39,103.62) node [anchor=north west][inner sep=0.75pt]   [align=left] {$\mathsf{odd}$};
\draw (507.83,261.4) node [anchor=north west][inner sep=0.75pt]   [align=left] {$b$};
\draw (304.96,261.64) node [anchor=north west][inner sep=0.75pt]   [align=left] {$a$};
\draw (298.09,403.03) node [anchor=north west][inner sep=0.75pt]   [align=left] {$aaa$};
\draw (501.24,403.63) node [anchor=north west][inner sep=0.75pt]   [align=left] {$aba$};
\draw (491.97,545.58) node [anchor=north west][inner sep=0.75pt]   [align=left] {$aabaa$};
\draw (282.96,680.38) node [anchor=north west][inner sep=0.75pt]   [align=left] {$abaaaba$};
\draw (290.82,540.99) node [anchor=north west][inner sep=0.75pt]   [align=left] {$baaab$};
\draw (139.44,262.68) node [anchor=north west][inner sep=0.75pt]   [align=left] {$aa$};
\draw (130.74,181.63) node [anchor=north west][inner sep=0.75pt]   [align=left] {$a$};
\draw (346.58,175.24) node [anchor=north west][inner sep=0.75pt]   [align=left] {$a$};
\draw (295.8,324.86) node [anchor=north west][inner sep=0.75pt]   [align=left] {$a$};
\draw (295.8,606.21) node [anchor=north west][inner sep=0.75pt]   [align=left] {$a$};
\draw (498.94,327.42) node [anchor=north west][inner sep=0.75pt]   [align=left] {$a$};
\draw (462.12,172.16) node [anchor=north west][inner sep=0.75pt]   [align=left] {$b$};
\draw (498.94,471.93) node [anchor=north west][inner sep=0.75pt]   [align=left] {$a$};
\draw (318.65,480.36) node [anchor=north west][inner sep=0.75pt]   [align=left] {$b$};
\draw (560.83,257.4) node [anchor=north west][inner sep=0.75pt]   [align=left] {$\{2, 6\}$};
\draw (561.83,400.4) node [anchor=north west][inner sep=0.75pt]   [align=left] {$\{1, 5\}$};
\draw (561.83,543.4) node [anchor=north west][inner sep=0.75pt]   [align=left] {$\{\textcolor[rgb]{0.29,0.56,0.89}{4}\}$};
\draw (351.83,237.4) node [anchor=north west][inner sep=0.75pt]   [align=left] {$\{1, 3, 4, 5, 7, \textcolor[rgb]{0.29,0.56,0.89}{8}\}$};
\draw (356.83,385.4) node [anchor=north west][inner sep=0.75pt]   [align=left] {$\{3\}$};
\draw (241.83,537.4) node [anchor=north west][inner sep=0.75pt]   [align=left] {$\{2\}$};
\draw (241.83,674.4) node [anchor=north west][inner sep=0.75pt]   [align=left] {$\{1\}$};
\draw (54.83,233.4) node [anchor=north west][inner sep=0.75pt]   [align=left] {$\{3, 4, \textcolor[rgb]{0.29,0.56,0.89}{7}\}$};

\end{tikzpicture}
\caption{The eertree after performing $\texttt{push\_back}\rbra{a}$ on the string $s = abaaaba$.
For the definition and notations, please refer to Section \ref{sec:eertree}.
Newly created nodes and transitions are colored in {\textcolor[rgb]{0.29,0.56,0.89}{blue}}. 
A solid arrow from node $v$ with character $c$ means the transition $\next\rbra{v, c}$. 
A dashed arrow from node $v$ means the suffix link.
The set near to node $v$ means $\occur\rbra*{s, v}$.}
\label{fig:push_back}
\end{figure}

\begin{figure}
    \centering
\tikzset{every picture/.style={line width=0.75pt}} 

\begin{tikzpicture}[x=0.75pt,y=0.75pt,yscale=-1,xscale=1]

\draw   (104,116.07) .. controls (104,92.06) and (123.33,72.59) .. (147.17,72.59) .. controls (171.01,72.59) and (190.34,92.06) .. (190.34,116.07) .. controls (190.34,140.09) and (171.01,159.55) .. (147.17,159.55) .. controls (123.33,159.55) and (104,140.09) .. (104,116.07) -- cycle ;
\draw   (368.09,111.48) .. controls (368.09,87.47) and (387.42,68) .. (411.26,68) .. controls (435.1,68) and (454.43,87.47) .. (454.43,111.48) .. controls (454.43,135.49) and (435.1,154.96) .. (411.26,154.96) .. controls (387.42,154.96) and (368.09,135.49) .. (368.09,111.48) -- cycle ;
\draw   (469.66,269.54) .. controls (469.66,245.52) and (488.99,226.05) .. (512.83,226.05) .. controls (536.67,226.05) and (556,245.52) .. (556,269.54) .. controls (556,293.55) and (536.67,313.02) .. (512.83,313.02) .. controls (488.99,313.02) and (469.66,293.55) .. (469.66,269.54) -- cycle ;
\draw   (266.52,267.5) .. controls (266.52,243.49) and (285.84,224.02) .. (309.69,224.02) .. controls (333.53,224.02) and (352.85,243.49) .. (352.85,267.5) .. controls (352.85,291.51) and (333.53,310.98) .. (309.69,310.98) .. controls (285.84,310.98) and (266.52,291.51) .. (266.52,267.5) -- cycle ;
\draw   (266.52,408.17) .. controls (266.52,384.16) and (285.84,364.69) .. (309.69,364.69) .. controls (333.53,364.69) and (352.85,384.16) .. (352.85,408.17) .. controls (352.85,432.19) and (333.53,451.65) .. (309.69,451.65) .. controls (285.84,451.65) and (266.52,432.19) .. (266.52,408.17) -- cycle ;
\draw   (469.66,412.77) .. controls (469.66,388.75) and (488.99,369.29) .. (512.83,369.29) .. controls (536.67,369.29) and (556,388.75) .. (556,412.77) .. controls (556,436.78) and (536.67,456.25) .. (512.83,456.25) .. controls (488.99,456.25) and (469.66,436.78) .. (469.66,412.77) -- cycle ;
\draw  [fill={rgb, 255:red, 126; green, 211; blue, 33 }  ,fill opacity=1 ] (266.52,689.52) .. controls (266.52,665.51) and (285.84,646.04) .. (309.69,646.04) .. controls (333.53,646.04) and (352.85,665.51) .. (352.85,689.52) .. controls (352.85,713.53) and (333.53,733) .. (309.69,733) .. controls (285.84,733) and (266.52,713.53) .. (266.52,689.52) -- cycle ;
\draw   (266.52,548.85) .. controls (266.52,524.83) and (285.84,505.37) .. (309.69,505.37) .. controls (333.53,505.37) and (352.85,524.83) .. (352.85,548.85) .. controls (352.85,572.86) and (333.53,592.33) .. (309.69,592.33) .. controls (285.84,592.33) and (266.52,572.86) .. (266.52,548.85) -- cycle ;
\draw   (104,269.54) .. controls (104,245.52) and (123.33,226.05) .. (147.17,226.05) .. controls (171.01,226.05) and (190.34,245.52) .. (190.34,269.54) .. controls (190.34,293.55) and (171.01,313.02) .. (147.17,313.02) .. controls (123.33,313.02) and (104,293.55) .. (104,269.54) -- cycle ;
\draw    (411.26,154.96) -- (311.34,222.89) ;
\draw [shift={(309.69,224.02)}, rotate = 325.79] [fill={rgb, 255:red, 0; green, 0; blue, 0 }  ][line width=0.08]  [draw opacity=0] (12,-3) -- (0,0) -- (12,3) -- cycle    ;
\draw    (411.26,154.96) -- (506.14,224.87) ;
\draw [shift={(507.75,226.05)}, rotate = 216.38] [fill={rgb, 255:red, 0; green, 0; blue, 0 }  ][line width=0.08]  [draw opacity=0] (12,-3) -- (0,0) -- (12,3) -- cycle    ;
\draw    (309.69,310.98) -- (309.69,362.69) ;
\draw [shift={(309.69,364.69)}, rotate = 270] [fill={rgb, 255:red, 0; green, 0; blue, 0 }  ][line width=0.08]  [draw opacity=0] (12,-3) -- (0,0) -- (12,3) -- cycle    ;
\draw    (512.83,313.02) -- (512.83,367.29) ;
\draw [shift={(512.83,369.29)}, rotate = 270] [fill={rgb, 255:red, 0; green, 0; blue, 0 }  ][line width=0.08]  [draw opacity=0] (12,-3) -- (0,0) -- (12,3) -- cycle    ;
\draw    (309.69,451.65) -- (309.69,503.37) ;
\draw [shift={(309.69,505.37)}, rotate = 270] [fill={rgb, 255:red, 0; green, 0; blue, 0 }  ][line width=0.08]  [draw opacity=0] (12,-3) -- (0,0) -- (12,3) -- cycle    ;
\draw [color={rgb, 255:red, 126; green, 211; blue, 33 }  ,draw opacity=1 ]   (309.69,592.33) -- (309.69,644.04) ;
\draw [shift={(309.69,646.04)}, rotate = 270] [fill={rgb, 255:red, 126; green, 211; blue, 33 }  ,fill opacity=1 ][line width=0.08]  [draw opacity=0] (12,-3) -- (0,0) -- (12,3) -- cycle    ;
\draw    (147.17,159.55) -- (147.17,224.05) ;
\draw [shift={(147.17,226.05)}, rotate = 270] [fill={rgb, 255:red, 0; green, 0; blue, 0 }  ][line width=0.08]  [draw opacity=0] (12,-3) -- (0,0) -- (12,3) -- cycle    ;
\draw  [dash pattern={on 4.5pt off 4.5pt}]  (472.2,395.89) .. controls (418.43,372.96) and (370.81,348.56) .. (341.48,300.41) ;
\draw [shift={(340.16,298.19)}, rotate = 59.52] [fill={rgb, 255:red, 0; green, 0; blue, 0 }  ][line width=0.08]  [draw opacity=0] (10.72,-5.15) -- (0,0) -- (10.72,5.15) -- (7.12,0) -- cycle    ;
\draw  [dash pattern={on 4.5pt off 4.5pt}]  (293.18,227.08) .. controls (281.92,154.02) and (243.57,114.46) .. (192.67,115.98) ;
\draw [shift={(190.34,116.07)}, rotate = 356.91] [fill={rgb, 255:red, 0; green, 0; blue, 0 }  ][line width=0.08]  [draw opacity=0] (10.72,-5.15) -- (0,0) -- (10.72,5.15) -- (7.12,0) -- cycle    ;
\draw  [dash pattern={on 4.5pt off 4.5pt}]  (477.28,243.7) .. controls (435.8,190.53) and (286.06,112.29) .. (193.14,115.94) ;
\draw [shift={(190.34,116.07)}, rotate = 356.68] [fill={rgb, 255:red, 0; green, 0; blue, 0 }  ][line width=0.08]  [draw opacity=0] (10.72,-5.15) -- (0,0) -- (10.72,5.15) -- (7.12,0) -- cycle    ;
\draw  [dash pattern={on 4.5pt off 4.5pt}]  (190.34,98.17) .. controls (229.11,70.21) and (314.69,64.82) .. (368.21,92.78) ;
\draw [shift={(370.63,94.08)}, rotate = 208.88] [fill={rgb, 255:red, 0; green, 0; blue, 0 }  ][line width=0.08]  [draw opacity=0] (10.72,-5.15) -- (0,0) -- (10.72,5.15) -- (7.12,0) -- cycle    ;
\draw  [dash pattern={on 4.5pt off 4.5pt}]  (451.89,89.99) .. controls (496.91,62.27) and (513.6,138.79) .. (454.63,125.2) ;
\draw [shift={(451.89,124.51)}, rotate = 14.96] [fill={rgb, 255:red, 0; green, 0; blue, 0 }  ][line width=0.08]  [draw opacity=0] (10.72,-5.15) -- (0,0) -- (10.72,5.15) -- (7.12,0) -- cycle    ;
\draw [color={rgb, 255:red, 126; green, 211; blue, 33 }  ,draw opacity=1 ] [dash pattern={on 4.5pt off 4.5pt}]  (352.85,689.52) .. controls (428.27,632.29) and (397.92,444.45) .. (467.52,413.65) ;
\draw [shift={(469.66,412.77)}, rotate = 158.93] [fill={rgb, 255:red, 126; green, 211; blue, 33 }  ,fill opacity=1 ][line width=0.08]  [draw opacity=0] (10.72,-5.15) -- (0,0) -- (10.72,5.15) -- (7.12,0) -- cycle    ;
\draw  [dash pattern={on 4.5pt off 4.5pt}]  (352.85,533.76) .. controls (395.59,475.01) and (397.26,300.94) .. (467.51,270.42) ;
\draw [shift={(469.66,269.54)}, rotate = 158.93] [fill={rgb, 255:red, 0; green, 0; blue, 0 }  ][line width=0.08]  [draw opacity=0] (10.72,-5.15) -- (0,0) -- (10.72,5.15) -- (7.12,0) -- cycle    ;
\draw  [dash pattern={on 4.5pt off 4.5pt}]  (266.52,408.17) .. controls (244.01,370.88) and (222.73,343.02) .. (179.62,301.86) ;
\draw [shift={(177.64,299.97)}, rotate = 43.52] [fill={rgb, 255:red, 0; green, 0; blue, 0 }  ][line width=0.08]  [draw opacity=0] (10.72,-5.15) -- (0,0) -- (10.72,5.15) -- (7.12,0) -- cycle    ;
\draw  [dash pattern={on 4.5pt off 4.5pt}]  (177.64,236.03) .. controls (204.6,213.82) and (255.2,217.8) .. (278.1,233.05) ;
\draw [shift={(280.48,234.75)}, rotate = 217.6] [fill={rgb, 255:red, 0; green, 0; blue, 0 }  ][line width=0.08]  [draw opacity=0] (10.72,-5.15) -- (0,0) -- (10.72,5.15) -- (7.12,0) -- cycle    ;

\draw (131.44,109.93) node [anchor=north west][inner sep=0.75pt]   [align=left] {$\mathsf{even}$};
\draw (398.39,103.62) node [anchor=north west][inner sep=0.75pt]   [align=left] {$\mathsf{odd}$};
\draw (507.83,261.4) node [anchor=north west][inner sep=0.75pt]   [align=left] {$b$};
\draw (304.96,261.64) node [anchor=north west][inner sep=0.75pt]   [align=left] {$a$};
\draw (298.09,403.03) node [anchor=north west][inner sep=0.75pt]   [align=left] {$aaa$};
\draw (501.24,403.63) node [anchor=north west][inner sep=0.75pt]   [align=left] {$aba$};
\draw (282.96,680.38) node [anchor=north west][inner sep=0.75pt]   [align=left] {$abaaaba$};
\draw (290.82,540.99) node [anchor=north west][inner sep=0.75pt]   [align=left] {$baaab$};
\draw (139.44,262.68) node [anchor=north west][inner sep=0.75pt]   [align=left] {$aa$};
\draw (130.74,181.63) node [anchor=north west][inner sep=0.75pt]   [align=left] {$a$};
\draw (346.58,175.24) node [anchor=north west][inner sep=0.75pt]   [align=left] {$a$};
\draw (295.8,324.86) node [anchor=north west][inner sep=0.75pt]   [align=left] {$a$};
\draw (295.8,606.21) node [anchor=north west][inner sep=0.75pt]   [align=left] {$a$};
\draw (498.94,327.42) node [anchor=north west][inner sep=0.75pt]   [align=left] {$a$};
\draw (462.12,172.16) node [anchor=north west][inner sep=0.75pt]   [align=left] {$b$};
\draw (313.65,463.36) node [anchor=north west][inner sep=0.75pt]   [align=left] {$b$};
\draw (560.83,257.4) node [anchor=north west][inner sep=0.75pt]   [align=left] {$\{2, 6\}$};
\draw (561.83,400.4) node [anchor=north west][inner sep=0.75pt]   [align=left] {$\{1, \textcolor[rgb]{0.49,0.83,0.13}{5}\}$};
\draw (351.83,237.4) node [anchor=north west][inner sep=0.75pt]   [align=left] {$\{1, 3, 4, 5, \textcolor[rgb]{0.49,0.83,0.13}{7}\}$};
\draw (356.83,385.4) node [anchor=north west][inner sep=0.75pt]   [align=left] {$\{3\}$};
\draw (241.83,537.4) node [anchor=north west][inner sep=0.75pt]   [align=left] {$\{2\}$};
\draw (241.83,674.4) node [anchor=north west][inner sep=0.75pt]   [align=left] {$\{\textcolor[rgb]{0.49,0.83,0.13}{1}\}$};
\draw (67.83,234.4) node [anchor=north west][inner sep=0.75pt]   [align=left] {$\{3, 4\}$};

\end{tikzpicture}
\caption{The eertree after performing $\texttt{pop\_back}\rbra{}$ on the string $s = abaaaba$.
Newly deleted nodes and transitions are colored in {\textcolor[rgb]{0.49,0.83,0.13}{green}}. 
See Figure \ref{fig:push_back} for the meaning of the arrows. }
\label{fig:pop_back}
\end{figure}

    The data structure adopted in our implementation is an extended eertree. This indicates that we always have access to the eertree of the current string, and therefore the basic operations of eertrees are naturally inherited. 
    Roughly speaking, we propose a self-organizing data structure double-ended eertree which supports online deque operations on the stored string in linear time.
    We formally state our efficient algorithm for the double-ended eertree in the following theorem. 
    
    \begin{theorem} [Double-ended eertrees] \label{thm:main}
        Double-ended eertree can be implemented
        with worst-case time and space complexity $O\rbra*{\log\rbra*{\sigma}}$ per operation, where $\sigma$ is the size of the alphabet.
    \end{theorem}
    
    The double-ended eertree is assumed to work in the word RAM model \cite{FW90} under the constant cost criterion (see Section \ref{sec:word-ram} for the formal definition), which considers any operations from the C programming language as constant time, thereby a more practical computational model than RAM \cite{CR73}. The word RAM model is arguably the most widely used computational model for practical algorithms (cf. \cite{Hag98}). In practice, the size $\sigma$ of the alphabet is usually a constant, e.g., $\sigma = 2$ for binary strings, $\sigma = 4$ for DNA sequences, and $\sigma = 26$ for English dictionaries; in this case, the implementation of double-ended eertrees in Theorem \ref{thm:main} achieves worst-case time and space complexity $O\rbra*{1}$ per operation.
    We propose two different methods to implement double-ended eertrees (see Section \ref{sec:tech} for further discussions), with their practical and efficient C/C++ implementations provided in \cite{Wan22}. 
    
    For comparison, we collect known implementations of different types of eertrees in Table \ref{tab:eertrees}.
    If only basic operations \texttt{push\_back} and \texttt{pop\_back} are allowed, the double-ended eertree is called a stack eertree; and if only basic operations \texttt{push\_back} and \texttt{pop\_front} are allowed, the double-ended eertree is called a queue eertree. See Section \ref{sec:other-possible-implementation} for further comparisons and discussions.
    
    \begin{table}[!htp]
    \centering
    \caption{Different types of eertrees.}
    \label{tab:eertrees}
    \begin{tabular}{cc}
    \toprule
    Eertree Type         & Time Complexity Per Operation \\ \midrule
    Stack Eertree \cite{RS18}        & $O\rbra*{\log\rbra*{\sigma}}$                             \\ 
    Queue Eertree \cite{MWN+22}       & $O\rbra*{\log\rbra*{\sigma}}$                             \\ 
    Double-Ended Eertree (This Paper) & $O\rbra*{\log\rbra*{\sigma}}$                             \\ \bottomrule
    \end{tabular}
    \end{table}
    
    \subsection{Our techniques} \label{sec:tech}
    
    We propose a framework of double-ended eertrees (see Section \ref{sec:framework}).
    The difficulty to implement double-ended eertrees lies in two parts, namely, checking the uniqueness of palindromic substrings and maintaining the longest palindromic prefix and suffix. 
    Under this framework, we propose 
    a method to implement double-ended eertrees, called the surface recording method.
    This method is 
    based on a new concept called surface. 
    
    Surfaces are essential palindromic substring occurrences that capture the palindromic structure of the whole string, thereby yielding a linear-time algorithm to maintain double-ended eertrees. To the best of our knowledge, surfaces have not been noticed in the literature (see Section \ref{sec:intro-surface-recording} for details and Section \ref{sec:discussion-surfaces} for more discussions). 

    In Section \ref{sec:intro-occurrence-recording}, we introduce the concept of reduced set of occurrences. Maintaining such sets for all palindromic substrings, one can efficiently decide the uniqueness of a palindrome and find the longest prefix- and suffix-palindromes. Then, in Section \ref{sec:intro-surface-recording}, we describe the unique minimal reduced sets, their relation to surfaces, and their efficient update on deque operations.

    \subsubsection{Reduced sets of occurrences} \label{sec:intro-occurrence-recording}
    
    To overcome the difficulties mentioned above, we first propose a suboptimal algorithm to implement double-ended eertrees with time complexity per operation $O\rbra*{\log\rbra*{\sigma} + \log\rbra*{n}}$, where $n$ is the length of the current string, based on occurrence recording.
    
    For each palindromic substring $t$ of the current string $s$, let $v = \node \rbra*{t}$ denote the node in the eertree that represents string $t$, and also write $\str\rbra*{v} = t$ and $\len\rbra*{v} = \abs*{t}$.
    Let $\occur\rbra*{s, v}$ be the set of all occurrences of $\str\rbra*{v}$ in $s$. 
    If we have access to $\occur\rbra*{s, v}$ for every node $v$, then the uniqueness of palindromic substrings can be easily checked and the longest palindromic suffix and prefix can be maintained. 
    However, the total size of all $\occur\rbra*{s, v}$, i.e., the number of all occurrences over all palindromic substrings of $s$, can be $\Omega\rbra{n^2}$. 
    To reduce the computation, an observation is that an occurrence $i$ of node $v$ implies two occurrences of node $\link\rbra*{v}$, where $\link\rbra*{v}$ is the node of the longest palindromic suffix of $\str\rbra*{v}$: one is $i$ itself, and the other is $i + \len\rbra*{v} - \len\rbra*{\link\rbra*{v}}$ induced by $v$. 
    To this end, our idea is to store any reduced (sub)set $S\rbra*{s, v}$ of $\occur\rbra*{s, v}$ such that 
    \[
    \occur\rbra*{s,v} = S\rbra*{s,v} \cup \bigcup_{\link\rbra*{u} = v} \rbra*{\occur\rbra*{s,u} \cup \overline{\occur}\rbra*{s,u}},
    \]
    where $\overline{\occur}\rbra*{s,v}$ denotes the set of the occurrences induced by $\occur\rbra*{s, v}$. 
    With the help of $S\rbra*{s, v}$, we are able to check the uniqueness of a palindromic substring directly (see Lemma \ref{lemma:unique-by-sv}), as well as maintain the longest palindromic prefix and suffix of $s$ indirectly, under deque operations.
    This part of the method can be seen as a generalization of occurrence recording in queue eertrees \cite{MWN+22}, wherein the last two occurrences are recorded. By contrast, our method records a reduced set of occurrences. 
    
    Specifically, we store two sets $\prenode\rbra*{s, i}$ and $\sufnode\rbra*{s, i}$ related to the reduced sets $S\rbra*{s, v}$ for every $1 \leq i \leq \abs*{s}$.
    Roughly speaking, $\prenode\rbra*{s, i}$ (resp. $\sufnode\rbra*{s, i}$) collects all nodes $v$ with reduced occurrence $i$ (resp. end position $i$), i.e., $i \in S\rbra*{s, v}$ (resp. $i - \len\rbra*{v} + 1 \in S\rbra*{s, v}$).
    It is worth noting that the longest palindromic prefix (resp. suffix) of $s$ is indeed the palindrome represented by the node of the longest length in $\prenode\rbra*{s, 1}$ (resp. $\sufnode\rbra*{s, \abs*{s}}$) (see Lemma  \ref{lemma:prepal-by-prenode}).
    Since the two series of sets $\prenode\rbra*{s, i}$ and $\sufnode\rbra*{s, i}$ can be maintained by balanced binary search trees as the reduced sets $S\rbra*{s, v}$ change, the key is to find any suitable choice of reduced sets $S\rbra*{s, v}$ that can be maintained efficiently. 
    
    A subtle observation yields that a certain construction of the reduced sets $S\rbra*{s, v}$ can be maintained by modifying only amortized $O\rbra*{1}$ elements for each deque operation. 
    With these, we can design an online algorithm to implement double-ended eertrees with time complexity $O\rbra*{\log\rbra*{\sigma} + \log\rbra*{n}}$ per operation, where the term $\log\rbra*{n}$ comes from operations of balanced binary search trees. 
    This running time can be further improved to $O(\log \sigma)$ by an appropriate choice of reduced sets, described in Section \ref{sec:intro-surface-recording}. Accordingly, we put the details of the suboptimal method into Appendix \ref{sec:occurrence-recording-app} for readability.
    
    \subsubsection{Surface recording} \label{sec:intro-surface-recording}
    
    In the occurrence recording method in Section \ref{sec:intro-occurrence-recording}, we find a construction of the reduced sets $S\rbra*{s, v}$ with only $O\rbra*{1}$ amortized modifications per operation, but still need extra $O\rbra*{\log\rbra*{n}}$ binary tree operations to maintain them. Digging further into the reduced sets, we find that they have a unique minimal choice (see Lemma \ref{lemma:relation-sv-surface}). Inspired by this, we identify a class of substring occurrences, called surfaces (see Definition \ref{def:surface}), that can characterize how palindromic substrings are distributed in the string. Surprisingly, the unique minimal choice of the reduced sets consists exclusively of all surfaces.
    
    Roughly speaking, a surface in a string $s$ is a palindromic substring occurrence of $s$ that is not a prefix or suffix of any other palindromic substring occurrences of $s$. The name ``surface'' is chosen to denote not being ``covered'' by any other entities. Intuitively, long palindromes derive shorter ones, but a surface is such a palindrome that is not implied by any other palindromes. For example, in the string $s = abacaba$,  $s\substr{1}{7}$ is a surface but $s\substr{1}{3}$ is not because $s\substr{1}{3}$ is a palindromic proper prefix of, therefore ``covered'' by, $s\substr{1}{7}$. In other words, surface $s\substr{1}{7}$ naturally implies shorter palindromes  $s\substr{1}{3}$ and $s\substr{5}{7}$. 
    See Section \ref{sec:discussion-surfaces} for more discussions about surfaces.
    
    The key of the surface recording method (see Section \ref{sec:surface-recording}) is to maintain all surfaces in the string 
    in an implicit manner. To achieve this, we store two lists $\presurf\rbra*{s, i}$ and $\sufsurf\rbra*{s, i}$
    similar to the sets $\prenode\rbra*{s, i}$ and $\sufnode\rbra*{s, i}$ used in the occurrence recording method. By contrast, $\presurf\rbra*{s, i}$ and $\sufsurf\rbra*{s, i}$ only store (the pointer to the node of) a palindrome. 
    Specifically, $\presurf\rbra*{s, i}$ (resp. $\sufsurf\rbra*{s, i}$) indicates the surface with occurrence $i$ (resp. with end position $i$); or the empty string $\epsilon$, if no such surface exists. 
    It is shown in Observation \ref{lemma:prepal-by-presurf} that the longest palindromic prefix (resp. suffix) of $s$ is the leftmost surface with start position $1$ (resp. the rightmost surface with end position $\abs*{s}$), which is actually $\presurf\rbra*{s, 1}$ (resp. $\sufsurf\rbra*{s, \abs*{s}}$). 
    Therefore, the longest palindromic prefix and suffix are naturally obtained if we maintain $\presurf\rbra*{s, i}$ and $\sufsurf\rbra*{s, i}$. Looking into the relationship between surfaces, we find a way to maintain $\presurf\rbra*{s, i}$ and $\sufsurf\rbra*{s, i}$ by modifying only $O\rbra*{1}$ elements with time complexity $O\rbra*{1}$ per deque operation 
    (see Section \ref{sec:surface-recording-algorithm}).
    
    It remains to check the uniqueness of palindromic substrings. To achieve this, we find an efficient algorithm to maintain the number of occurrences of each existing palindromic substrings of the current string.
    We define $\precnt\rbra*{s, v}$ (resp. $\sufcnt\rbra*{s, v}$) to denote the number of occurrences of node $v$ in string $s$ that are not a prefix (resp. suffix) of a longer palindromic substring. Surprisingly, we find that $\precnt\rbra*{s, v} = \sufcnt\rbra*{s, v}$ always holds (see Lemma \ref{lemma:precnt=sufcnt}), thereby letting $\cnt\rbra*{s, v}$ denote either of them; 
    also, the number of occurrences of every node $v$, i.e., the cardinality of $\occur\rbra*{s, v}$, is the sum of $\cnt\rbra*{s, u}$ over all nodes $u$ in the link tree rooted at node $v$ (see Lemma \ref{lemma:unique-by-cnt}). 
    It is clear that uniqueness (i.e., only one occurrence) can be checked because the number of occurrences can be computed through $\cnt\rbra*{s, v}$. 
    At last, a simple update rule of $\cnt\rbra*{s, v}$ is established for each deque operation with time complexity $O\rbra*{1}$ per operation (see Lemma \ref{lemma:cnt-update}). 
    
    \subsection{Related works}
    
    Approximate palindromes are strings close to a palindrome with gaps or mismatches, which were investigated in a series of works \cite{Gus97,Por99,PB02,HCC10,CHC12,AP14}.
    Formulas for the expected number of gapped palindromes in a random string were presented in \cite{DNP17}.
    As in information processing, the study of palindromes aims to find efficient algorithms for strings concerning their palindromic structures. Trie \cite{DLB59,Fre60} is one of the earliest data structures that can store and look up words in a dictionary, which is widely used in string-searching problems with its extensions developed, including Aho-Corasick automata \cite{AC75,Mey85} and suffix trees \cite{Wei73,McC76,Ukk95,Far97}. Suffix arrays \cite{MM93,GBS92} were introduced to improve the space complexity of suffix trees, and were later improved to linear-time \cite{KSB06,NZC09}. Suffix trees with deletions on one end and insertions on the other end, namely, suffix trees that support queue operations or also called suffix trees in a sliding window, were studied in a series of works \cite{FG89,Lar96,IS11,MHK+20,MKA+20}. 
    
    \subsection{Discussion}
    
    \subsubsection{Surfaces} \label{sec:discussion-surfaces}
    
    A similar notion to surface is the maximal palindromic substring, which has been extensively studied in the literature  \cite{Man75,GS78,Jeu94,ABG95,IIBT10,IIT13}. 
    A substring occurrence $s\substr{i}{j}$ is said to be maximal palindromic, if it is the longest palindromic substring occurrence with center $\frac{i+j}{2}$.
    By contrast, a surface is a palindromic substring occurrence that is not a prefix or suffix of any other palindromic substring occurrences. 
    It can be seen that a maximal palindromic substring occurrence is not necessarily a surface, and vice versa:
    (i) for $s = abaaaba$, $s\substr{1}{3} = s\substr{5}{7} = aba$ are maximal palindromic substring occurrences but not surfaces, and (ii) for $s = abcba$, $s\substr{2}{4} = bcb$ is a surface but not a maximal palindromic substring occurrence.
    From the perspective of eertrees, maximal palindromic substring occurrences characterize the trie-like structure of eertrees while surfaces characterize the link tree of eertrees. 
    
    Another similar notion to surface is the border-maximal palindrome \cite{MWN+22}. A border-maximal palindrome is a palindrome that is not a proper suffix of other palindromic substrings. The difference between surfaces and border-maximal palindromes is that the former is index sensitive but the latter is not. For example, $bab$ is not a border-maximal palindrome in $s = aababbaababab$ because it is a proper suffix of $s\substr{9}{13} = babab$. By contrast, $s\substr{3}{5} = bab$ is a surface in $s$ because it is not contained in any other palindromic substrings as their prefix or suffix. Intuitively, border-maximal palindromes characterize global properties of a string while surfaces catch local properties.
    
    As shown in this paper, we find that surfaces are quite useful in palindrome related problems. 
    We are not aware that the concept of surface has been defined and used elsewhere but we believe that it will bring new insights into string processing. 
    
    \subsubsection{Other possible implementations} \label{sec:other-possible-implementation}
    
    It has been shown in \cite{RS18} that eertree is a very efficient data structure to process palindromes. 
    Theoretically, one might doubt whether it is possible to implement basic operations of double-ended eertrees in other ways, with or without the notion of eertrees.
    
    \paragraph{By suffix trees.}
    Intuitively, one might wonder whether it is possible to adapt suffix tree tricks (cf. \cite{Gus97}) to implementing double-ended eertrees. A common way concerning palindromic substrings of a string $s$, for example, is to build a suffix tree of $s\$s^R$, where $s^R$ is the reverse of $s$ and $\$$ is any character that does not appear in $s$. Then, palindromic substrings of $s$ can be found using longest common extension queries on $s$ and $s^R$. This kind of trick was already used in finding the longest palindrome \cite{Gus97}, counting palindromes \cite{GPR10}, and finding distinct palindromes online \cite{KRS13}. However, in our case, after a $\texttt{push\_back}(c)$ (resp. $\texttt{push\_front}(c)$) operation, the suffix tree of $sc\$cs^R$ (resp. $cs\$s^Rc$) is required; also, after a $\texttt{pop\_back}$ (resp. $\texttt{pop\_front}$) operation, the suffix tree of $s'\$s'^R$ is required, where $s' = s\substr{1}{\abs*{s}-1}$ (resp. $s' = s\substr{2}{\abs*{s}}$). 
    As can be seen, to support deque operations together with queries concerning palindromes on string $s$, we need complicated modifications on suffix trees, such as inserting and removing characters at both ends and certain intermediate positions of a string; however, the state-of-the-art dynamic suffix tree \cite{KK22} only supports these kinds of operations in $\polylog\rbra{n}$ time per operation. 
    Concerning these, it could be difficult to implement a linear-time double-ended eertree directly using suffix trees. 
    
    \paragraph{By other variants of eertrees.}
    
    Since the double-ended eertree can be considered as an extension of stack eertree \cite{RS18} or queue eertree \cite{MWN+22}, it might be possible to implement the double-ended eertree using the tricks in stack eertree and queue eertree at the first glance. Next, we discuss about this issue. 
    \begin{itemize}
        \item The stack eertree proposed in \cite{RS18} depends on the structure of stack. That is, a character will not change before all characters after it are removed. Based on this fact, a backup strategy is to store necessary information, e.g., the node of the current longest suffix palindromic substring, when a character is inserted; and restore the previous configuration of the eertree through the backups when the last character is removed. In the scene of double-ended eertrees, characters at the very front of the string can apparently be removed or inserted, thereby making the backup strategy no longer effective. 
        \item The queue eertree proposed in \cite{MWN+22} requires auxiliary data structures for removing characters, where the key technique is to check whether a palindromic substring is unique in the current string. To achieve this, they maintain the rightmost occurrence and the second rightmost occurrence of every node in the eertree in a lazy manner. When a character is inserted at the back, the occurrences of only the longest palindromic suffix are updated; when a character is removed from the front, we do not have to deal with the occurrences of the longest palindromic prefix, but its longest palindromic proper prefix (by lazy maintenance). This approach, however, is based on the monotonicity of queue operations. That is, the rightmost occurrence can become the second rightmost occurrence, but not vice versa. In the scene of double-ended eertrees, the lazy strategy will fail when characters are removed from the back, because this time the second rightmost occurrence can be the rightmost occurrence.
    \end{itemize}
    
    As discussed, we could not obtain a straightforward implementation of double-ended eertrees from similar or related data structures mentioned above. Nevertheless, we will be glad to see if there exist other efficient implementations of double-ended eertrees.
    
    \subsubsection{Potential applications}
    
    As mentioned in \cite{RS18}, eertrees could have potential in Watson-Crick palindromes \cite{KM10,MMP22} and RNA structures \cite{MP05,Str07,MKB+11,BA15}. In addition, gene editing is an emerging research field in biology (cf. \cite{BGOP18}), e.g. the clustered regularly interspaced short palindrome repeats (CRISPR)-Cas9 system (cf. \cite{BD16,LHCT18}). One might require real-time palindrome-related properties while editing DNA and RNA sequences. We believe that double-ended eertrees could have potential practical applications in such tasks.

    There are several simple applications of our double-ended eertree with range queries by adopting the trick in Mo's algorithm (cf. \cite{DKPW20}): 
    \begin{itemize}
        \item \textsc{Counting Distinct Palindromic Substrings} \cite{RS17}: Find the number of distinct palindromic substring of $s\substr{l}{r}$. 
        \item \textsc{Longest Palindromic Substring} \cite{ACPR20,MMSH23}: Find the longest palindromic substring of $s\substr{l}{r}$.
        \item \textsc{Shortest Unique Palindromic Substring}: Find the shortest unique palindromic substring of $s\substr{l}{r}$. 
        This problem is motivated by molecular biology \cite{KYK+92,YYK+92}. Algorithms about the shortest unique palindromic substring was investigated in a series of works \cite{INM+18,WNI+20,FM21,MF22}.
        \item \textsc{Shortest Absent Palindrome}: Find the shortest absent palindrome of $s\substr{l}{r}$. 
        This is the palindromic version \cite{CMRS00,MRS02,CC12} of finding absent words.
    \end{itemize}
    In addition, our double-ended eertree can be used to the following counting problem:
    \begin{itemize}
        \item \textsc{Counting Rich Strings with Given Word}: Given a string $t$ of length $n$ and a number $k$, count the number of palindromic rich strings $s$ of length $n + k$ such that $t$ is a substring of $s$. This generalizes the problem of counting palindromic rich strings studied in \cite{RS18}.
    \end{itemize}
    The details of these applications can be found in Appendix \ref{sec:app-app}. 
    
    \subsubsection{Future work} \label{sec:conclusion}
    
    In this paper, we propose a linear-time implementation for double-ended eertrees, and provide a practical and efficient implementation for double-ended eertrees. There are several aspects left for future research.
    \begin{itemize}
        \item The concept of surface is useful in our algorithms. It would be interesting to find how surfaces can be used in string processing.
        \item Another direction is to consider the dynamic maintenance of palindromes in a collection of strings under split and concatenation. A possible way could be to combine the dynamic data structures of \cite{AB19} and \cite{GKK+18}.
        \item Recently, palindromes in circles \cite{Sim14} and trees \cite{BLP15,GKRW15,GSS19,FNI+19} have been investigated. It would be interesting to generalize eertrees for palindromes in these special structures. 
        \item Quantum algorithms have been proposed for string problems such as pattern matching \cite{RV03,Mon17,KNW25}, edit distance \cite{BEG+21,GJKT24}, longest common substring \cite{LGS22,AJ23,JN24}, lexicographically minimal string rotation \cite{WY24,AJ23,CKKD+25,Wan25}, and longest distinct substring \cite{ABB+25}. 
        Among them, only a palindrome-related quantum algorithm was proposed in \cite{LGS22}. 
        An interesting direction is to find more quantum algorithms for palindromes.
    \end{itemize}

    \subsection{Recent developments}
    After the work of this paper, the double-ended palindromic tree became testable on the platform Library Checker hosted by Morita \cite{Mor20} in 2024. 
    After that, Kulkov \cite{Kul24} presented a different implementation of the double-ended palindromic tree. 
    
    \subsection{Organization of this paper}
    
    In the rest of this paper, we first introduce preliminaries in Section \ref{sec:preliminaries}. Then, we define a framework of implementing double-ended eertrees in Section \ref{sec:framework}. Under the framework, we propose an occurrence recording method to maintain double-ended eertrees in Section \ref{sec:occurrence-recording}, inspired by which we define the concept of surface in Section \ref{sec:surface}, and then propose a more efficient method called surface recording in Section \ref{sec:surface-recording}. 

    \section{Preliminaries} \label{sec:preliminaries}

    \subsection{Strings}

    Let $\Sigma$ be an alphabet of size $\sigma = \abs{\Sigma}$. A string $s$ of length $n$ over $\Sigma$ is an array $s\sbra*{1} s\sbra*{2} \dots s\sbra*{n}$, where $s\sbra*{i} \in \Sigma$ is the $i$-th character of $s$ for $1 \leq i \leq n$. We write $\abs{s}$ to denote the length of string $s$.
    Let $\Sigma^n$ denote the set of all strings of length $n$ over $\Sigma$, and $\Sigma^*$ denotes the set of all (finite) strings over $\Sigma$. Especially, $\epsilon$ denotes the empty string, i.e., $\abs*{\epsilon} = 0$.
    The concatenation of two strings $s$ and $t$ is denoted by the string $st = s\sbra*{1} s\sbra*{2} \dots s\sbra*{\abs{s}} t\sbra{1} t\sbra{2} \dots t\sbra{\abs{t}}$ of length $\abs{st} = \abs{s} + \abs{t}$.
    The substring occurrence
    $s\substr{i}{j}$ denotes the string $s\sbra*{i} s\sbra*{i+1} \dots s\sbra*{j}$ if $1 \leq i \leq j \leq \abs{s}$, and $\epsilon$ otherwise.
    A string $t$ is called a substring of $s$ if $t = s\substr{i}{j}$ for some $i$ and $j$.
    If a substring of $s$ starts with the first character of $s$, then it is a prefix of $s$; if it ends with the last character of $s$, then it is a suffix of $s$.
    Formally, a substring $t = s\substr{i}{j}$ of string $s$ is called a prefix of $s$ if $i = 1$, and is called a suffix of $s$ if $j = \abs*{s}$.
    In particular, the empty string $\epsilon$ is a substring, prefix and suffix of any strings. 
    A prefix (resp. suffix) $t$ of string $s$ is proper if $t \neq s$.
    A non-empty substring $t$ of $s$ is unique in $s$ if $t$ occurs only once in $s$, i.e., there is only one pair of indices $i$ and $j$ such that $1 \leq i \leq j \leq \abs{s}$ and $s\substr{i}{j} = t$.
    A string $s$ is palindromic (or a palindrome) if $s\sbra*{i} = s\sbra*{\abs*{s}-i+1}$ for all $1 \leq i \leq \abs*{s}$. In particular, the empty string $\epsilon$ is a palindrome.
    The center of a palindromic substring $s\substr{i}{j}$ of $s$ is $\rbra*{i+j}/2$. Then, the center of $s\substr{i}{j}$ is an integer if its length is odd, and the center of $s\substr{i}{j}$ is a half-integer if its length is even.
    A positive integer $p$ is a period of a non-empty string $s$, if $s\sbra*{i} = s\sbra*{i+p}$ for all $1 \leq i \leq \abs*{s}-p$. Note that $\abs*{s}$ is always a period of non-empty string $s$.
    
    The following lemma shows a useful property of the periods of palindromes. 
    \begin{lemma} [Periods of palindromes, Lemma 2 and Lemma 3 in \cite{KRS15}] \label{lemma:period-of-pal}
        Suppose $s$ is a non-empty palindrome and $1 \leq p \leq \abs*{s}$ is an integer. Then $p$ is a period of $s$ if and only if $s\substr{1}{\abs*{s}-p}$ is a palindrome.
    \end{lemma}

    Suppose $s$ is a string and $1 \leq i \leq \abs*{s}$. 
    Let $\prepal\rbra*{s, i}$ (resp. $\sufpal\rbra*{s, i}$) denote the longest prefix (resp. suffix) palindromic substring of $s\substr{i}{\abs*{s}}$ (resp. $s\substr{1}{i}$), and $\prelen\rbra*{s, i}$ (resp. $\suflen\rbra*{s, i}$) denote its length. Formally,
    \begin{align*}
        \prepal\rbra*{s, i} & = s\substr{i}{i+\prelen\rbra*{s, i}-1}, \\
        \sufpal\rbra*{s, i} & = s\substr{i-\suflen\rbra*{s, i}+1}{i},
    \end{align*}
    where
    \begin{align*}
        \prelen\rbra*{s, i} & = \max \set{1 \leq l \leq \abs{s} - i + 1}{s\substr{i}{i+l-1} \text{ is a palindrome}}, \\
        \suflen\rbra*{s, i} & = \max \set{1 \leq l \leq i}{s\substr{i-l+1}{i} \text{ is a palindrome}}.
    \end{align*}
    
    \subsection{Word RAM} \label{sec:word-ram}
    
    The word RAM (random access machine) \cite{FW90} is an extended computational model of RAM \cite{CR73}. In the word RAM model, all data are stored in the memory as an array $A$ of $w$-bit words. Here, a $w$-bit word means an integer between $0$ and $2^w-1$ (inclusive), which can be represented by a binary string of length $w$. The instruction set of the word RAM is an analog of that of the RAM, with indirect addressing the only exception. In the word RAM, the value of $A\sbra*{i}$ can only store limited addresses from $0$ to $2^w-1$, therefore we only have access to $2^w$ addresses. The input of size $n$ is stored in the first $n$ elements of $A$ initially. In our case, we are only interested in algorithms with space complexity polynomial in $n$, thereby assuming that $w = \Omega\rbra*{\log n}$ such that $w$ is large enough.
    
    For our purpose, our algorithms work in the word RAM under the constant cost criterion. That is, the cost of each instruction of the word RAM is a constant, i.e., $O\rbra*{1}$. In our algorithms, we only require the instruction set of the word RAM to contain basic arithmetic operations (addition and subtraction) but not multiplication, division or bit operations. 

    \subsection{Eertrees} \label{sec:eertree}
    
\newcommand{\prev}{\operatorname{prev}}

    The eertree of a string $s$ is a data structure that stores the information of all palindromic substrings of $s$ in a efficient way \cite{RS18}. Roughly speaking, the eertree is a finite-state machine that resembles trie-like trees with additional links between internal nodes, which is similar to, for example, Aho-Corasick automata \cite{AC75}. Formally, the eertree of a string $s$ is a tuple
    \[
        \eertree\rbra*{s} = \rbra*{ V, \even, \odd, \next, \link },
    \]
    where
    \begin{enumerate}
      \item $V$ is a finite set of nodes, each of which is used to represent a palindromic substring of $s$.
      \item $\even, \odd \in V$ are two special nodes, indicating the roots of palindromes of even and odd lengths, respectively.
      \item $\next \colon V \times \Sigma \to V \cup \cbra*{\nullptr}$ describes the tree structure of $\eertree\rbra*{s}$. Specifically, for every node $v \in V$ and character $c \in \Sigma$, $\next\rbra*{v, c}$ indicates the outgoing edge of $v$ labeled by $c$. In case of $\next\rbra*{v, c} = \nullptr$, it means that there is no outgoing edge of $v$ labeled by $c$. If $\next\rbra*{v, c} = u$, we write $\prev\rbra*{u} = v$ to denote the node $v$ that points to $u$. It is guaranteed that all of the outgoing edges together form a forest consisting of two trees, whose roots are $\even$ and $\odd$, respectively.
      \item $\link \colon V \to V$ is the suffix link.\footnote{The terminology ``suffix link'' is also used in suffix trees. Here, we use the same terminology ``suffix link'' in eertrees as they have similar meanings.}
    \end{enumerate}

    Let $\str\rbra*{v}$ denote the palindromic string that is represented by node $v$. For every node $v \in V$ and character $c \in \Sigma$ such that $\next(v, c) \neq \nullptr$, define $\str\rbra*{\next(v, c)} = c \str\rbra*{v} c$. Especially, $\str\rbra*{\even} = \epsilon$ is the empty string and $\str\rbra*{\odd} = \epsilon_{-1}$ is the empty string of length $-1$,\footnote{Here, the empty string $\epsilon_{-1}$ of negative length $-1$ is non-standard. This is for convenience of notations and definitions, and it is only used for defining a palindromic substring of an odd length that an eertree node represents.}
    i.e., $\abs*{\epsilon_{-1}} = -1$. Here, $c \epsilon_{-1} c = c$ for every character $c \in \Sigma$. We write $\len\rbra*{v} = \abs*{\str\rbra*{v}}$ to denote the length of the palindrome represented by node $v$. Especially, $\len\rbra*{\even} = 0$ and $\len\rbra*{\odd} = -1$. For every node $v \in V$ and character $c \in \Sigma$ such that $\next(v, c) \neq \nullptr$, we have $\len\rbra*{\next\rbra*{v, c}} = \len\rbra*{v} + 2$. For convenience, for every palindromic substring $t$ of $s$, let $\node\rbra*{t} \in V$ be the node in $\eertree\rbra*{s}$ such that $\str\rbra*{\node\rbra*{t}} = t$. For every node $v \in V \setminus \cbra*{\even, \odd}$, let $\occur\rbra*{s, v}$ be the set of start position of occurrences of $\str\rbra*{v}$ in $s$, i.e.,
    \begin{equation} \label{eq:def-occur}
        \occur\rbra*{s, v} = \set{1 \leq i \leq \abs{s} - \len\rbra*{v} + 1}{s\substr{i}{i+\len\rbra*{v}-1} = \str\rbra*{v}},
    \end{equation}
    and it holds that $\occur\rbra*{s, v} \neq \emptyset$.
    It is obvious that a palindromic substring $t$ of $s$ is unique in $s$ if and only if $\abs*{\occur\rbra*{s, \node\rbra*{t}}} = 1$. For an occurrence $i \in \occur\rbra*{s, v}$ of $\str\rbra*{v}$ in $s$, we also call $i$ the start position and $i + \len\rbra*{v} - 1$ the end position of this occurrence.

    For every node $v \in V$, $\link\rbra*{v}$ is the node of the longest proper palindromic suffix of $\str\rbra*{v}$, i.e.,
    \[
        \link\rbra*{v} = \argmax_{u \in V} \set{\len\rbra*{u}}{ \text{$\str\rbra*{u}$ is a proper palindromic suffix of $\str\rbra*{v}$} }.
    \]
    In particular, $\link\rbra*{\even} = \link\rbra*{\odd} = \odd$. 
    Indeed, $\link \colon V \to V$ defines a rooted tree $T = \rbra*{V, E}$ called the link tree with $\odd$ being the root, where
    \[
        E = \set{\rbra*{\link\rbra*{v}, v}}{v \in V \setminus \cbra*{\odd}}. 
    \]
    For every $v \in V$, let $T_v = \rbra*{V_v, E_v}$ denote the subtree of node $v$ in the link tree $T$, i.e., 
    \begin{align*}
        V_v & = \set{u \in V}{v \in \link^*\rbra*{u}}, \\
        E_v & = \set{\rbra*{a, b} \in E}{\cbra*{a, b} \subseteq V_v},
    \end{align*}
    where $\link^*\rbra{v} = \set{\link^k\rbra{v}}{k \in \mathbb{N}} = \cbra{v, \link\rbra{v}, \link^2\rbra{v}, \dots}$, $\link^k\rbra{v} = \link\rbra{\link^{k-1}\rbra{v}}$ for $k \geq 1$, and $\link^0\rbra*{v} = v$.
    
    For space efficiency, we are only interested in the minimal eertree of a string $s$. That is, there is no redundant node in the eertree with respect to string $s$. A node $v$ in the eertree is redundant with respect to a string $s$ if $\occur\rbra*{s, v} = \emptyset$. Throughout this paper, we use $\eertree\rbra*{s}$ to denote the minimal eertree of string $s$. 
    For every node $v \in V$ in $\eertree\rbra*{s}$, let $\linkcnt\rbra*{s,v}$ be the number of nodes which link to $v$, i.e.,
    \begin{equation} \label{eq:def-linkcnt}
        \linkcnt\rbra*{s,v} = \abs*{\set{u \in V}{\link\rbra*{u} = v}}.
    \end{equation}
    Note that node $u$ in Equation~(\ref{eq:def-linkcnt}) ranges over all nodes that are not redundant with respect to $s$.
    
    We first show a relationship between every node and its ancestors in the link tree.
    \begin{lemma} \label{lemma:occur-link}
        Suppose $s$ is a string and $u$ is a node in $\eertree\rbra*{s}$. For every $v \in \link^*\rbra*{u}$ with $\len\rbra*{v} \geq 1$, we have $\occur\rbra*{s,u} \cup \rbra*{\occur\rbra*{s,u} + \len\rbra*{u} - \len\rbra*{v}} \subseteq \occur\rbra*{s,v}$,
        where $\occur\rbra*{s,u} + \len\rbra*{u} - \len\rbra*{v} = \set{i + \len\rbra*{u} - \len\rbra*{v}}{i \in \occur\rbra*{s,u}}$.
    \end{lemma}
    \begin{proof}
        Let $i \in \occur\rbra*{s,u}$. By the definition, $s\substr{i}{i+\len\rbra*{u}-1} = \str\rbra*{u}$ is a palindrome. Note that $\str\rbra*{v}$ is a palindromic prefix (and also suffix) of $\str\rbra*{u}$. It follows that $s\substr{i}{i+\len\rbra*{v}-1} = s\substr{i+\len\rbra*{u}-\len\rbra*{v}}{i+\len\rbra*{u}-1} = \str\rbra*{v}$ are two occurrences of $v$, i.e., $\cbra*{i, i+\len\rbra*{u}-\len\rbra*{v}} \subseteq \occur\rbra*{s,v}$. 
    \end{proof}
    
    Inspired by Lemma \ref{lemma:occur-link}, we realize that $\occur\rbra*{s,v}$ actually implies certain occurrences of $\link\rbra*{v}$ for every node $v \in V$, provided that $\len\rbra*{\link\rbra*{v}} \geq 1$. For this reason, we define the complement of $\occur\rbra*{s,v}$ by
    \begin{equation} \label{eq:def-overline-occur}
        \overline{\occur}\rbra*{s,v} = \set{ i+\len\rbra*{v}-\len\rbra*{\link\rbra*{v}} }{i \in \occur\rbra*{s,v}}.
    \end{equation}
    Intuitively, the complement $\overline{\occur}\rbra*{s,v}$ of $\occur\rbra*{s,v}$ contains occurrences of $\str\rbra*{\link\rbra*{v}}$ which are implied by and symmetric to those occurrences in $\occur\rbra*{s,v}$ with respect to the center of $\str\rbra*{v}$ as shown in the following corollary.
    
    \begin{corollary} \label{corollary:occur}
        Suppose $s$ is a string and $u$ is a node in $\eertree\rbra*{s}$. If $\len\rbra*{\link\rbra*{u}} \geq 1$, then $\occur\rbra*{s,u} \cup \overline{\occur}\rbra*{s,u} \subseteq \occur\rbra*{s,\link\rbra*{u}}$.
    \end{corollary}
    
    \subsubsection{Stack eertrees}
    
\newcommand{\directlink} {\mathit{dlink}}
    
    Eertrees that support stack operations, which we call stack eertrees, were studied in \cite{RS18}, with an efficient trick, called the direct link, proposed for adding a new character to the end of the string. The direct link is an auxiliary information that allows us to find the longest suffix palindromic substring in $O\rbra*{1}$ time. Formally, the direct link $\directlink \colon V \times \Sigma \to V$ is defined by
    \[
        \directlink\rbra*{v, c} = \argmax_{u \in \link^*\rbra{v} \setminus \cbra*{\odd}} \set{\len\rbra*{u}}{ \str\rbra*{v}\sbra*{\len\rbra*{v} - \len\rbra*{u}} = c },
    \]
    and $\directlink\rbra*{v, c} = \odd$ if no $u$ satisfies the conditions. Proposition \ref{prop:lsp-by-dlink} shows that the longest suffix palindromic substring can be found by using direct links, and that direct links can be maintained efficiently.
    
    \begin{proposition} [Longest suffix palindromes using direct links \cite{RS18}] \label{prop:lsp-by-dlink}
    Suppose $s$ is a string and $c$ is a character. We have
    \[
        \suflen\rbra*{sc, \abs*{sc}} = \len\rbra*{ \directlink\rbra*{ \node\rbra*{\sufpal\rbra*{s, \abs*{s}}}, c } } + 2.
    \]
    \end{proposition}
    
    \begin{lemma} [Efficient direct links \cite{RS18}] \label{lemma:dlink}
        The arrays $\directlink\rbra*{v, \cdot}$ for all nodes $v$ can be stored in a persistent binary search tree such that the array $\directlink\rbra*{v, \cdot}$ for a new node $v$ can be created from $\directlink\rbra*{\prev\rbra*{v}, \cdot}$ in $O\rbra*{\log\rbra*{\sigma}}$ time and space, where $\sigma$ is the size of the alphabet.
    \end{lemma}
    
    With this efficient approach to maintain direct links, Rubinchik and Shur \cite{RS18} developed an algorithm that supports efficient online stack operations on eertrees.

    \begin{theorem} [Stack operations for eertrees \cite{RS18}]
        Stack eertrees can be implemented with time and space complexity per operation  $O\rbra*{\log\rbra*{\sigma}}$, where $\sigma$ is the size of the alphabet.
    \end{theorem}
    
    \subsubsection{Queue eertrees}
    
    Recently, eertrees for a sliding window were studied in \cite{MWN+22}. From the perspective of data structures, they support queue operations on eertrees, which we call queue eertrees. The main difficulty met in queue operations on eertrees is to check the uniqueness of the longest palindromic prefix of a string. This is because when the front character of the string is being removed, the longest palindromic prefix of the string should be deleted from the eertree if it is unique (i.e., appears only once). This difficulty was resolved in \cite{MWN+22} by maintaining the second rightmost occurrence of every palindromic substrings. To check whether a palindromic substring is unique, it is sufficient to check whether its second rightmost occurrence exists. Based on this observation, Mieno et al. \cite{MWN+22} proposed an algorithm that supports efficient online queue operations on eertrees. 

    \begin{theorem} [Queue operations for eertrees \cite{MWN+22}]
        Queue eertrees can be implemented with time and space complexity per operation  $O\rbra*{\log\rbra*{\sigma}}$, where $\sigma$ is the size of the alphabet.
    \end{theorem}
    
    \section{Framework of Double-Ended Eertrees} \label{sec:framework}
    
\newcommand{\stpos} {\textup{\texttt{start\_pos}}}
\newcommand{\edpos} {\textup{\texttt{end\_pos}}}
\newcommand{\data} {\textup{\texttt{data}}}
    
    For readability, we first illustrate the framework of double-ended eertrees in this section. Throughout this paper, three kinds of font types are used in our algorithms:
    \begin{enumerate}
        \item Typewriter font: Variables in typewriter font are being maintained in the algorithms, e.g., $\stpos$, $\edpos$, $\texttt{S}$, $\texttt{prenode}$, $\texttt{sufnode}$, etc. Moreover, after a variable of typewriter font, we write its index inside a squared bracket, e.g., $\texttt{S}\sbra*{v}$, $\texttt{data}\sbra*{\stpos}$, etc.
        \item Math italic: Variables in math italics are mathematical concepts defined independent of the execution of the algorithms, e.g., $\prepal$, $\sufpal$, etc.
        \item Roman font: Variables in Roman font are simple notions clear in the context and easy to maintain (thus we need not take much attention on their maintenance), though may depend on the execution of the algorithms, e.g., $\len$, $\link$, $\linkcnt$, etc.
    \end{enumerate}
    
    \subsection{Global index system} \label{sec:global-index}

    We need some auxiliary information that helps to maintain the whole data structure. In order to indicate indices of the strings that once occur during the operations easily and clearly, we use a range $\left[ \stpos , \edpos-1 \right]$ to store a string $s$ of length $\abs*{s} = \edpos - \stpos$ with global indices $\stpos$ and $\edpos$, where the $i$-th character of $s$ is associated with the index $\stpos + i - 1$ for $1 \leq i \leq \abs*{s}$ during the execution of the algorithm. To this end, we introduce an array $\data$ such that $\data\sbra*{\stpos + i - 1} = s\sbra*{i}$ holds during the algorithm.
    Every time an element is pushed into (resp. popped from) the front, we set $\stpos \gets \stpos - 1$ (resp. $\stpos \gets \stpos + 1$); and if it is pushed into (resp. popped from) the back, we set $\edpos \gets \edpos + 1$ (resp. $\edpos \gets \edpos - 1$). Initially, $s = \epsilon$ is the empty string, and we set $\stpos = \edpos = 0$.
    
    \subsection{\texttt{push\_back} and \texttt{push\_front}} \label{sec:framework-push}
    
    The overall idea to push a new character $c$ into the front (resp. back) of the string $s$ is straightforward (see Algorithm \ref{algo:framework:push_back} and Algorithm \ref{algo:framework:push_front}). Let $t$ be the longest palindromic suffix (resp. the longest palindromic prefix) of $s$. Then $\eertree\rbra*{sc}$ (resp. $\eertree\rbra*{cs}$) can be obtained from $\eertree\rbra*{s}$ with the help of $t$, which can be done by the original incremental construction of eertrees given in \cite{RS18} (see Appendix \ref{app:incremental-eertree} for more details). 
    
    \begin{algorithm}[!htp]
        \caption{The framework of $\texttt{push\_back}\rbra*{c}$}
        \label{algo:framework:push_back}
        \begin{algorithmic}[1]

        \State $s \gets \data\substr{\stpos}{\edpos-1}$.
        \State $t \gets \sufpal\rbra*{s, \abs*{s}}$.
        \State Obtain $\eertree\rbra*{sc}$ from $\eertree\rbra*{s}$ with the help of $t$.
        \State $\data\sbra*{\edpos} \gets c$.
        \State $\edpos \gets \edpos + 1$.
        \end{algorithmic}
    \end{algorithm}
    
    \begin{algorithm}[!htp]
        \caption{The framework of $\texttt{push\_front}\rbra*{c}$}
        \label{algo:framework:push_front}
        \begin{algorithmic}[1]

        \State $s \gets \data\substr{\stpos}{\edpos-1}$.
        \State $t \gets \prepal\rbra*{s, 1}$.
        \State Obtain $\eertree\rbra*{cs}$ from $\eertree\rbra*{s}$ with the help of $t$.
        \State $\stpos \gets \stpos - 1$.
        \State $\data\sbra*{\stpos} \gets c$.
        \end{algorithmic}
    \end{algorithm}
    
    \subsection{\texttt{pop\_back} and \texttt{pop\_front}} \label{sec:framework-pop}
    
    When dealing with pop operations, we need to delete the nodes in the eertree that are no longer involved (see Algorithm \ref{algo:framework:pop_back} and Algorithm \ref{algo:framework:pop_front}). Let $t$ be the longest palindromic suffix (resp. prefix) of $s$. $\eertree\rbra*{s\substr{1}{\abs*{s}-1}}$ (resp. $\eertree\rbra*{s\substr{2}{\abs*{s}}}$) can be obtained from $\eertree\rbra*{s}$, with $\node\rbra*{t}$ deleted if $t$ is unique in $s$.
    
    \begin{algorithm}[!htp]
        \caption{The framework of $\texttt{pop\_back}\rbra*{}$}
        \label{algo:framework:pop_back}
        \begin{algorithmic}[1]

        \State $s \gets \data\substr{\stpos}{\edpos-1}$.
        \State $t \gets \sufpal\rbra*{s, \abs*{s}}$.
        \If {$t$ is unique in $s$}
            \State Delete $\node\rbra*{t}$ from the eertree.
        \EndIf
        \State $\edpos \gets \edpos - 1$.
        \end{algorithmic}
    \end{algorithm}
    
    \begin{algorithm}[!htp]
        \caption{The framework of $\texttt{pop\_front}\rbra*{}$}
        \label{algo:framework:pop_front}
        \begin{algorithmic}[1]

        \State $s \gets \data\substr{\stpos}{\edpos-1}$.
        \State $t \gets \prepal\rbra*{s, 1}$.
        \If {$t$ is unique in $s$}
            \State Delete $\node\rbra*{t}$ from the eertree.
        \EndIf
        \State $\stpos \gets \stpos + 1$.
        \end{algorithmic}
    \end{algorithm}
    
    \subsection{Main difficulties}
    
    As discussed in Section \ref{sec:framework-push} and Section \ref{sec:framework-pop}, we see that there are two major difficulties that we must face:
    \begin{enumerate}
        \item An efficient maintenance of $\sufpal\rbra*{s, \abs*{s}}$ and $\prepal\rbra*{s, 1}$ is needed in all deque operations.
        \item An efficient approach for checking the uniqueness of a palindromic prefix or suffix is needed in all pop operations. 
    \end{enumerate}
    We will provide a solution to the above difficulties in Section \ref{sec:surface-recording}.
    
    \section{Occurrence Recording Method and Surfaces} \label{sec:occurrence-recording}
    
    In order to understand the main idea of our linear-time algorithm, we first introduce a suboptimal occurrence recording method to maintain double-ended eertrees, which requires reduced sets of occurrences in its implementation.
    In Section \ref{sec:reduced-set}, we discuss the use of reduced sets of occurrences, based on which we define surfaces in Section \ref{sec:surface}. 
    
    \subsection{Reduced sets of occurrences} \label{sec:reduced-set}

    The main idea of the occurrence recording method is to record the occurrences of all palindromic substrings. 
    The idea of this method was inspired by the algorithm to implement queue eertrees \cite{MWN+22}. However, the method used in \cite{MWN+22} is heavily based on the monotonicity of queue operations. Here, we extend the occurrence recording strategy for our more general case.
    
    Suppose that $s$ is a string and $\eertree\rbra*{s} = \rbra*{ V, \even, \odd, \next, \link }$.
    In fact, it is difficult to directly maintain the set $\occur\rbra*{s,v}$ for every node $v \in V$ because the total number of elements that appear in $\occur\rbra*{s,v}$ for all $v \in V$ can be $\Omega\rbra{\abs{s}^2}$, which is too large for our purpose. A simple example is that $s = a^n$, where $\abs*{\occur\rbra*{s,\node\rbra*{a^k}}} = n-k+1$ for every $1 \leq k \leq n$. 
    
    To overcome this issue, we maintain a reduced set $S\rbra{s, v}$ for every $v \in V \setminus \cbra*{\even, \odd}$ instead, which represents $\occur\rbra*{s,v}$ for every $v \in V \setminus \cbra*{\even, \odd}$ indirectly. 
    Specifically, we wish to maintain a reduced set $S\rbra*{s, v}$ for every $v \in V \setminus \cbra*{\even, \odd}$ such that 
    \begin{equation} \label{eq:sv}
        \occur\rbra*{s,v} = S\rbra*{s,v} \cup \bigcup_{\link\rbra*{u} = v} \rbra*{\occur\rbra*{s,u} \cup \overline{\occur}\rbra*{s,u}}.
    \end{equation}
    Intuitively, some of the occurrences of node $v$ can be obtained implicitly from the occurrences of its child nodes in the link tree; and $S\rbra*{s,v}$ is used to contain those occurrences of node $v$ that are not implied by the former.
    Specifically, for every node $u$ with $\link\rbra{u} = v$, an occurrence $i$ of node $u$ directly implies two occurrences $i$ and $i+\len\rbra{u}-\len\rbra{v}$ of node $v$ (see Corollary \ref{corollary:occur}).
    However, we may not find all occurrences of node $v$ in this way (by enumerating all nodes $u$ with $\link\rbra{u} = v$).
    In this sense, the set $S\rbra{s, v}$ complements the remaining occurrences of node $v$.

    We claim that any reduced sets $S\rbra*{s,v}$ that satisfy Equation~(\ref{eq:sv}) can help to check the uniqueness of a palindromic substring of $s$ as follows.
    \begin{lemma} [Uniqueness checking via reduced sets of occurrences] \label{lemma:unique-by-sv}
        Let $s$ be a non-empty string.
        Suppose that every node $v$ in $\eertree\rbra{s}$ is associated with a set $S\rbra{s, v}$ satisfying Equation~(\ref{eq:sv}).
        Let $s\substr{i}{j}$ be a palindrome for some $1 \leq i \leq j \leq \abs{s}$ and $v = \node\rbra{s\substr{i}{j}}$.
        Then, $s\substr{i}{j}$ is unique in $s$ if and only if  $\linkcnt\rbra*{s,v} = 0$ and $\abs*{S\rbra*{s,v}} = 1$.
    \end{lemma}
    
    \begin{proof}
    	``$\Longrightarrow$''. If $s\substr{i}{j}$ is unique in $s$, then $v$ must be a leaf node of the link tree $T$ of $\eertree\rbra*{s}$. Otherwise, $v$ has a child node $u$ in $T$, i.e., there exists a node $u$ such that $\link\rbra*{u} = v$. Let $i \in \occur\rbra*{s,u}$, and we have $\cbra*{i, i+\len\rbra*{u}-\len\rbra*{v}} \subseteq \occur\rbra*{s,v}$. Note that $i \neq i+\len\rbra*{u}-\len\rbra*{v}$ due to $\len\rbra*{v} < \len\rbra*{u}$. Therefore, $\abs*{\occur\rbra*{s,v}} \geq 2$, which contradicts with the uniqueness of $\str\rbra*{v} = s\substr{i}{j}$. Now we conclude that $v$ is a leaf node of $T$. According to Equation~(\ref{eq:def-linkcnt}), we have $\linkcnt\rbra*{s,v} = 0$, and then Equation~(\ref{eq:sv}) becomes
    	\begin{equation} \label{eq:sv_t}
    		\occur\rbra*{s,v} = S\rbra*{s,v},
    	\end{equation}
    	which implies that $\abs*{S\rbra*{s,v}} = \abs*{\occur\rbra*{s,v}} = 1$.
    	
    	``$\Longleftarrow$''. If $\linkcnt\rbra*{s,v} = 0$, then Equation~(\ref{eq:sv_t}) holds again and thus $\abs*{\occur\rbra*{s,v}} = \abs*{S\rbra*{s,v}} = 1$, which implies the uniqueness of $\str\rbra*{v} = s\substr{i}{j}$ in $s$. 
    \end{proof}

    Given a family of the reduced sets $S\rbra*{s,v}$, we further define
    \begin{align}
        \prenode\rbra*{s, i} & = \set{v \in V}{i \in S\rbra*{s,v}}, \label{eq:def-prenode} \\
        \sufnode\rbra*{s, i} & = \set{v \in V}{i - \len\rbra*{v} + 1 \in S\rbra*{s,v}}. \label{eq:def-sufnode}
    \end{align}

    As will be shown in Lemma \ref{lemma:prepal-by-prenode}, we can represent $\prepal\rbra*{s,1}$ and $\sufpal\rbra*{s,\abs*{s}}$ by an element in $\prenode\rbra*{s,1}$ and $\sufnode\rbra*{s,\abs*{s}}$, respectively. 
	\begin{lemma} \label{lemma:prepal-by-prenode}
		Suppose $s$ is a non-empty string, $S\rbra*{s,v}$ satisfies Equation~(\ref{eq:sv}) for every node $v$ in $\eertree\rbra*{s}$. We have
		\begin{align}
		\prepal\rbra*{s,1} & = \str\rbra*{\argmax_{v \in \prenode\rbra*{s,1}} \len\rbra*{v}}, \label{eq:prepal-max} \\
		\sufpal\rbra*{s,\abs*{s}} & = \str\rbra*{\argmax_{v \in \sufnode\rbra*{s,\abs*{s}}} \len\rbra*{v}}. \label{eq:sufpal-max}
		\end{align}
	\end{lemma}
	\begin{proof}
            We only show Equation~(\ref{eq:prepal-max}), and Equation~(\ref{eq:sufpal-max}) can be obtained symmetrically. 
            Let $v = \node\rbra*{\prepal\rbra*{s,1}}$. 
            Then, $\str\rbra{v}$ is non-empty, i.e., $\len\rbra{v} \geq 1$. 
            This is because the first character of $s$ forms a palindromic substring itself. 
            Moreover, there is no node $u$ in $\eertree\rbra*{s}$ satisfying both $\link\rbra{u} = v$ and $1 \in \occur\rbra{s, v}$;
            otherwise, $\str\rbra{u}$ is a palindromic prefix of $s$ and $\len\rbra{u} > \len\rbra{v}$, which violates that $\str\rbra{v}$ is the longest palindromic prefix of $s$. 
            This means that for any node $u$ with $\link\rbra{u} = v$, it always holds that $1 \notin \occur\rbra{s, u}$; 
            moreover, we have $1 \notin \overline{\occur}\rbra{s, u}$ according to Equation~(\ref{eq:def-overline-occur}).

            On the other hand, $1 \in \occur\rbra{s, v}$. 
            This leads to $1 \in S\rbra{s, v}$ by Equation~(\ref{eq:sv}).
            By Equation~(\ref{eq:def-prenode}), we have $v \in \prenode\rbra{s, 1}$; this also implies that $\prenode\rbra{s, 1} \neq \emptyset$. 
            To complete the proof, it remains to show that there is no node $u \in \prenode\rbra{s, 1}$ such that $\len\rbra{u} > \len\rbra{v}$. 
            If there is such a node $u$ satisfying these, then $1 \in S\rbra{s, u} \subseteq \occur\rbra{s, u}$, which means that $\str\rbra{u}$ is a palindromic prefix of $s$ longer than $\str\rbra{v} = \prepal\rbra{s, 1}$, violating the definition of $\prepal\rbra{s, 1}$. 
	\end{proof}

    By Lemmas \ref{lemma:unique-by-sv} and \ref{lemma:prepal-by-prenode}, we can implement a double-ended eertree by maintaining the reduced sets $S\rbra*{s,v}$. 
    In Appendix \ref{sec:occurrence-recording-app}, we provide an implementation of eertree based on the reduced sets, with amortized time complexity $O\rbra{\log\rbra{n} + \log\rbra{\sigma}}$ and worst-case space complexity $O\rbra{\log\rbra{\sigma}}$.

    \subsection{Surfaces}
	\label{sec:surface}
	
	We introduce a new notion of \textit{surfaces}, which will be an important concept in our surface recording method (see Section \ref{sec:surface-recording}).
	\begin{definition} [Surface] \label{def:surface}
		A substring occurrence $s\substr{i}{j}$ of $s$, where $1 \leq i \leq j \leq \abs*{s}$, is called a surface in $s$, if $s\substr{i}{j}$ is a palindrome, but neither $s\substr{i}{r}$ nor $s\substr{l}{j}$ is a palindrome for any $1 \leq l < i$ or $j < r \leq \abs*{s}$.
	\end{definition}
	\begin{remark}
		In Definition \ref{def:surface}, a palindromic substring $t$ of $s$ can have multiple occurrences, some of which are surfaces while the others are not. Consider the string $s = abacabaxyaba$. There are three occurrences $s\substr{1}{3}, s\substr{5}{7}, s\substr{10}{12}$ of $aba$ in $s$. According to the definition, neither $s\substr{1}{3}$ nor $s\substr{5}{7}$ is a surface in $s$ because $s\substr{1}{7} = abacaba$ is a palindrome with $s\substr{1}{3}$ and $s\substr{5}{7}$ being its prefix and suffix, respectively. However, it can be easily verified that $s\substr{10}{12}$ is a surface in $s$.
	\end{remark}
	Intuitively, a surface marks a palindromic substring occurrence that is not covered by any other palindromic substring occurrences. By the definitions of $\prelen\rbra*{s,i}$ and $\suflen\rbra*{s,i}$, we have the following straightforward properties.
    \begin{observation}
        Let $s$ be a non-empty string.
        For $1 \leq i \leq j \leq \abs{s}$, $s\substr{i}{j}$ is a surface if and only if $\prelen\rbra*{s, i} = \suflen\rbra*{s, j} = j-i+1$. 
        In particular, $s\substr{1}{\prelen\rbra*{s, 1} }$ and $s\substr{\abs*{s} - \suflen\rbra*{s, \abs*{s}} + 1}{\abs*{s}}$ are surfaces.
    \end{observation}
	
    \newcommand{\surf}{\operatorname{surf}}
    
    For every node $v$ in $\eertree\rbra*{s}$, we define the set of surfaces corresponding to $\str\rbra*{v}$ as
    \[
        \surf\rbra*{s, v} = 
        \set{ i \in \occur\rbra*{s, v} }{ s\substr{i}{i+\len\rbra*{v}-1} \text{ is a surface in $s$} }.
    \]
    The following lemma relates sets of surfaces to reduced sets of occurrences $S\rbra*{s, v}$ related to surfaces.
    
    \begin{lemma} \label{lemma:relation-sv-surface}
        Suppose $s$ is a string and $t$ is a non-empty palindromic substring of $s$.
        Let $v = \node\rbra{t}$ be the node corresponding to $t$ in $\eertree\rbra{s}$.
        The reduced set $S\rbra*{s, v}$ satisfies Equation~(\ref{eq:sv}) if and only if $\surf\rbra*{s, v} \subseteq S\rbra*{s, v} \subseteq \occur\rbra*{s, v}$.
    \end{lemma}
    \begin{proof}
        We only need to show that
        \[
            \surf\rbra*{s, v} = \occur\rbra*{s, v} \setminus \bigcup_{\link\rbra*{u} = v} \rbra[\Big]{ \occur\rbra*{s, u} \cup \overline{\occur}\rbra*{s, u} }
        \]
        for every node $v$ in $\eertree\rbra*{s}$. 
        
        ``$\supseteq$''.
        To show that the set on the right hand side is contained in $\surf\rbra*{s, v}$, we choose any element $i \in \occur\rbra*{s, v}$ such that $i \notin \occur\rbra*{s, u}$ and $i \notin \overline{\occur}\rbra*{s, u}$ for every node $u$ with $\link\rbra*{u} = v$. That is, the following three conditions hold:
        \begin{enumerate}
            \item $s\substr{i}{i+\len\rbra*{v}-1} = \str\rbra*{v}$,
            \item $s\substr{i}{i+\len\rbra*{u}-1} \neq \str\rbra*{u}$ for every node $u$ with $\link\rbra*{u} = v$,
            \item $s\substr{i - \len\rbra*{u} + \len\rbra*{v}}{i+\len\rbra*{v}-1} \neq \str\rbra*{u}$ for every node $u$ with $\link\rbra*{u} = v$.
        \end{enumerate}
        If $s\substr{i}{i+\len\rbra*{v}-1}$ is a proper suffix of a palindromic string $w = s\substr{l}{i+\len\rbra*{v}-1}$ for some $1 \leq l < i$, then there is a positive integer $k$ such that $\link^k\rbra*{\node\rbra*{w}} = v$, thus condition 3 does not hold by letting $u = \link^{k-1}\rbra*{\node\rbra*{w}}$. Therefore, $s\substr{i}{i+\len\rbra*{v}-1}$ is not a proper suffix of any palindromic substrings of $s$. Similarly, $s\substr{i}{i+\len\rbra*{v}-1}$ is not a proper prefix of any palindromic substrings of $s$. By Definition \ref{def:surface}, we have that $s\substr{i}{i+\len\rbra*{v}-1}$ is a surface in $s$, thereby $i \in \surf\rbra*{s, v}$.
        
        ``$\subseteq$''. To show that $\surf\rbra*{s, v}$ is contained in the set on the right hand side, we choose any $i \in \surf\rbra*{s, v}$, then $s\substr{i}{i+\len\rbra*{v}-1}$ is a surface in $s$. We are about to check the three conditions given above. Condition 1 holds immediately. If condition 2 does not hold, i.e., $s\substr{i}{i+\len\rbra*{u}-1} = \str\rbra*{u}$ for some node $u$ with $\link\rbra*{u} = v$, then $\str\rbra*{v} = s\substr{i}{i+\len\rbra*{v}-1}$ is a proper prefix of $\str\rbra*{u} = s\substr{i}{i+\len\rbra*{u}-1}$, which implies that $s\substr{i}{i+\len\rbra*{v}-1}$ is not a surface in $s$ --- a contradiction. Therefore, condition 2 holds. Similarly, condition 3 holds. Finally, we have that $i$ is in the set on the right hand side. 
    \end{proof}
    
    As seen in Lemma \ref{lemma:relation-sv-surface}, the set $\surf\rbra*{s, v}$ is the minimal choice of the reduced set $S\rbra*{s, v}$. Inspired by this, we are going to maintain the set $\surf\rbra*{s, v}$ implicitly.
    To this end, we write $\presurf\rbra*{s, i}$ and $\sufsurf\rbra*{s, i}$ to denote the surface in $s$ with start and end index $i$, respectively. That is,
    \begin{equation} \label{eq:def-presurf}
        \presurf\rbra*{s, i} = \begin{cases}
            \prepal\rbra*{s, i}, & s\substr{i}{i+\prelen\rbra*{s, i}-1} \text{ is a surface in } $s$, \\
            \epsilon, & \text{otherwise}.
        \end{cases}
    \end{equation}
    \begin{equation} \label{eq:def-sufsurf}
        \sufsurf\rbra*{s, i} = \begin{cases}
            \sufpal\rbra*{s, i}, & s\substr{i-\suflen\rbra*{s, i}+1}{i} \text{ is a surface in } $s$, \\
            \epsilon, & \text{otherwise}.
        \end{cases}
    \end{equation}

    The following lemma shows that $\prepal\rbra*{s, 1}$ and $\sufpal\rbra*{s, \abs*{s}}$ can be obtained by $\presurf\rbra*{s, i}$ and $\sufsurf\rbra*{s, i}$, respectively.
    This suggests an alternative way to implement a double-ended eertree, which will be investigated in Section \ref{sec:surface-recording}.
    \begin{observation} \label{lemma:prepal-by-presurf}
        Suppose $s$ is a non-empty string. Then, we have
        \begin{align*}
            \prepal\rbra*{s, 1} & = \presurf\rbra*{s, 1}, \\
            \sufpal\rbra*{s, \abs*{s}} & = \sufsurf\rbra*{s, \abs*{s}}.
        \end{align*}
    \end{observation}
    
    \section{Surface Recording Method} \label{sec:surface-recording}
    
    In the occurrence recording method, we maintain $S\rbra*{s, v}$ as a reduced set of occurrences of every palindromic substring of $s$. Although the total size of such abstractions is $O\rbra*{\abs*{s}}$, it requires an extra logarithmic factor $O\rbra*{\log\rbra*{\abs*{s}}}$ to maintain $\prepal\rbra*{s, i}$ and $\sufpal\rbra*{s, i}$ for each $1 \leq i \leq \abs*{s}$ and the elements in each $S\rbra*{s, v}$ (see Appendix \ref{sec:occurrence-recording-app} for details).
    Lemma \ref{lemma:unique-by-sv} suggests a way to check the uniqueness of palindromic substrings of $s$, which, however, does require the cardinality of $S\rbra*{s, v}$ rather than the exact elements in it. There are two aspects that can be optimized, thereby the reduced set $S\rbra*{s, v}$ no longer needed explicitly. 
    
    \begin{itemize}
        \item To make it better, our new idea is to maintain certain number $\cnt\rbra*{s, v}$ rather than a set $S\rbra*{s, v}$ for each node $v$, where $\cnt\rbra*{s, v}$ can be also used to check the uniqueness of palindromic substrings (see Lemma \ref{lemma:unique-by-cnt}). 
        
        \item We avoid maintaining $\prenode\rbra*{s, i}$ and $\sufnode\rbra*{s, i}$ according to $S\rbra*{s, v}$ as in the occurrence recording method. Instead, we find a different way to maintain the longest palindromic prefix and suffix of $s$, called the surface recording method. 
    \end{itemize}
    
    We first introduce how to check the uniqueness of palindromic substrings without the help of $S\rbra*{s, v}$ in Section \ref{sec:indirect-occur-count}, and then provide the surface recording method in Section \ref{sec:surface-recording-algorithm} using the properties of surfaces.
    
    \subsection{Indirect occurrence counting} \label{sec:indirect-occur-count}

    In this subsection, we propose an indirect occurrence counting.
    Let $\precnt\rbra*{s, v}$ and $\sufcnt\rbra*{s, v}$ denote the number of occurrences of $\str\rbra*{v}$ that 
    are, respectively, longest palindromic prefixes of suffixes of $s$ and longest palindromic suffixes of prefixes of $s$. That is,
    \begin{equation} \label{eq:def-precnt}
        \precnt\rbra*{s, v} = \abs*{\set{1 \leq i \leq \abs*{s}}{\prepal\rbra*{s, i} = \str\rbra*{v}}}.
    \end{equation}
    \begin{equation} \label{eq:def-sufcnt}
        \sufcnt\rbra*{s, v} = \abs*{\set{1 \leq i \leq \abs*{s}}{\sufpal\rbra*{s, i} = \str\rbra*{v}}}.
    \end{equation}
    The following lemma shows that the values of $\precnt\rbra*{s, v}$ and $\sufcnt\rbra*{s, v}$ are equal.

\begin{lemma} \label{lemma:precnt=sufcnt}
    Suppose $s$ is a string and $t$ is a non-empty palindromic substring of $s$.
    Let $v = \node\rbra{t}$ be the node corresponding to $t$ in $\eertree\rbra{s}$. 
    Then, we have $\precnt\rbra{s, v} = \sufcnt\rbra{s, v}$, and
        \begin{equation} \label{eq:occur=sum-sufcnt}
        \abs*{\occur\rbra*{s, v}} = \sum_{u \in T_v} \precnt\rbra*{s, u} = \sum_{u \in T_v} \sufcnt\rbra*{s, u},
        \end{equation}
    where $T_v$ is the subtree of node $v$ in the link tree of $\eertree\rbra*{s}$.
\end{lemma}
\begin{proof}
    Let $\mathit{pre}\rbra{s, v} = \set{1 \leq i \leq \abs{s}}{\prepal\rbra{s, i} = \str\rbra{v}}$.
    Then, we have $\precnt\rbra{s, v} = \abs{\mathit{pre}\rbra{s, v}}$ and
    \[
    \sum_{u \in T_v} \precnt\rbra{s, u} = \sum_{u \in T_v} \abs*{\mathit{pre}\rbra{s, u}} = \abs*{\bigsqcup_{u \in T_v} \mathit{pre}\rbra{s, u}},
    \]
    where $\bigsqcup$ denotes disjoint union. 
    For every $i \in \occur\rbra{s, v}$, we have $i \in \mathit{pre}\rbra{s, \prepal\rbra{s, i}}$ where $\prepal\rbra{s, i} \in T_v$, which gives $\abs{\occur\rbra{s, v}} \leq \abs{\bigsqcup_{u \in T_v} \mathit{pre}\rbra{s, u}}$. 
    On the other hand, for every $i \in \mathit{pre}\rbra{s, u}$ for some $u \in T_v$, we have $s\substr{i}{i+\len\rbra{v}-1} = \str\rbra{v}$ as $\str\rbra{v}$ is a prefix of $\str\rbra{u}$, which gives $\abs{\bigsqcup_{u \in T_v} \mathit{pre}\rbra{s, u}} \leq \abs{\occur\rbra{s, v}}$.
    Therefore, we have
    \[
    \abs*{\occur\rbra{s, v}} = \abs*{\bigsqcup_{u \in T_v} \mathit{pre}\rbra{s, u}} = \sum_{u \in T_v} \precnt\rbra{s, u}.
    \]
    Similarly, we can also show that 
    \[
    \abs*{\occur\rbra{s, v}} = \sum_{u \in T_v} \sufcnt\rbra{s, u},
    \]
    and obtain Equation~(\ref{eq:occur=sum-sufcnt}).

    To see $\precnt\rbra{s, v} = \sufcnt\rbra{s, v}$, note that
    \begin{align*}
        \sum_{u \in T_v} \precnt\rbra{s, u} & = \precnt\rbra{s, v} + \sum_{\link\rbra{w} = v} \sum_{u \in T_w} \precnt\rbra{s, u} \\
        & = \precnt\rbra{s, v} + \sum_{\link\rbra{w} = v} \abs{\occur\rbra{s, w}},
    \end{align*}
    where the last equality uses Equation~(\ref{eq:occur=sum-sufcnt}) for each node $w$. 
    Similarly, we have
    \[
        \sum_{u \in T_v} \sufcnt\rbra{s, u} = \sufcnt\rbra{s, v} + \sum_{\link\rbra{w} = v} \abs{\occur\rbra{s, w}}.
    \]
    Again using Equation~(\ref{eq:occur=sum-sufcnt}), we obtain that $\precnt\rbra{s, v} = \sufcnt\rbra{s, v}$.
\end{proof}

According to Lemma \ref{lemma:precnt=sufcnt}, we can define an attribute $\cnt\rbra*{s, v}$ associated with a node $v$ in $\eertree\rbra{s}$ as follows. 
\begin{definition} \label{def:cnt}
    Suppose $s$ is a string. For every node $v$ in $\eertree\rbra{s}$ that is not $\even$ or $\odd$, we define
    \[
    \cnt\rbra*{s, v}:= \precnt\rbra*{s, v} = \sufcnt\rbra*{s, v}.
    \]
\end{definition}
	
We find that $\cnt\rbra*{s, v}$ can be used to count the number of occurrences of palindromic substrings.
\begin{corollary} \label{lemma:occur-by-cnt}
    Suppose $s$ is a string and $t$ is a non-empty palindromic substring of $s$.
    Let $v = \node\rbra{t}$ be the node corresponding to $t$ in $\eertree\rbra{s}$. 
    Then, we have
    \[
    \abs*{\occur\rbra*{s, v}} = \sum_{u \in T_v} \cnt\rbra*{s, u},
    \]
    where $T_v$ is the subtree of node $v$ in the link tree $T$ of $\eertree\rbra*{s}$.
\end{corollary}
	
	Note that $\cnt\rbra*{s, v}$ is not the cardinality of $S\rbra*{s, v}$. Indeed, both of $\cnt\rbra*{s, v}$ and $S\rbra*{s, v}$ describe certain properties of $\occur\rbra*{v}$ from different perspectives. Lemma \ref{lemma:unique-by-cnt} shows that $\cnt\rbra*{s, v}$ can be used to check whether a palindromic substring is unique.

    \begin{lemma} [Uniqueness checking via indirect occurrence counting] \label{lemma:unique-by-cnt}
        Suppose $s$ is a string and $t$ is a non-empty palindromic substring of $s$. Let $v = \node\rbra{t}$ be the node corresponding to $t$ in $\eertree\rbra{s}$.
        Then $t$ is unique in $s$ if and only if  $\linkcnt\rbra*{s,v} = 0$ and $\cnt\rbra*{s, v} = 1$.
    \end{lemma}
    \begin{proof}
        By the definition of $\occur\rbra*{s, v}$, we know that $t = s\substr{i}{j}$ is unique if and only if $\abs*{\occur\rbra*{s, v}} = 1$. 
        
        ``$\Longrightarrow$''. If $s\substr{i}{j}$ is unique, i.e., $\abs*{\occur\rbra*{s, v}} = 1$, by Corollary \ref{lemma:occur-by-cnt}, we have 
        \begin{equation} \label{eq:sum-cnt=1}
            1 = \abs*{\occur\rbra*{s, v}} = \sum_{u \in T_v} \cnt\rbra*{s, u}. 
        \end{equation}
        Since $s\substr{i}{j}$ is the leftmost occurrence of $\str\rbra{v}$, we have $\sufpal\rbra{s, j} = s\substr{i}{j}$; otherwise, if $\sufpal\rbra{s, j} = s\substr{i'}{j}$ for some $i' < i$, then $s\substr{i'}{i'+\len\rbra{v}-1} = s\substr{i}{j}$ is also an occurrence of $\str\rbra{v}$, which conflicts with the uniqueness of $s\substr{i}{j}$.
        Therefore, $\cnt\rbra{s, v} \geq 1$.
        
        On the other hand, we have $\linkcnt\rbra{s, v} = 0$; otherwise, there is a node $u$ such that $\link\rbra{u} = v$ and $\str\rbra{u}$ occurs at least once in $s$, which implies that $\str\rbra{v}$ appears at least twice and thus conflicts with the uniqueness of $s\substr{i}{j}$.
        Therefore, the summation over $u \in T_v$ in Equation~(\ref{eq:sum-cnt=1}) involves only one term $\cnt\rbra{s, v}$, which implies that $\cnt\rbra{s, v} = 1$.
        
        ``$\Longleftarrow$''. If $\linkcnt\rbra*{s, v} = 0$ and $\cnt\rbra*{s, v} = 1$, 
        by Corollary \ref{lemma:occur-by-cnt}, we have that $\abs*{\occur\rbra*{s, v}} = \cnt\rbra*{s, v} = 1$. 
    \end{proof}
    
    To conclude this subsection, we show the relationship between $\cnt\rbra*{s, v}$ and $\cnt\rbra*{s', v}$, with $s'$ being modified from $s$ by adding or deleting a character at any of the both ends. 
    
    \begin{lemma} \label{lemma:cnt-update}
        Let $s$ be a string and $c$ be a character. Then
        \begin{itemize}
            \item \textup{\texttt{push\_back}}:
            \[
            \cnt\rbra*{sc, v} = \begin{cases}
                \cnt\rbra*{s, v} + 1, & \str\rbra*{v} = \sufpal\rbra*{sc, \abs*{sc}}, \\
                \cnt\rbra*{s, v}, & \text{otherwise}.
            \end{cases}
            \]
            \item \textup{\texttt{push\_front}}:
            \[
            \cnt\rbra*{cs, v} = \begin{cases}
                \cnt\rbra*{s, v} + 1, & \str\rbra*{v} = \prepal\rbra*{cs, 1}, \\
                \cnt\rbra*{s, v}, & \text{otherwise}.
            \end{cases}
            \]
            \item \textup{\texttt{pop\_back}}: If $s$ is not empty, then
            \[
            \cnt\rbra*{s\substr{1}{\abs*{s}-1}, v} = \begin{cases}
                \cnt\rbra*{s, v} - 1, & \str\rbra*{v} = \sufpal\rbra*{s, \abs*{s}}, \\
                \cnt\rbra*{s, v}, & \text{otherwise}.
            \end{cases}
            \]
            \item \textup{\texttt{pop\_front}}: If $s$ is not empty, then
            \[
            \cnt\rbra*{s\substr{2}{\abs*{s}}, v} = \begin{cases}
                \cnt\rbra*{s, v} - 1, & \str\rbra*{v} = \prepal\rbra*{s, 1}, \\
                \cnt\rbra*{s, v}, & \text{otherwise}.
            \end{cases}
            \]
        \end{itemize}
        
    \end{lemma}
    
    \begin{proof}
        We first show the identity for \texttt{push\_back}. For this case, we use the property that $\cnt\rbra*{s, v} = \sufcnt\rbra*{s, v}$. For string $sc$, the longest palindromic suffix of $sc$ is the surface with end position $\abs*{sc}$, which is $s\substr{\abs*{sc} - \suflen\rbra*{sc, \abs*{sc}} + 1}{\abs*{sc}}$; and $\sufpal\rbra*{sc, i}$ does not depend on $c$ for every $1 \leq i \leq \abs*{s}$. By the definition of $\sufcnt\rbra*{sc, v}$, we have 
        \[
            \sufcnt\rbra*{sc, v} = \begin{cases}
                \sufcnt\rbra*{s, v} + 1, & \str\rbra*{v} = \sufpal\rbra*{sc, \abs*{sc}}, \\
                \sufcnt\rbra*{s, v}, & \text{otherwise}.
            \end{cases}
        \]
        
        Next, we will show the identity for \texttt{pop\_back}. For this case, we also use the property that $\cnt\rbra*{s, v} = \sufcnt\rbra*{s, v}$. Here, we only consider the case that $\abs*{s} \geq 2$. 
        Note that the longest palindromic suffix $\sufpal\rbra*{s, \abs*{s}}$ is deleted while $\sufpal\rbra*{s, i} = \sufpal\rbra*{s\substr{1}{\abs*{s} - 1}, i}$ for $1 \leq i < \abs{s}$. 
        By the definition of $\sufcnt\rbra*{sc, v}$, we have 
        \[
            \sufcnt\rbra*{s\substr{1}{\abs*{s}-1}, v} = \begin{cases}
                \sufcnt\rbra*{s, v} - 1, & \str\rbra*{v} = \sufpal\rbra*{s, \abs*{s}}, \\
                \sufcnt\rbra*{s, v}, & \text{otherwise}.
            \end{cases}
        \]
        
        At last, the identities for \texttt{push\_front} and \texttt{pop\_front} are symmetric to those for \texttt{push\_back} and \texttt{pop\_back}. Then, these yield the proof. 
    \end{proof}

    \subsection{The algorithm} \label{sec:surface-recording-algorithm}
    
    In this subsection, we will provide an efficient algorithm for implementing double-ended eertree based on surface recording. The algorithm is rather simple, but with a complicated correctness proof.
    As we will focus on the surfaces, let
    \[
    \mathrm{surface}\rbra{s} = \set{\rbra{i, j}}{s\substr{i}{j}\text{ is a surface in }s\text{ for }1 \leq i \leq j \leq \abs{s}}
    \]
    denote the set of (the index pairs of) all surfaces in a string $s$. 
    
    \subsubsection{\texttt{push\_back} and \texttt{push\_front}} \label{sec:surface-recording-push}

    To see how surfaces change after appending a character, we point out a simple observation.
    
    \begin{observation} \label{lemma:cond2}
        Among all suffixes of a string $s$, only the longest palindromic suffix of $s$ is a surface in $s$.
    \end{observation}
    
    We also need the following basic properties of surfaces. 
    \begin{lemma} \label{lemma:cond1}
        A non-surface substring occurrence cannot become a surface after appending a character.
        That is, for any string $s$ and character $c$, if $s\substr{i}{j}$ is not a surface where $1 \leq i \leq j \leq \abs{s}$, then $sc\substr{i}{j}$ is not a surface. 
    \end{lemma}
    \begin{proof}
    We consider the following cases. 
    \begin{itemize}
        \item The string occurrence $s\substr{i}{j}$ is not a palindrome. In this case, $sc\substr{i}{j}$ is thus not a palindrome, implying that $sc\substr{i}{j}$ is not a surface in $sc$. 
        \item The string occurrence $s\substr{i}{j}$ is a palindrome but not a surface in $s$. In this case, there must be either an index $1 \leq i' < i$ or an index $j < j' \leq \abs{s}$ such that $s\substr{i'}{j}$ or $s\substr{i}{j'}$ is a palindrome, implying that $sc\substr{i}{j}$ is not a surface in $sc$. 
    \end{itemize}
    \end{proof}

    \begin{lemma} \label{lemma:cond3}
        Suppose that $s\substr{i}{j}$ is a surface in $s$. 
        Then, $sc\substr{i}{j}$ is not a surface in $sc$ if and only if $sc\substr{i}{\abs{s}+1}$ is a surface in $sc$. 
        Moreover, if $sc\substr{i}{\abs{s}+1}$ is a surface in $sc$, then $sc\substr{i}{j}$ is the longest proper palindromic prefix of $sc\substr{i}{\abs{s}+1}$.
    \end{lemma}
    \begin{proof}
        ``$\Longleftarrow$''. If $sc\substr{i}{\abs{s}+1}$ is a surface in $sc$, then $sc\substr{i}{j}$ is not a surface as $sc\substr{i}{j}$ is a proper prefix of $sc\substr{i}{\abs{s}+1}$. 

        ``$\Longrightarrow$''. If $sc\substr{i}{j}$ is not a surface in $sc$, then there is an index $j < j' \leq \abs{s}+1$ such that $sc\substr{i}{j'}$ is a palindrome.
        We further conclude that $j' = \abs{s}+1$; otherwise, $j < j' \leq \abs{s}$, which means that $sc\substr{i}{j'} = s\substr{i}{j'}$ is a palindrome, violating that $s\substr{i}{j}$ is a surface in $s$. 
        Therefore, we have that $sc\substr{i}{\abs{s}+1}$ is a palindrome. 
        We can further conclude that $sc\substr{i}{\abs{s}+1}$ is a surface in $sc$; otherwise, there is an index $1 \leq i' < i$ such that $sc\substr{i'}{\abs{s}+1}$ is a palindrome. 
        Then we consider the following three cases. See Figure \ref{fig:illustration} for the illustration of each case. 
        \begin{enumerate}[(a)]
            \item $i+j < i'+\abs{s}+1$, i.e., the center of $sc\substr{i}{j}$ is to the left of the center of $sc\substr{i'}{\abs{s}+1}$. 
            Let $i'' = i'+\abs{s}+1 - i \leq \abs{s}$ be the symmetric index of $i$ about the center of $sc\substr{i'}{\abs{s}+1}$. 
            Then, $sc\substr{i}{i''}$ is a palindrome (note that $i \leq i''$), as $sc\substr{i'}{\abs{s}+1}$ is a palindrome. 
            On the other hand, $s\substr{i}{j} = sc\substr{i}{j}$ is a proper prefix of $s\substr{i}{i''} = sc\substr{i}{i''}$ (note that $j < i''$), which violates that $s\substr{i}{j}$ is a surface in $s$. 
            \item $i+j = i'+\abs{s}+1$, i.e., the centers of $sc\substr{i}{j}$ and $sc\substr{i'}{\abs{s}+1}$ coincide. 
            Note that $sc\substr{i}{\abs{s}+1}$ is a palindromic suffix of $sc\substr{i'}{\abs{s}+1}$, then by symmetry, $sc\substr{i'}{j}$ is a palindromic prefix of $sc\substr{i'}{\abs{s}+1}$. 
            On the other hand, $sc\substr{i}{j} = s\substr{i}{j}$ is a proper prefix of $sc\substr{i'}{j} = s\substr{i'}{j}$, which violates that $s\substr{i}{j}$ is a surface in $s$. 
            \item $i+j > i'+\abs{s}+1$, i.e., the center of $sc\substr{i}{j}$ is to the right of the center of $sc\substr{i'}{\abs{s}+1}$. 
            Let $j'' = i'+\abs{s}+1 - j \leq \abs{s}$ be the symmetric index of $j$ about the center of $sc\substr{i'}{\abs{s}+1}$. 
            Then, $sc\substr{j''}{j}$ is a palindrome (note that $j'' \leq j$), as $sc\substr{i'}{\abs{s}+1}$ is a palindrome. 
            On the other hand, $s\substr{i}{j} = sc\substr{i}{j}$ is a proper suffix of $s\substr{j''}{j} = sc\substr{j''}{j}$ (note that $j'' < i$), which violates that $s\substr{i}{j}$ is a surface in $s$. 
        \end{enumerate}
    \end{proof}

    \begin{figure}
        \centering
        \tikzset{every picture/.style={line width=0.75pt}} 
\begin{subfigure}[t]{.31\linewidth}
\centering
\begin{tikzpicture}[x=0.75pt,y=0.75pt,yscale=-1,xscale=1]

\draw  [color={rgb, 255:red, 0; green, 0; blue, 0 }  ,draw opacity=1 ][fill={rgb, 255:red, 126; green, 211; blue, 33 }  ,fill opacity=1 ] (130,151) -- (290,151) -- (290,161) -- (130,161) -- cycle ;
\draw  [fill={rgb, 255:red, 245; green, 166; blue, 35 }  ,fill opacity=1 ] (150,131) -- (230,131) -- (230,141) -- (150,141) -- cycle ;
\draw  [fill={rgb, 255:red, 189; green, 16; blue, 224 }  ,fill opacity=1 ] (150,141) -- (270,141) -- (270,151) -- (150,151) -- cycle ;
\draw [color={rgb, 255:red, 0; green, 0; blue, 0 }  ,draw opacity=1 ] [dash pattern={on 4.5pt off 4.5pt}]  (210,101) -- (210,181) ;
\draw [color={rgb, 255:red, 245; green, 166; blue, 35 }  ,draw opacity=1 ] [dash pattern={on 4.5pt off 4.5pt}]  (190,101) -- (190,181) ;
\draw    (150,104) -- (150,131) ;
\draw [shift={(150,101)}, rotate = 90] [fill={rgb, 255:red, 0; green, 0; blue, 0 }  ][line width=0.08]  [draw opacity=0] (8.93,-4.29) -- (0,0) -- (8.93,4.29) -- cycle    ;
\draw    (152,121) -- (208,121) ;
\draw [shift={(210,121)}, rotate = 180] [color={rgb, 255:red, 0; green, 0; blue, 0 }  ][line width=0.75]    (10.93,-3.29) .. controls (6.95,-1.4) and (3.31,-0.3) .. (0,0) .. controls (3.31,0.3) and (6.95,1.4) .. (10.93,3.29)   ;
\draw [shift={(150,121)}, rotate = 0] [color={rgb, 255:red, 0; green, 0; blue, 0 }  ][line width=0.75]    (10.93,-3.29) .. controls (6.95,-1.4) and (3.31,-0.3) .. (0,0) .. controls (3.31,0.3) and (6.95,1.4) .. (10.93,3.29)   ;
\draw    (212,121) -- (268,121) ;
\draw [shift={(270,121)}, rotate = 180] [color={rgb, 255:red, 0; green, 0; blue, 0 }  ][line width=0.75]    (10.93,-3.29) .. controls (6.95,-1.4) and (3.31,-0.3) .. (0,0) .. controls (3.31,0.3) and (6.95,1.4) .. (10.93,3.29)   ;
\draw [shift={(210,121)}, rotate = 0] [color={rgb, 255:red, 0; green, 0; blue, 0 }  ][line width=0.75]    (10.93,-3.29) .. controls (6.95,-1.4) and (3.31,-0.3) .. (0,0) .. controls (3.31,0.3) and (6.95,1.4) .. (10.93,3.29)   ;
\draw    (270,104) -- (270,151) ;
\draw [shift={(270,101)}, rotate = 90] [fill={rgb, 255:red, 0; green, 0; blue, 0 }  ][line width=0.08]  [draw opacity=0] (8.93,-4.29) -- (0,0) -- (8.93,4.29) -- cycle    ;
\draw    (130,161) -- (130,178) ;
\draw [shift={(130,181)}, rotate = 270] [fill={rgb, 255:red, 0; green, 0; blue, 0 }  ][line width=0.08]  [draw opacity=0] (8.93,-4.29) -- (0,0) -- (8.93,4.29) -- cycle    ;
\draw    (290,161) -- (290,178) ;
\draw [shift={(290,181)}, rotate = 270] [fill={rgb, 255:red, 0; green, 0; blue, 0 }  ][line width=0.08]  [draw opacity=0] (8.93,-4.29) -- (0,0) -- (8.93,4.29) -- cycle    ;
\draw    (230,104) -- (230,131) ;
\draw [shift={(230,101)}, rotate = 90] [fill={rgb, 255:red, 0; green, 0; blue, 0 }  ][line width=0.08]  [draw opacity=0] (8.93,-4.29) -- (0,0) -- (8.93,4.29) -- cycle    ;
\draw (124,183) node [anchor=north west][inner sep=0.75pt]   [align=left] {$i'$};
\draw (267,183) node [anchor=north west][inner sep=0.75pt]   [align=left] {$\abs{s}+1$};
\draw (147,82) node [anchor=north west][inner sep=0.75pt]   [align=left] {$i$};
\draw (226,82) node [anchor=north west][inner sep=0.75pt]   [align=left] {$j$};
\draw (266,82) node [anchor=north west][inner sep=0.75pt]   [align=left] {$i''$};

\end{tikzpicture}
\caption{$i+j < i'+\abs{s}+1$}
\end{subfigure}
\begin{subfigure}[t]{.31\linewidth}
\centering

\tikzset{every picture/.style={line width=0.75pt}} 

\begin{tikzpicture}[x=0.75pt,y=0.75pt,yscale=-1,xscale=1]

\draw  [color={rgb, 255:red, 0; green, 0; blue, 0 }  ,draw opacity=1 ][fill={rgb, 255:red, 126; green, 211; blue, 33 }  ,fill opacity=1 ] (360,140) -- (520,140) -- (520,150) -- (360,150) -- cycle ;
\draw  [fill={rgb, 255:red, 245; green, 166; blue, 35 }  ,fill opacity=1 ] (400,109) -- (480,109) -- (480,120) -- (400,120) -- cycle ;
\draw    (480,94) -- (480,120) ;
\draw [shift={(480,91)}, rotate = 90] [fill={rgb, 255:red, 0; green, 0; blue, 0 }  ][line width=0.08]  [draw opacity=0] (8.93,-4.29) -- (0,0) -- (8.93,4.29) -- cycle    ;
\draw    (360,151) -- (360,168) ;
\draw [shift={(360,171)}, rotate = 270] [fill={rgb, 255:red, 0; green, 0; blue, 0 }  ][line width=0.08]  [draw opacity=0] (8.93,-4.29) -- (0,0) -- (8.93,4.29) -- cycle    ;
\draw    (520,151) -- (520,168) ;
\draw [shift={(520,171)}, rotate = 270] [fill={rgb, 255:red, 0; green, 0; blue, 0 }  ][line width=0.08]  [draw opacity=0] (8.93,-4.29) -- (0,0) -- (8.93,4.29) -- cycle    ;
\draw    (400,93) -- (400,120) ;
\draw [shift={(400,90)}, rotate = 90] [fill={rgb, 255:red, 0; green, 0; blue, 0 }  ][line width=0.08]  [draw opacity=0] (8.93,-4.29) -- (0,0) -- (8.93,4.29) -- cycle    ;
\draw  [fill={rgb, 255:red, 189; green, 16; blue, 224 }  ,fill opacity=1 ] (360,120) -- (480,120) -- (480,130) -- (360,130) -- cycle ;
\draw  [fill={rgb, 255:red, 189; green, 16; blue, 224 }  ,fill opacity=1 ] (400,130) -- (520,130) -- (520,140) -- (400,140) -- cycle ;
\draw [color={rgb, 255:red, 0; green, 0; blue, 0 }  ,draw opacity=1 ] [dash pattern={on 4.5pt off 4.5pt}]  (440,91) -- (440,171) ;

\draw (354,172) node [anchor=north west][inner sep=0.75pt]   [align=left] {$i'$};
\draw (497,172) node [anchor=north west][inner sep=0.75pt]   [align=left] {$|s|+1$};
\draw (397,72) node [anchor=north west][inner sep=0.75pt]   [align=left] {$i$};
\draw (474,72) node [anchor=north west][inner sep=0.75pt]   [align=left] {$j$};

\end{tikzpicture}

\caption{$i+j = i'+\abs{s}+1$}
\end{subfigure}
\begin{subfigure}[t]{.31\linewidth}
\centering

\tikzset{every picture/.style={line width=0.75pt}} 

\begin{tikzpicture}[x=0.75pt,y=0.75pt,yscale=-1,xscale=1]

\draw  [color={rgb, 255:red, 0; green, 0; blue, 0 }  ,draw opacity=1 ][fill={rgb, 255:red, 126; green, 211; blue, 33 }  ,fill opacity=1 ] (130,151) -- (290,151) -- (290,161) -- (130,161) -- cycle ;
\draw  [fill={rgb, 255:red, 245; green, 166; blue, 35 }  ,fill opacity=1 ] (190,131) -- (270,131) -- (270,141) -- (190,141) -- cycle ;
\draw  [fill={rgb, 255:red, 189; green, 16; blue, 224 }  ,fill opacity=1 ] (150,141) -- (270,141) -- (270,151) -- (150,151) -- cycle ;
\draw [color={rgb, 255:red, 0; green, 0; blue, 0 }  ,draw opacity=1 ] [dash pattern={on 4.5pt off 4.5pt}]  (210,101) -- (210,181) ;
\draw [color={rgb, 255:red, 245; green, 166; blue, 35 }  ,draw opacity=1 ] [dash pattern={on 4.5pt off 4.5pt}]  (230,100) -- (230,180) ;
\draw    (150,104) -- (150,141) ;
\draw [shift={(150,101)}, rotate = 90] [fill={rgb, 255:red, 0; green, 0; blue, 0 }  ][line width=0.08]  [draw opacity=0] (8.93,-4.29) -- (0,0) -- (8.93,4.29) -- cycle    ;
\draw    (152,121) -- (208,121) ;
\draw [shift={(210,121)}, rotate = 180] [color={rgb, 255:red, 0; green, 0; blue, 0 }  ][line width=0.75]    (10.93,-3.29) .. controls (6.95,-1.4) and (3.31,-0.3) .. (0,0) .. controls (3.31,0.3) and (6.95,1.4) .. (10.93,3.29)   ;
\draw [shift={(150,121)}, rotate = 0] [color={rgb, 255:red, 0; green, 0; blue, 0 }  ][line width=0.75]    (10.93,-3.29) .. controls (6.95,-1.4) and (3.31,-0.3) .. (0,0) .. controls (3.31,0.3) and (6.95,1.4) .. (10.93,3.29)   ;
\draw    (212,121) -- (268,121) ;
\draw [shift={(270,121)}, rotate = 180] [color={rgb, 255:red, 0; green, 0; blue, 0 }  ][line width=0.75]    (10.93,-3.29) .. controls (6.95,-1.4) and (3.31,-0.3) .. (0,0) .. controls (3.31,0.3) and (6.95,1.4) .. (10.93,3.29)   ;
\draw [shift={(210,121)}, rotate = 0] [color={rgb, 255:red, 0; green, 0; blue, 0 }  ][line width=0.75]    (10.93,-3.29) .. controls (6.95,-1.4) and (3.31,-0.3) .. (0,0) .. controls (3.31,0.3) and (6.95,1.4) .. (10.93,3.29)   ;
\draw    (270,104) -- (270,151) ;
\draw [shift={(270,101)}, rotate = 90] [fill={rgb, 255:red, 0; green, 0; blue, 0 }  ][line width=0.08]  [draw opacity=0] (8.93,-4.29) -- (0,0) -- (8.93,4.29) -- cycle    ;
\draw    (130,161) -- (130,178) ;
\draw [shift={(130,181)}, rotate = 270] [fill={rgb, 255:red, 0; green, 0; blue, 0 }  ][line width=0.08]  [draw opacity=0] (8.93,-4.29) -- (0,0) -- (8.93,4.29) -- cycle    ;
\draw    (290,161) -- (290,178) ;
\draw [shift={(290,181)}, rotate = 270] [fill={rgb, 255:red, 0; green, 0; blue, 0 }  ][line width=0.08]  [draw opacity=0] (8.93,-4.29) -- (0,0) -- (8.93,4.29) -- cycle    ;
\draw    (190,104) -- (190,131) ;
\draw [shift={(190,101)}, rotate = 90] [fill={rgb, 255:red, 0; green, 0; blue, 0 }  ][line width=0.08]  [draw opacity=0] (8.93,-4.29) -- (0,0) -- (8.93,4.29) -- cycle    ;

\draw (124,183) node [anchor=north west][inner sep=0.75pt]   [align=left] {$i'$};
\draw (267,183) node [anchor=north west][inner sep=0.75pt]   [align=left] {$\abs{s}+1$};
\draw (145,82) node [anchor=north west][inner sep=0.75pt]   [align=left] {$j''$};
\draw (186,82) node [anchor=north west][inner sep=0.75pt]   [align=left] {$i$};
\draw (266,82) node [anchor=north west][inner sep=0.75pt]   [align=left] {$j$};

\end{tikzpicture}

\caption{$i+j > i'+\abs{s}+1$}
\end{subfigure}
\caption{Illustration for the proof of Lemma \ref{lemma:cond3}.}
    \label{fig:illustration}
\end{figure}
    
    By Observation \ref{lemma:cond2}, Lemma \ref{lemma:cond1}, and Lemma \ref{lemma:cond3}, we can fully characterize the change of surfaces after appending a character to a string. 

    \begin{lemma} \label{lemma:surface-push-back}
        Suppose that $s$ is a non-empty string and $c$ is a character. 
        Let $t = \sufpal\rbra{sc, \abs{sc}}$.
        \begin{enumerate}
        \item If $\abs{t} = 1$, then \label{item:1}
        \[
        \mathrm{surface}\rbra{sc} = \mathrm{surface}\rbra{s} \cup \cbra{\rbra{\abs{sc}, \abs{sc}}}.
        \]
        \item If $\abs{t} \geq 2$, then \label{item:2}
        \[
        \mathrm{surface}\rbra{sc} = \mathrm{surface}\rbra{s} \setminus \cbra{\rbra{\abs{sc}-\abs{t}+1, \abs{sc}-\abs{t}+\abs{t'}}} \cup \cbra{\rbra{\abs{sc}-\abs{t}+1, \abs{sc}}},
        \]
        where $t' = \str\rbra{\link\rbra{\node\rbra{t}}}$.
        \end{enumerate}
    \end{lemma}
    \begin{proof}
        By Observation \ref{lemma:cond2} and Lemma \ref{lemma:cond1}, most surfaces do not change.
        Item \ref{item:1} is trivial, and Item \ref{item:2} is due to the fact that only the elements $\rbra{i, j}$ with $i=\abs{sc}-\abs{t}+1$ will change due to Lemma \ref{lemma:cond3}.
    \end{proof}

    The following lemma shows how to efficiently maintain the surfaces by revealing how $\presurf\rbra*{s, i}$ and $\sufsurf\rbra*{s, i}$ relate to $\presurf\rbra*{sc, i}$ and $\sufsurf\rbra*{sc, i}$ as a character $c$ is added at the back of string $s$. 
    
    \begin{lemma} [Surface recording for \texttt{push\_back}] \label{lemma:surface-recording-push-back}
        Let $s$ be a string and $c$ be a character. 
        Let $t = \sufpal\rbra*{sc, \abs*{sc}}$ and $t' = \str\rbra*{\link\rbra*{\node\rbra*{t}}}$. Then
        \[
            \presurf\rbra*{sc, i} = \begin{cases}
                t, & i = \abs*{sc} - \abs*{t} + 1, \\
                \epsilon, & i = \abs*{sc} \text{ and } \abs*{t} \neq 1, \\
                \presurf\rbra*{s, i}, & \text{otherwise},
            \end{cases}
        \]
        and
        \[
            \sufsurf\rbra*{sc, i} = \begin{cases}
                t, & i = \abs*{sc}, \\
                \epsilon, & i = \abs*{sc} - \abs*{t} + \abs*{t'} \text{ and } \sufsurf\rbra*{s, \abs*{sc} - \abs*{t} + \abs*{t'}} = t' \text{ and } \abs*{t'} \geq 1, \\
                \sufsurf\rbra*{s, i}, & \text{otherwise}.
            \end{cases}
        \]
    \end{lemma}
    \begin{proof}
        This is straightforward due to Lemma \ref{lemma:surface-push-back}.
    \end{proof}
    
    Symmetrically, we show how $\presurf\rbra*{s, i}$ and $\sufsurf\rbra*{s, i}$ relate to $\presurf\rbra*{cs, i}$ and $\sufsurf\rbra*{cs, i}$ as a character is added at the front of string $s$.
    
    \begin{lemma} [Surface recording for \texttt{push\_front}] \label{lemma:surface-recording-push-front}
        Let $s$ be a string and $c$ be a character. 
        Let $t = \prepal\rbra*{cs, 1}$ and $t' = \str\rbra*{\link\rbra*{\node\rbra*{t}}}$. Then
        \[
            \presurf\rbra*{cs, i} = \begin{cases}
                t, & i = 1, \\
                \epsilon, & i = \abs*{t} - \abs*{t'} + 1 \text{ and } \presurf\rbra*{s, \abs*{t} - \abs*{t'}} = t' \text{ and } \abs*{t'} \geq 1, \\
                \presurf\rbra*{s, i - 1}, & \text{otherwise}.
            \end{cases}
        \]
        and
        \[
            \sufsurf\rbra*{cs, i} = \begin{cases}
                t, & \abs*{t}, \\
                \epsilon, & i = 1 \text{ and } \abs*{t} \neq 1, \\
                \sufsurf\rbra*{s, i - 1}, & \text{otherwise}.
            \end{cases}
        \]
    \end{lemma}
    
    \begin{proof}
        Because of symmetry, the proof is similar to that of Lemma \ref{lemma:surface-recording-push-back}.
    \end{proof}
    
    The algorithms for $\texttt{push\_back}$ and $\texttt{push\_front}$ are given as follows.
    
    \begin{itemize}
        \item $\texttt{push\_back}\rbra*{c}$:
        \begin{enumerate}
            \item Find the node $v$ of the longest palindromic suffix $\sufpal\rbra*{s, \abs*{s}}$ of $s$.
            \item Modify $\eertree\rbra*{s}$ to $\eertree\rbra*{sc}$ according to $v$ and $c$ by online construction of eertrees \cite{RS18}.
            \item Maintain $\texttt{presurf}$ and $\texttt{sufsurf}$ according to Lemma \ref{lemma:surface-recording-push-back}.
            \item Maintain $\texttt{cnt}$ according to Lemma \ref{lemma:cnt-update}.
        \end{enumerate}
        \item $\texttt{push\_front}\rbra*{c}$:
        \begin{enumerate}
            \item Find the node $v$ of the longest palindromic prefix $\prepal\rbra*{s, 1}$ of $s$.
            \item Modify $\eertree\rbra*{s}$ to $\eertree\rbra*{cs}$ according to $v$ and $c$ by online construction of eertrees \cite{RS18}.
            \item Maintain $\texttt{presurf}$ and $\texttt{sufsurf}$ according to Lemma \ref{lemma:surface-recording-push-front}.
            \item Maintain $\texttt{cnt}$ according to Lemma \ref{lemma:cnt-update}.
        \end{enumerate}
    \end{itemize}
    
    For completeness, we provide formal and detailed descriptions of the algorithm for \texttt{push\_back} in Algorithm \ref{algo:surface:push_back} and Algorithm \ref{algo:surface:push_front}.
    
    \begin{algorithm}[htp]
        \caption{$\texttt{push\_back}\rbra*{c}$ via surface recording method}
        \label{algo:surface:push_back}
        \begin{algorithmic}[1]

        \State $s \gets \data\substr{\stpos}{\edpos-1}$.
        \State $v \gets \texttt{sufsurf}\sbra*{\edpos - 1}$. \Comment{This ensures that $\str\rbra*{v} = \sufpal\rbra*{s, \abs*{s}}$.}
        \State Obtain $\eertree\rbra*{sc}$ from $\eertree\rbra*{s}$ with the help of $v$, and let $v' \gets \node\rbra*{\sufpal\rbra*{sc, \abs*{sc}}}$.
        \State $\data\sbra*{\edpos} \gets c$.
        \State $\edpos \gets \edpos + 1$.
        \State $\texttt{presurf}\sbra*{\edpos - 1} \gets \even$.
        \State $\texttt{presurf}\sbra*{\edpos - \len\rbra*{v'}} \gets v'$.
        \State $\texttt{sufsurf}\sbra*{\edpos - 1} \gets v'$.
        \If {$\len\rbra*{\link\rbra*{v'}} \geq 1$ and $\texttt{sufsurf}\sbra*{\edpos - \len\rbra*{v'} + \len\rbra*{\link\rbra*{v'}} - 1} = \link\rbra*{v'}$}
            \State $\texttt{sufsurf}\sbra*{\edpos - \len\rbra*{v'} + \len\rbra*{\link\rbra*{v'}} - 1} \gets \even$.
        \EndIf
        \State $\texttt{cnt}\sbra*{v'} \gets \texttt{cnt}\sbra*{v'} + 1$.
        \end{algorithmic}
    \end{algorithm}

    \begin{algorithm}[htp]
        \caption{$\texttt{push\_front}\rbra*{c}$ via surface recording method}
        \label{algo:surface:push_front}
        \begin{algorithmic}[1]

        \State $s \gets \data\substr{\stpos}{\edpos-1}$.
        \State $v \gets \texttt{presurf}\sbra*{\stpos}$. \Comment{This ensures that $\str\rbra*{v} = \prepal\rbra*{s, 1}$.}
        \State Obtain $\eertree\rbra*{sc}$ from $\eertree\rbra*{s}$ with the help of $v$, and let $v' \gets \node\rbra*{\sufpal\rbra*{cs, \abs*{cs}}}$.
        \State $\stpos \gets \stpos - 1$.
        \State $\data\sbra*{\stpos} \gets c$.
        \State $\texttt{presurf}\sbra*{\stpos} \gets v'$.
        \State $\texttt{sufsurf}\sbra*{\stpos} \gets \even$.
        \State $\texttt{sufsurf}\sbra*{\stpos + \len\rbra*{v'} - 1} \gets v'$.
        \If {$\len\rbra*{\link\rbra*{v'}} \geq 1$ and $\texttt{presurf}\sbra*{\stpos + \len\rbra*{v'} - \len\rbra*{\link\rbra*{v'}}} = \link\rbra*{v'}$}
            \State $\texttt{presurf}\sbra*{\stpos + \len\rbra*{v'} - \len\rbra*{\link\rbra*{v'}}} \gets \even$.
        \EndIf
        \State $\texttt{cnt}\sbra*{v'} \gets \texttt{cnt}\sbra*{v'} + 1$.
        \end{algorithmic}
    \end{algorithm}
    
    \subsubsection{\texttt{pop\_back} and \texttt{pop\_front}} \label{sec:surface-recording-pop}

    To see how surfaces change after deleting a character, we show the following basic properties of surfaces. 

    \begin{observation} \label{lemma:cond1p}
        A surface that does not cover the last character will still be a surface after deleting the last character. 
        That is, for any non-empty string $s$, if $s\substr{i}{j}$ is a surface in $s$ with $j < \abs{s}$, then $s'\substr{i}{j}$ is a surface in $s'$, where $s' = s\substr{1}{\abs{s}-1}$.
    \end{observation}

    \begin{lemma} \label{lemma:cond2p}
        Suppose that $s$ is a string with $\abs{s} \geq 2$ and $t = \sufpal\rbra{s, \abs{s}}$ satisfies $\abs{t} \geq 2$. 
        Let $s\substr{i}{j}$ be a palindrome but not a surface in $s$, where $1 \leq i \leq j < \abs{s}$. 
        Denote $s' = s\substr{1}{\abs{s}-1}$ and $t' = \str\rbra{\link\rbra{\node\rbra{t}}}$. 
        Then, $s'\substr{i}{j}$ is a surface in $s'$ if and only if all of the following three conditions hold:
        \begin{enumerate}
            \item $s\substr{i}{\abs{s}} = t$. \label{item:c1}
            \item $s\substr{i}{j} = t'$. \label{item:c2}
            \item $\sufpal\rbra{s, j} = t'$. \label{item:c3}
        \end{enumerate}
    \end{lemma}
    \begin{proof}
        ``$\Longrightarrow$''. If $s'\substr{i}{j}$ is a surface in $s'$, then $s\substr{i}{\abs{s}}$ is a surface in $s$ by using Lemma \ref{lemma:cond3} with $s \coloneqq s'$ and $c \coloneqq s\sbra{\abs{s}}$. 
        By Observation \ref{lemma:prepal-by-presurf}, we further have $s\substr{i}{\abs{s}} = \sufpal\rbra{s, \abs{s}} = t$, which shows Condition \ref{item:c1}. 
        Note that $s\substr{i}{j}$ is a proper palindromic prefix of $s\substr{i}{\abs{s}}$, which implies that $\abs{t'} \geq j-i+1$. 
        We can conclude that $\abs{t'} = j-i+1$, which shows Condition \ref{item:c2}; otherwise, $s\substr{i}{i+\abs{t'}-1}$ is a palindrome with a proper prefix $s\substr{i}{j}$, violating that $s'\substr{i}{j}$ is a surface in $s'$. 
        Now that $s\substr{i}{j} = t'$, which implies that $\suflen\rbra{s, j} \geq \abs{t'}$. 
        We can conclude that $\suflen\rbra{s, j} = \abs{t'}$, which shows Condition \ref{item:c3}; otherwise, $\suflen\rbra{s, j} > \abs{t'}$, which implies that $s\substr{j-\suflen\rbra{s,j}+1}{j}$ is a palindrome with a proper suffix $s\substr{i}{j}$, violating that $s'\substr{i}{j}$ is a surface in $s'$. 

        ``$\Longleftarrow$''. Assume that Conditions \ref{item:c1} to \ref{item:c3} hold. 
        If $s'\substr{i}{j}$ is not a surface in $s'$, then we have to consider the following two cases. 
        \begin{itemize}
            \item There is an index $1 \leq i' < i$ such that $s'\substr{i'}{j}$ is a palindrome, which violates Condition \ref{item:c3}. 
            \item There is an index $j < j' < \abs{s}$ such that $s'\substr{i}{j'}$ is a palindrome of length $j' - i + 1 > j - i + 1 = \abs{t'}$, which violates the definition of $t'$, i.e., $t' = s\substr{i}{j}$ (by Condition \ref{item:c2}) is the longest proper palindromic prefix of $t = s\substr{i}{\abs{s}}$ (by Condition \ref{item:c1}).
        \end{itemize}
        Therefore, $s'\substr{i}{j}$ is a surface in $s'$. 
    \end{proof}

    By Observation \ref{lemma:cond1p} and Lemma \ref{lemma:cond2p}, we can fully characterize the change of surfaces after deleting a character from a string. 

    \begin{lemma} \label{lemma:surface-pop-back}
        Suppose that $s$ is a string of length $\abs{s} \geq 2$. 
        Let $t = \sufpal\rbra{s, \abs{s}}$ and $t' = \str\rbra{\link\rbra{\node\rbra{t}}}$.
        \begin{enumerate}
        \item If $\abs{t} = 1$, then \label{item:1p}
        \[
        \mathrm{surface}\rbra{s\substr{1}{\abs{s}-1}} = \mathrm{surface}\rbra{s} \setminus \cbra{\rbra{\abs{s}, \abs{s}}}.
        \]
        \item If $\abs{t} \geq 2$ and $\sufpal\rbra{s, \abs{s}-\abs{t}+\abs{t'}} = t'$, then \label{item:2p}
        \[
        \mathrm{surface}\rbra{s\substr{1}{\abs{s}-1}} = \mathrm{surface}\rbra{s} \setminus \cbra{\rbra{\abs{s}-\abs{t}+1, \abs{s}}} \cup \cbra{\rbra{\abs{s}-\abs{t}+1, \abs{s} - \abs{t} + \abs{t'}}},
        \]
        \item If $\abs{t} \geq 2$ and $\sufpal\rbra{s, \abs{s}-\abs{t}+\abs{t'}} \neq t'$, then \label{item:3p}
        \[
        \mathrm{surface}\rbra{s\substr{1}{\abs{s}-1}} = \mathrm{surface}\rbra{s} \setminus \cbra{\rbra{\abs{s}-\abs{t}+1, \abs{s}}}.
        \]
        \end{enumerate}
    \end{lemma}
    \begin{proof}
        By Observation \ref{lemma:cond1p}, most surfaces do not change. 
        Item \ref{item:1p} is trivial, and Items \ref{item:2p} and \ref{item:3p} are based on whether Condition \ref{item:c3} in Lemma \ref{lemma:cond2p} holds. 
    \end{proof}
    
    The following lemma shows how to efficiently maintain the surfaces by revealing (i) how $\presurf\rbra*{s, i}$ and $\sufsurf\rbra*{s, i}$ relate to $\presurf\rbra*{s\substr{1}{\abs*{s}-1}, i}$ and $\sufsurf\rbra*{s\substr{1}{\abs*{s}-1}, i}$ as a character is deleted from the back of string $s$ and (ii) how $\presurf\rbra*{s, i}$ and $\sufsurf\rbra*{s, i}$ relate to $\presurf\rbra*{s\substr{2}{\abs*{s}}, i}$ and $\sufsurf\rbra*{s\substr{2}{\abs*{s}}, i}$ as a character is deleted from the front of string $s$. 
    
    \begin{lemma}[Surface recording for \texttt{pop\_back} and \texttt{pop\_front}] \label{lemma:surface-recording-pop-back}
        Let $s$ be a non-empty string. 
        
        1. For \textup{\texttt{pop\_back}} operations, let $t = \sufpal\rbra*{s, \abs*{s}}$ and $t' = \str\rbra*{\link\rbra*{\node\rbra*{t}}}$. Then
        \[
            \presurf\rbra*{s\substr{1}{\abs*{s}-1}, i} = \begin{cases}
                t', & i = \abs*{s} - \abs*{t} + 1 \text{ and }  \sufsurf\rbra*{s, \abs*{s} - \abs*{t} + \abs*{t'}} = \epsilon, \\
                \epsilon, & i = \abs*{s} - \abs*{t} + 1 \text{ and } \abs*{\sufsurf\rbra*{s, \abs*{s} - \abs*{t} + \abs*{t'}}} \geq \abs{t'}, \\
                \presurf\rbra*{s, i}, & \text{otherwise},
            \end{cases}
        \]
        and
        \[
            \sufsurf\rbra*{s\substr{1}{\abs*{s}-1}, i} = \begin{cases}
                t', & i = \abs*{s} - \abs*{t} + \abs*{t'} \text{ and }  \sufsurf\rbra*{s, \abs*{s} - \abs*{t} + \abs*{t'}} = \epsilon, \\
                \sufsurf\rbra*{s, i}, & \text{otherwise}.
            \end{cases}
        \]

        2. For \textup{\texttt{pop\_front}} operations, let $t = \prepal\rbra*{s, 1}$ and $t' = \str\rbra*{\link\rbra*{\node\rbra*{t}}}$. Then
        \[
            \presurf\rbra*{s\substr{2}{\abs*{s}}, i} = \begin{cases}
                t', & i = \abs*{t} - \abs*{t'} \text{ and } \abs*{\presurf\rbra*{s, \abs*{t} - \abs*{t'} + 1}} < \abs{t'}, \\
                \presurf\rbra*{s, i + 1}, & \text{otherwise},
            \end{cases}
        \]
        and
        \[
            \sufsurf\rbra*{s\substr{2}{\abs*{s}}, i} = \begin{cases}
                t', & i = \abs*{t} - 1 \text{ and } \abs*{\presurf\rbra*{s, \abs*{s} - \abs*{t} + 1}} < \abs{t'}, \\
                \epsilon, & i = \abs*{t} - 1 \text{ and } \abs*{\presurf\rbra*{s, \abs*{s} - \abs*{t} + 1}} \geq \abs{t'}, \\
                \sufsurf\rbra*{s, i + 1}, & \text{otherwise}.
            \end{cases}
        \]
    \end{lemma}
    \begin{proof}
        For \textup{\texttt{pop\_back}} operations, it is straightforward due to Lemma \ref{lemma:surface-pop-back}.
        For \textup{\texttt{pop\_front}} operations, the proof is similar because of symmetry. 
    \end{proof}

    The algorithms for \texttt{pop\_back} and \texttt{pop\_front} are given as follows.
    
    \begin{itemize}
        \item \texttt{pop\_back}: 
        \begin{enumerate}
            \item Let $v$ be the node of the longest palindromic suffix $\sufpal\rbra*{s, \abs*{s}}$ of $s$. 
            \item If $\str\rbra*{v}$ is unique in $s$ (checked by Lemma \ref{lemma:unique-by-cnt}), then delete $v$ from the eertree. This will modify $\eertree\rbra*{s}$ to $\eertree\rbra*{s\substr{1}{\abs*{s}-1}}$.
            \item Maintain $\texttt{presurf}$ and $\texttt{sufsurf}$ according to Lemma \ref{lemma:surface-recording-pop-back}.
            \item Maintain $\texttt{cnt}\sbra*{v}$ according to Lemma \ref{lemma:cnt-update}.
        \end{enumerate}
        \item \texttt{pop\_front}: 
        \begin{enumerate}
            \item Let $v$ be the node of the longest palindromic prefix $\prepal\rbra*{s, 1}$ of $s$. 
            \item If $\str\rbra*{v}$ is unique in $s$ (checked by Lemma \ref{lemma:unique-by-cnt}), then delete $v$ from the eertree. This will modify $\eertree\rbra*{s}$ to $\eertree\rbra*{s\substr{2}{\abs*{s}}}$.
            \item Maintain $\texttt{presurf}$ and $\texttt{sufsurf}$ according to Lemma \ref{lemma:surface-recording-pop-back}.
            \item Maintain $\texttt{cnt}\sbra*{v}$ according to Lemma \ref{lemma:cnt-update}.
        \end{enumerate}
    \end{itemize}
    
    For completeness, we provide formal and detailed descriptions of the algorithm for \texttt{push\_back} in Algorithm \ref{algo:surface:pop_back} and Algorithm \ref{algo:surface:pop_front}.
    
    \begin{algorithm}[!htp]
        \caption{$\texttt{pop\_back}\rbra*{}$ via surface recording method}
        \label{algo:surface:pop_back}
        \begin{algorithmic}[1]

        \State $s \gets \data\substr{\stpos}{\edpos-1}$.
        \State $v \gets \texttt{sufsurf}\sbra*{\edpos - 1}$. \Comment{This ensures that $\str\rbra*{v} = \sufpal\rbra*{s, \abs*{s}}$.}
        \If {$\linkcnt\rbra*{s,v} = 0$ and $\texttt{cnt}\sbra*{v} = 1$} \Comment{By Lemma \ref{lemma:unique-by-cnt}, it means that $\str\rbra*{v}$ is unique in $s$.}
            \State Delete $v$ from the eertree.
        \EndIf
        \State $\texttt{ cnt}\sbra*{v} \gets \texttt{cnt}\sbra*{v} - 1$.
        \If {$\len\rbra*{\texttt{sufsurf}\sbra*{\edpos - \len\rbra*{v} + \len\rbra*{\link\rbra*{v}} - 1}} < \len\rbra*{\link\rbra*{v}}$}
            \State $\texttt{sufsurf}\sbra*{\edpos - \len\rbra*{v} + \len\rbra*{\link\rbra*{v}} - 1} \gets \link\rbra*{v}$.
            \State $\texttt{presurf}\sbra*{\edpos - \len\rbra*{v}} \gets \link\rbra*{v}$.
        \Else
            \State $\texttt{presurf}\sbra*{\edpos - \len\rbra*{v}} \gets \even$.
        \EndIf
        \State $\edpos \gets \edpos - 1$.
        \end{algorithmic}
    \end{algorithm}

    \begin{algorithm}[!htp]
        \caption{$\texttt{pop\_front}\rbra*{}$ via surface recording method}
        \label{algo:surface:pop_front}
        \begin{algorithmic}[1]

        \State $s \gets \data\substr{\stpos}{\edpos-1}$.
        \State $v \gets \texttt{presurf}\sbra*{\stpos}$. \Comment{This ensures that $\str\rbra*{v} = \prepal\rbra*{s, 1}$.}
        \If {$\linkcnt\rbra*{s,v} = 0$ and $\texttt{cnt}\sbra*{v} = 1$} \Comment{By Lemma \ref{lemma:unique-by-cnt}, it means that $\str\rbra*{v}$ is unique in $s$.}
            \State Delete $v$ from the eertree.
        \EndIf
        \State  $\texttt{cnt}\sbra*{v} \gets \texttt{cnt}\sbra*{v} - 1$.
        \If {$\len\rbra*{\texttt{presurf}\sbra*{\stpos + \len\rbra*{v} - \len\rbra*{\link\rbra*{v}}}} < \len\rbra*{\link\rbra*{v}}$}
            \State $\texttt{presurf}\sbra*{\stpos + \len\rbra*{v} - \len\rbra*{\link\rbra*{v}}} \gets \link\rbra*{v}$.
            \State $\texttt{sufsurf}\sbra*{\stpos + \len\rbra*{v} - 1} \gets \link\rbra*{v}$.
        \Else
            \State $\texttt{sufsurf}\sbra*{\stpos + \len\rbra*{v} - 1} \gets \even$.
        \EndIf
        \State $\stpos \gets \stpos + 1$.
        \end{algorithmic}
    \end{algorithm}
    
    \subsubsection{Complexity analysis}

    \begin{theorem}
        [Double-ended eertree by surface recording] \label{thm:eertree-surface-recording}
        Double-ended eertrees can be implemented with worst-case time and space complexity per operation $O\rbra*{\log\rbra*{\sigma}}$, where $\sigma$ is the size of the alphabet. More precisely,
        \begin{itemize}
            \item A \textup{\texttt{push\_back}} or \textup{\texttt{push\_front}} operation requires worst-case time and space complexity $O\rbra*{\log\rbra*{\sigma}}$.
            \item A \textup{\texttt{pop\_back}} or \textup{\texttt{pop\_front}} operation requires worst-case time and space complexity $O\rbra*{1}$.
        \end{itemize}
    \end{theorem}
    \begin{proof}
        The correctness has already been proved right after providing the algorithms in Section \ref{sec:surface-recording-push} and Section \ref{sec:surface-recording-pop}. We only analyze the time and space complexity of each deque operation. It is clear that no loops exist in Algorithms \ref{algo:surface:push_back}, \ref{algo:surface:push_front},  \ref{algo:surface:pop_back} and \ref{algo:surface:pop_front}. The only space consumption is the online construction of eertrees in Algorithms \ref{algo:surface:push_back} and \ref{algo:surface:push_front}, which require $O\rbra*{\log\rbra*{\sigma}}$ time and space per operation.
    \end{proof}

    \begin{remark}
        The complexity stated in Theorem \ref{thm:eertree-surface-recording} assumes unlimited storage. 
        After $q$ push operations, the space complexity will be $O\rbra{q \log\rbra{\sigma}}$. 
        In practice, one may wish to keep the space complexity comparable to the length $n$ of the current string.
        To achieve this, we can release the space for the node to be deleted in each pop operation, which requires $O\rbra{\log\rbra{\sigma}}$ time. 
        This approach results in worst-case time complexity $O\rbra{\log\rbra{\sigma}}$ for both push and pop operations. 

        Another way to achieve linear space in practice is to rebuild the whole eertree once the current space usage is over $O\rbra{n\log\rbra{\sigma}}$, while this will turn the worst-case complexity into amortized complexity.
    \end{remark}

\section*{Acknowledgment}

    The authors would like to thank Takuya Mieno for communication regarding the related works \cite{MMSH23,FMN+22}, Oleksandr Kulkov for noting a different implementation of double-ended palindromic tree \cite{Kul24}, and Kohei Morita for communication regarding the Library Checker \cite{Mor20}.
    Qisheng Wang would also like to thank Fran\c{c}ois Le Gall and Zhicheng Zhang for helpful discussions. 

    The work of Qisheng Wang was supported in part by the Engineering and Physical Sciences Research Council under Grant \mbox{EP/X026167/1} and in part by the MEXT Quantum Leap Flagship Program (MEXT Q-LEAP) under Grant \mbox{JPMXS0120319794}. 
    
\addcontentsline{toc}{section}{References}
\bibliographystyle{plainurl}
\bibliography{main}

\appendix

\section{Incremental Construction of Eertrees} \label{app:incremental-eertree}

In Section \ref{sec:framework}, when dealing with \texttt{push\_back} and \texttt{push\_front} operations, we need to obtain $\eertree\rbra*{sc}$ (resp. $\eertree\rbra*{cs}$) from $\eertree\rbra*{s}$ with the help of $\sufpal\rbra*{s, \abs*{s}}$ (resp. $\prepal\rbra*{s, 1}$). 
However, the construction is not explicitly given there. For completeness, we refine the construction of eertrees based on direct links given in \cite{RS18}, and then provide a formal description of these processes. 

In this section, we discuss how to maintain three basic attributes of eertrees: $\next$, $\link$, and $\directlink$. Other basic attributes such as $\len$ and $\linkcnt$ can be maintained correspondingly, and we omit them here. To highlight those attributes under maintenance, we write them in typewriter font: $\texttt{next}$, $\texttt{link}$, and $\texttt{dlink}$; and use squared brackets to indicate indices, e.g., $\texttt{next}\sbra*{u, c}$ and $\texttt{link}\sbra*{v}$. 

In Algorithm \ref{algo:sc-from-s}, we first obtain a node $u$ such that $c \str\rbra*{u} c = \suflen\rbra*{sc, \abs*{sc}}$ efficiently by Proposition \ref{prop:lsp-by-dlink}, and then see if $\node\rbra*{c \str\rbra*{u} c}$ exists or not. If not, we add this new node to $\texttt{next}\sbra*{u, c}$ in the eertree, and maintain necessary information after that. Specifically, we maintain $\link\sbra*{\next\rbra*{u, c}}$ by the definition of suffix link, and $\texttt{dlink}\sbra*{\texttt{next}\sbra*{u, c}, c'}$ for every character $c'$ by Lemma \ref{lemma:dlink}. 

\begin{algorithm}[!htp]
    \caption{Incremental construction of $\eertree\rbra*{sc}$ from $\eertree\rbra*{s}$}
    \label{algo:sc-from-s}
    \begin{algorithmic}[1]
    \Require String $s$, its eertree $\mathcal{T} = \eertree\rbra*{s}$, character $c$, and $t = \sufpal\rbra*{s, \abs*{s}}$.
    \Ensure $\mathcal{T} = \eertree\rbra*{sc}$.

    \State $v \gets \node\rbra*{t}$.
    \State $u \gets \begin{cases}
        v, & s\sbra*{\abs*{s}-\abs*{t}} = c, \\
        \texttt{dlink}\sbra*{v, c}, & \text{otherwise}.
    \end{cases}$
    \If {$\texttt{next}\sbra*{u, c} = \nullptr$}
        \State $\texttt{next}\sbra*{u, c} \gets \node\rbra*{c\str\rbra*{u}c}$. 
        \State $w \gets \begin{cases}
            \texttt{link}\sbra*{u}, & s\sbra*{\abs*{s}-\len\rbra*{\texttt{link}\sbra*{u}}} = c, \\
            \texttt{dlink}\sbra*{\texttt{link}\sbra*{u}, c}, & \text{otherwise}. 
        \end{cases}$
        \State $\texttt{link}\sbra*{\texttt{next}\sbra*{u, c}} \gets \texttt{next}\sbra*{w, c}$. 
        \For {every character $c'$}
        \State $\texttt{dlink}\sbra*{\texttt{next}\sbra*{u, c}, c'} \gets \begin{cases}
            \texttt{next}\sbra*{w, c}, & c' = s\sbra*{\abs*{s}-\len\rbra*{\texttt{next}\sbra*{w, c}}}, \\
            \texttt{dlink}\sbra*{\texttt{next}\sbra*{w, c}, c'}, & \text{otherwise}.
        \end{cases}$
        \EndFor
    \EndIf
    
    \end{algorithmic}
\end{algorithm}

Symmetrically, we also provide how to obtain $\eertree\rbra*{cs}$ from $\eertree\rbra*{s}$ in Algorithm \ref{algo:cs-from-s}.

\begin{algorithm}[!htp]
    \caption{Incremental construction of $\eertree\rbra*{cs}$ from $\eertree\rbra*{s}$}
    \label{algo:cs-from-s}
    \begin{algorithmic}[1]
    \Require String $s$, its eertree $\mathcal{T} = \eertree\rbra*{s}$, character $c$, and $t = \prepal\rbra*{s, 1}$.
    \Ensure $\mathcal{T} = \eertree\rbra*{cs}$.

    \State $v \gets \node\rbra*{t}$.
    \State $u \gets \begin{cases}
        v, & s\sbra*{\abs*{t} + 1} = c, \\
        \texttt{dlink}\sbra*{v, c}, & \text{otherwise}.
    \end{cases}$
    \If {$\texttt{next}\sbra*{u, c} = \nullptr$}
        \State $\texttt{next}\sbra*{u, c} \gets \node\rbra*{c\str\rbra*{u}c}$. 
        \State $w \gets \begin{cases}
            \texttt{link}\sbra*{u}, & s\sbra*{\len\rbra*{\texttt{link}\sbra*{u}} + 1} = c, \\
            \texttt{dlink}\sbra*{\texttt{link}\sbra*{u}, c}, & \text{otherwise}. 
        \end{cases}$
        \State $\texttt{link}\sbra*{\texttt{next}\sbra*{u, c}} \gets \texttt{next}\sbra*{w, c}$. 
        \For {every character $c'$}
        \State $\texttt{dlink}\sbra*{\texttt{next}\sbra*{u, c}, c'} \gets \begin{cases}
            \texttt{next}\sbra*{w, c}, & c' = s\sbra*{\len\rbra*{\texttt{next}\sbra*{w, c}}+1}, \\
            \texttt{dlink}\sbra*{\texttt{next}\sbra*{w, c}, c'}, & \text{otherwise}.
        \end{cases}$
        \EndFor
    \EndIf
    
    \end{algorithmic}
\end{algorithm}

It can be seen that Algorithm \ref{algo:sc-from-s} and Algorithm \ref{algo:cs-from-s} have time complexity $O\rbra*{\sigma}$ per operation due to the enumeration of every character $c'$. This can be improved by noting that there is at most one character $c'$ such that $\texttt{dlink}\sbra*{\texttt{next}\sbra*{u, c}, c'} \neq \texttt{dlink}\sbra*{\texttt{next}\sbra*{w, c}, c'}$. Then, we can store $\texttt{dlink}\sbra*{v, \cdot}$ by a persistent binary search tree for each node $v$. Therefore, the algorithms can be improved to time complexity $O\rbra*{\log\rbra*{\sigma}}$ per operation. 

\section{Occurrence Recording Method} \label{sec:occurrence-recording-app}

In this section, we will provide an algorithm based on occurrence recording that maintains certain valid $S\rbra*{s,v}$ satisfying Eq.~(\ref{eq:sv}) for every $v \in V$ in $O\rbra*{\log\rbra*{\abs*{s}}}$ time for every double-ended queue operation on $s$. In particular, for each operation, there are totally $O\rbra*{1}$ elements modified in all $S\rbra*{s,v}$.
    
    \subsection{\texttt{push\_back} and \texttt{push\_front}} \label{sec:push-occurrence-recording}
    
    For \texttt{push\_back} and \texttt{push\_front} operations, we only focus on how to maintain $S\rbra*{s, v}$, since other auxiliary data are defined by $S\rbra*{s, v}$. Here, we use the variable $\texttt{S}\sbra*{v}$ in our algorithm, which ought to be $S\rbra*{s, v}$ for the current string $s$. The algorithms for \texttt{push\_back} and \texttt{push\_front} are given as follows.
    \begin{itemize}
        \item $\texttt{push\_back}\rbra*{c}$: 
        \begin{enumerate}
            \item Find the node $v$ of the longest palindromic suffix $\sufpal\rbra*{s, \abs*{s}}$ of $s$. \item Modify $\eertree\rbra*{s}$ to $\eertree\rbra*{sc}$ according to $v$ and $c$ by online construction of eertrees \cite{RS18}, with $v'$ being the node of the longest palindromic suffix $\sufpal\rbra*{sc, \abs*{sc}}$ of $sc$. 
            \item Add the start position $\abs*{sc} - \len\rbra*{v'} + 1$ of $\sufpal\rbra*{sc, \abs*{sc}}$ into $\texttt{S}\sbra*{v'}$.
        \end{enumerate}
        \item $\texttt{push\_front}\rbra*{c}$:  \begin{enumerate}
            \item Find the node $v$ of the longest palindromic prefix $\prepal\rbra*{s, 1}$ of $s$. \item Modify $\eertree\rbra*{s}$ to $\eertree\rbra*{cs}$ according to $v$ and $c$ by online construction of eertrees \cite{RS18}, with $v'$ being the node of the longest palindromic prefix $\prepal\rbra*{cs, 1}$ of $cs$. 
            \item Shift right all elements in $\texttt{S}\sbra*{u}$ by $1$ for all nodes $u$.
            \item Add the start position $1$ of $\prepal\rbra*{cs, 1}$ into $\texttt{S}\sbra*{v'}$.
        \end{enumerate}
    \end{itemize}
    
    It can be seen that the algorithms for \texttt{push\_back} and \texttt{push\_front} are almost symmetric, with the only exception that all elements in $\texttt{S}\sbra*{u}$ are shifted right by $1$ for all nodes $u$. This is because all characters in $s$ should move right in order to insert the new character $c$ at the front of $s$. 
    A straightforward implementation of shift-right operations needs to modify all (up to $O\rbra*{\abs*{s}}$) elements in $\texttt{S}\sbra*{u}$ for all node $u$. 
    To make it efficient, we store global indices (see Section \ref{sec:global-index}) rather than relative ones. In this way, a string $s\substr{1}{\abs*{s}}$ is stored and represented by a range $[\stpos, \edpos)$ with global indices $\stpos$ and $\edpos$.
    With the global indices, $\texttt{S}\sbra*{u}$ can be maintained as if no shift-rights were needed. The algorithms for \texttt{push\_back} and \texttt{push\_front} are now restated as follows.
    \begin{itemize}
        \item $\texttt{push\_back}\rbra*{c}$: 
        \begin{enumerate}
            \item Modify $\eertree\rbra*{s}$ to $\eertree\rbra*{sc}$, with $v'$ being the node of $\sufpal\rbra*{sc, \abs*{sc}}$. 
            \item Increment $\edpos$ by $1$. Then, $[\stpos, \edpos)$ indicates the range of string $sc$. 
            \item Add the start global position $\edpos - \len\rbra*{v'}$ of $\sufpal\rbra*{sc, \abs*{sc}}$ into $\texttt{S}\sbra*{v'}$.
        \end{enumerate}
        \item $\texttt{push\_front}\rbra*{c}$: 
        \begin{enumerate}
            \item Modify $\eertree\rbra*{s}$ to $\eertree\rbra*{cs}$, with $v'$ being the node of $\prepal\rbra*{cs, 1}$. 
            \item Decrement $\stpos$ by $1$. Then, $[\stpos, \edpos)$ indicates the range of string $cs$. 
            \item Add the start global position $\stpos$ of $\prepal\rbra*{cs, 1}$ into $\texttt{S}\sbra*{v'}$.
        \end{enumerate}
    \end{itemize}
    
    For completeness, we provide formal and detailed descriptions of the algorithms for \texttt{push\_back} and \texttt{push\_front} in Algorithm \ref{algo:occur:push_back} and Algorithm \ref{algo:occur:push_front}, respectively. 
    
    \begin{algorithm}[!htp]
        \caption{$\texttt{push\_back}\rbra*{c}$ via occurrence recording method}
        \label{algo:occur:push_back}
        \begin{algorithmic}[1]

        \State $s \gets \data\substr{\stpos}{\edpos-1}$.
        \State $v \gets \argmax_{u \in \texttt{sufnode}\sbra*{\edpos - 1}} \len\rbra*{u}$. \Comment{This ensures that $\str\rbra*{v} = \sufpal\rbra*{s, \abs*{s}}$.}
        \State Obtain $\eertree\rbra*{sc}$ from $\eertree\rbra*{s}$ according to $v$ and $c$, and let $v' \gets \node\rbra*{\sufpal\rbra*{sc, \abs*{sc}}}$.
        \State $\data\sbra*{\edpos} \gets c$.
        \State $\edpos \gets \edpos + 1$.
        \State Add $\edpos - \len\rbra*{v'}$ into $\texttt{S}\sbra*{v'}$.
        \State Add $v'$ into $\texttt{prenode}\sbra*{ \edpos - \len\rbra*{v'} }$ and $\texttt{sufnode}\sbra*{ \edpos - 1 }$.
        \end{algorithmic}
    \end{algorithm}
    
    \begin{algorithm}[!htp]
        \caption{$\texttt{push\_front}\rbra*{c}$ via occurrence recording method}
        \label{algo:occur:push_front}
        \begin{algorithmic}[1]

        \State $s \gets \data\substr{\stpos}{\edpos-1}$.
        \State $v \gets \argmax_{u \in \texttt{prenode}\sbra*{\stpos}} \len\rbra*{u}$. \Comment{This ensures that $\str\rbra*{v} = \prepal\rbra*{s, 1}$.}
        \State Obtain $\eertree\rbra*{cs}$ from $\eertree\rbra*{s}$ according to $v$ and $c$, and let $v' \gets \node\rbra*{\prepal\rbra*{cs, 1}}$.
        \State $\stpos \gets \stpos - 1$.
        \State $\data\sbra*{\stpos} \gets c$.
        \State Add $\stpos$ into $\texttt{S}\sbra*{v'}$.
        \State Add $v'$ into $\texttt{prenode}\sbra*{ \stpos }$ and $\texttt{sufnode}\sbra*{ \stpos + \len\rbra*{v'} - 1 }$.
        \end{algorithmic}
    \end{algorithm}
    
    The correctness of the algorithms for \texttt{push\_back} and \texttt{push\_front} is based on the following facts, Lemma \ref{lemma:indirect-recording-push-back} and \ref{lemma:indirect-recording-pop-back}. 
    
    \begin{lemma} [Correctness of \texttt{push\_back} by occurrence recording] \label{lemma:indirect-recording-push-back}
        Let $s$ be a string and $c$ be a character. 
        Suppose $S\rbra*{s,v}$ defined for $\eertree\rbra*{s}$ satisfies Eq.~(\ref{eq:sv}). 
        Let $v' = \node\rbra*{\sufpal\rbra*{sc, \abs*{sc}}}$, and
        \[
            S\rbra*{sc, v} = \begin{cases}
                S\rbra*{s,v} \cup \cbra*{\abs*{sc} - \len\rbra*{v'} + 1}, & v = v', \\
                S\rbra*{s,v}, & \text{otherwise}, 
            \end{cases}
        \]
        where $S\rbra*{s,v} = \emptyset$ if $v \notin \eertree\rbra*{s}$. Then $S\rbra*{sc,v}$ satisfies Eq.~(\ref{eq:sv}) defined for $\eertree\rbra*{sc}$.
    \end{lemma}
    \begin{proof}
        Let $\occur\rbra*{s,v}$ and $\occur\rbra*{s',v}$ be the set of occurrences of $\str\rbra*{v}$ in $s$ and $s' = sc$, respectively. According to Eq.~(\ref{eq:def-occur}), it is straightforward that
        \begin{equation} \label{eq:occur'-by-occur}
            \occur\rbra*{s',v} = \begin{cases}
                \occur\rbra*{s,v} \cup \cbra*{\abs*{sc} - \len\rbra*{v} + 1}, & v \in \link^*\rbra*{v'}, \\
                \occur\rbra*{s,v}, & \text{otherwise}.
            \end{cases}
        \end{equation}
        Since $S\rbra*{s',v} = S\rbra*{s,v}$ and $\occur\rbra*{s',v} = \occur\rbra*{s,v}$ for every $v \notin \link^*\rbra*{v'}$, we only need to show that $S\rbra*{s',v}$ satisfies Eq.~(\ref{eq:sv}) defined for $\eertree\rbra*{sc}$ for every $v \in \link^*\rbra*{v'} \setminus \cbra*{\even, \odd}$. 
        Now we will show that $S\rbra*{s',\link^k\rbra*{v'}}$ satisfies Eq.~(\ref{eq:sv}) defined for $\eertree\rbra*{sc}$ for every $k \geq 0$. 
        That is, for every $k \geq 0$ with $\link^k\rbra*{v'} \notin \cbra*{\even, \odd}$, it holds that
        \begin{equation} \label{eq:sv-for-push-back}
            \occur\rbra*{s',\link^k\rbra*{v'}} = S\rbra*{s',\link^k\rbra*{v'}} \cup \bigcup_{\link\rbra*{u} = \link^k\rbra*{v'}} \rbra*{\occur\rbra*{s',u} \cup \overline{\occur}\rbra*{s',u} }.
        \end{equation}
        The proof is split into two cases.
        
        \textbf{Case 1}. For $k = 0$, $\link^0\rbra*{v'} = v'$. By Eq.~(\ref{eq:occur'-by-occur}), the right hand side of Eq.~(\ref{eq:sv-for-push-back}) becomes
        \begin{align*}
            & ~~~~ S\rbra*{s',v'} \cup 
            \bigcup_{\link\rbra*{u} = v'} \rbra*{\occur\rbra*{s',u} \cup \overline{\occur}\rbra*{s',u}} \\
            & = S\rbra*{s,v'} \cup \cbra*{\abs*{sc} - \len\rbra*{v'} + 1} \cup 
            \bigcup_{\link\rbra*{u} = v'} \rbra*{\occur\rbra*{s,u} \cup 
            \overline{\occur}\rbra*{s,u}
            } \\
            & = \occur\rbra*{s,v'} \cup \cbra*{\abs*{sc} - \len\rbra*{v'} + 1} \\
            & = \occur\rbra*{s',v'}.
        \end{align*}
        
        \textbf{Case 2}. For $k \geq 1$ such that $\link^k\rbra*{v'} \notin \cbra{\even, \odd, v'}$, by Eq.~(\ref{eq:occur'-by-occur}), we have
        \[
            \occur\rbra*{s',\link^{k-1}\rbra*{v'}} = \occur\rbra*{s,\link^{k-1}\rbra*{v'}} \cup \cbra*{\abs*{sc} - \len\rbra*{\link^{k-1}\rbra*{v'}} + 1}.
        \]
        With this, we further have
        \begin{align*}
        \overline{\occur}\rbra*{s', \link^{k-1}\rbra*{v'}}
        & = \set{i + \len\rbra*{\link^{k-1}\rbra*{v'}} - \len\rbra*{\link^k\rbra*{v'}}}{i \in \occur\rbra*{s',\link^{k-1}\rbra*{v'}}} \\
        & = 
        \overline{\occur}\rbra*{s, \link^{k-1}\rbra*{v'}}
        \cup \cbra*{\abs*{sc} - \len\rbra*{\link^{k}\rbra*{v'}} + 1}.
        \end{align*}
        The right hand side of Eq.~(\ref{eq:sv-for-push-back}) becomes
        \begin{align*}
            & ~~~~ S\rbra*{s',\link^k\rbra*{v'}} \cup \bigcup_{\link\rbra*{u} = \link^k\rbra*{v'}} \rbra*{\occur\rbra*{s',u} \cup 
            \overline{\occur}\rbra*{s', u}
            } \\
            & = S\rbra*{s,\link^k\rbra*{v'}} \cup \bigcup_{\substack{\link\rbra*{u} = \link^k\rbra*{v'} \\ u \neq \link^{k-1}\rbra*{v'}}} \rbra*{\occur\rbra*{s,u} \cup 
            \overline{\occur}\rbra*{s, u}
            }
            \\
            & ~~~~
            \cup \occur\rbra*{s',\link^{k-1}\rbra*{v'}}
            \cup  \overline{\occur}\rbra*{s',\link^{k-1}\rbra*{v'}}
            \\
            & = S\rbra*{s,\link^k\rbra*{v'}} \cup \bigcup_{\link\rbra*{u} = \link^k\rbra*{v'}} 
            \rbra*{\occur\rbra*{s,u} \cup \overline{\occur}\rbra*{s, u} }
            \\
            & ~~~~ \cup \cbra*{\abs*{sc} - \len\rbra*{\link^{k-1}\rbra*{v'}} + 1} \cup \cbra*{\abs*{sc} - \len\rbra*{\link^{k}\rbra*{v'}} + 1} \\
            & = \occur\rbra*{s,\link^k\rbra*{v'}} \cup \cbra*{\abs*{sc} - \len\rbra*{\link^{k-1}\rbra*{v'}} + 1} \cup \cbra*{\abs*{sc} - \len\rbra*{\link^{k}\rbra*{v'}} + 1} \\
            & = \occur\rbra*{s',\link^k\rbra*{v'}} \cup \cbra*{\abs*{sc} - \len\rbra*{\link^{k-1}\rbra*{v'}} + 1}.
        \end{align*}
        It remains to show that $\abs*{sc} - \len\rbra*{\link^{k-1}\rbra*{v'}} + 1 \in \occur\rbra*{s',\link^k\rbra*{v'}}$ to complete the proof. This can be seen by Eq.~(\ref{eq:occur'-by-occur}) that $\abs*{sc} - \len\rbra*{\link^{k-1}\rbra*{v'}} + 1 \in \occur\rbra*{s',\link^{k-1}\rbra*{v'}}$, which means that $s'\substr{\abs*{sc} - \len\rbra*{\link^{k-1}\rbra*{v'}} + 1}{\abs{sc}} = \str\rbra*{\link^{k-1}\rbra*{v'}}$. On the other hand, $\str\rbra*{\link^{k}\rbra*{v'}}$ is a proper prefix of $\str\rbra*{\link^{k-1}\rbra*{v'}}$. Therefore, 
        \[
            \str\rbra*{\link^{k}\rbra*{v'}} = s'\substr{\abs*{sc} - \len\rbra*{\link^{k-1}\rbra*{v'}} + 1}{\abs*{sc} - \len\rbra*{\link^{k-1}\rbra*{v'}} + \len\rbra*{\link^{k}\rbra*{v'}}},
        \]
        which means that $\abs*{sc} - \len\rbra*{\link^{k-1}\rbra*{v'}} + 1 \in \occur\rbra*{s',\link^{k}\rbra*{v'}}$. 
        
        As a result, Eq.~(\ref{eq:sv-for-push-back}) is verified, which means that $S\rbra*{s', v}$ satisfies Eq.~(\ref{eq:sv}) defined for $\eertree\rbra*{sc}$.
    \end{proof}
    
    \begin{lemma} [Correctness of \texttt{push\_front} by occurrence recording] \label{lemma:indirect-recording-pop-back}
        Let $s$ be a string and $c$ be a character. 
        Suppose $S\rbra*{s,v}$ defined for $\eertree\rbra*{s}$ satisfies Eq.~(\ref{eq:sv}). 
        Let $v' = \node\rbra*{\prepal\rbra*{cs, 1}}$, and
        \[
            S\rbra*{cs,v} = \begin{cases}
                \rbra*{S\rbra*{s,v} + 1} \cup \cbra*{1}, & v = v', \\
                S\rbra*{s,v} + 1, & \text{otherwise}, 
            \end{cases}
        \]
        where $S\rbra*{s,v} = \emptyset$ if $v \notin \eertree\rbra*{s}$, and $A + a = \set{x + a}{x \in A}$ for a set $A$ and an element $a$. 
        Then $S\rbra*{cs, v}$ satisfies Eq.~(\ref{eq:sv}) defined for $\eertree\rbra*{cs}$.
    \end{lemma}
    
    \begin{proof}
        Because of symmetry, the proof is similar to that of Lemma \ref{lemma:indirect-recording-push-back}.
    \end{proof}
    
    \subsection{\texttt{pop\_back} and \texttt{pop\_front}} \label{sec:pop-occurrence-recording}
    
    Now we consider how to maintain $S\rbra*{s, v}$ for \texttt{pop\_back} and \texttt{pop\_front} operations. Here, we use the variable $\texttt{S}\sbra*{v}$ in our algorithm, which ought to be $S\rbra*{s, v}$ for the current string $s$. The algorithms for \texttt{pop\_back} and \texttt{pop\_front} are given as follows.
    
    \begin{itemize}
        \item \texttt{pop\_back}: 
        \begin{enumerate}
            \item Let $v$ be the node of the longest palindromic suffix $\sufpal\rbra*{s, \abs*{s}}$ of $s$. 
            \item If $\str\rbra*{v}$ is unique in $s$, then delete $v$ from the eertree. This will modify $\eertree\rbra*{s}$ to $\eertree\rbra*{s\substr{1}{\abs*{s}-1}}$.
            \item For every node $u$ in $\sufnode\rbra*{s, \abs*{s}}$, delete the start position $\abs*{s} - \len\rbra*{u} + 1$ of the suffix $\str\rbra*{u}$ of $s$ from $\texttt{S}\sbra*{u}$.
            \item If $\str\rbra*{\link\rbra*{v}}$ is not the empty string, add the start position $\abs*{s} - \len\rbra*{v} + 1$ of $\str\rbra*{\link\rbra*{v}}$ into $S\sbra*{\link\rbra*{v}}$.
        \end{enumerate}
        \item \texttt{pop\_front}: 
        \begin{enumerate}
            \item Let $v$ be the node of the longest palindromic prefix $\prepal\rbra*{s, 1}$ of $s$. 
            \item If $\str\rbra*{v}$ is unique in $s$, then delete $v$ from the eertree. This will modify $\eertree\rbra*{s}$ to $\eertree\rbra*{s\substr{2}{\abs*{s}}}$.
            \item For every node $u$ in $\prenode\rbra*{s, 1}$, delete the start position $1$ of the prefix $\str\rbra*{u}$ of $s$ from $\texttt{S}\sbra*{u}$.
            \item If $\str\rbra*{\link\rbra*{v}}$ is not the empty string, add the start position $\len\rbra*{v} - \len\rbra*{\link\rbra*{v}} + 1$ of $\str\rbra*{\link\rbra*{v}}$ into $S\sbra*{\link\rbra*{v}}$.
            \item Shift left all elements in $\texttt{S}\sbra*{u}$ by $1$ for all nodes $u$.
        \end{enumerate}
    \end{itemize}
    Similar to the algorithms for \texttt{push\_back} and \texttt{push\_front} in Section \ref{sec:push-occurrence-recording}, the algorithms for \texttt{pop\_back} and \texttt{pop\_front} given here are also almost symmetric, with the only exception that all elements in $\texttt{S}\sbra*{u}$ are shifted left by $1$ for all nodes $u$. To avoid modifying every element stored in $\texttt{S}\sbra*{u}$, we use the global indices as in Section \ref{sec:push-occurrence-recording}, and the algorithms for \texttt{pop\_back} and \texttt{pop\_front} are restated as follows.
    \begin{itemize}
        \item \texttt{pop\_back}: 
        \begin{enumerate}
            \item Let $v$ be the node of the longest palindromic suffix $\sufpal\rbra*{s, \abs*{s}}$ of $s$. 
            \item Modify $\eertree\rbra*{s}$ to $\eertree\rbra*{s\substr{1}{\abs*{s}-1}}$.
            \item For every node $u$ in $\sufnode\rbra*{s, \abs*{s}}$, delete the start global position $\edpos - \len\rbra*{u}$ of the suffix $\str\rbra*{u}$ of $s$ from $\texttt{S}\sbra*{u}$.
            \item If $\str\rbra*{\link\rbra*{v}}$ is not the empty string, add the start global position $\edpos - \len\rbra*{v}$ of $\str\rbra*{\link\rbra*{v}}$ into $S\sbra*{\link\rbra*{v}}$.
            \item Decrement $\edpos$ by $1$. Then, $[\stpos, \edpos)$ indicates the range of string $s\substr{1}{\abs*{s}-1}$. 
        \end{enumerate}
        \item \texttt{pop\_front}: 
        \begin{enumerate}
            \item Let $v$ be the node of the longest palindromic prefix $\prepal\rbra*{s, 1}$ of $s$. 
            \item Modify $\eertree\rbra*{s}$ to $\eertree\rbra*{s\substr{2}{\abs*{s}}}$.
            \item For every node $u$ in $\prenode\rbra*{s, 1}$, delete the start global position $\stpos$ of the prefix $\str\rbra*{u}$ of $s$ from $\texttt{S}\sbra*{u}$.
            \item If $\str\rbra*{\link\rbra*{v}}$ is not the empty string, add the start global position $\stpos + \len\rbra*{v} - \len\rbra*{\link\rbra*{v}}$ of $\str\rbra*{\link\rbra*{v}}$ into $S\sbra*{\link\rbra*{v}}$.
            \item Increment $\stpos$ by $1$. Then, $[\stpos, \edpos)$ indicates the range of string $s\substr{2}{\abs*{s}}$. 
        \end{enumerate}
    \end{itemize}
    
    For completeness, we provide formal and detailed descriptions of the algorithms for \texttt{pop\_back} and \texttt{pop\_front} in Algorithm \ref{algo:occur:pop_back} and Algorithm \ref{algo:occur:pop_front}, respectively. 
    
    \begin{algorithm}[!htp]
        \caption{$\texttt{pop\_back}\rbra*{}$ via occurrence recording method}
        \label{algo:occur:pop_back}
        \begin{algorithmic}[1]

        \State $s \gets \data\substr{\stpos}{\edpos-1}$.
        \State $v \gets \argmax_{u \in \texttt{sufnode}\sbra*{\edpos - 1}} \len\rbra*{u}$. \Comment{This ensures that $\str\rbra*{v} = \sufpal\rbra*{s, \abs*{s}}$.}
        \If {$\linkcnt\rbra*{s,v} = 0$ and $\abs*{\texttt{S}\sbra*{v}} = 1$} \Comment{By Lemma \ref{lemma:unique-by-sv}, it means that $\str\rbra*{v}$ is unique in $s$.}
            \State Delete $v$ from the eertree.
        \EndIf
        \For {every $u \in \texttt{sufnode}\sbra*{\edpos - 1}$}
            \State Delete $\edpos - \len\rbra*{u}$ from $\texttt{S}\sbra*{u}$.
            \State Delete $u$ from $\texttt{prenode}\sbra*{\edpos - \len\rbra*{u}}$ and $\texttt{sufnode}\sbra*{\edpos - 1}$.
        \EndFor
        \If {$\len\rbra*{\link\rbra*{v}} \geq 1$}
            \State Add $\edpos - \len\rbra*{v}$ into $\texttt{S}\sbra*{\link\rbra*{v}}$.
            \State Add $\link\rbra*{v}$ into $\texttt{prenode}\sbra*{\edpos - \len\rbra*{v}}$.
            \State Add $\link\rbra*{v}$ into $\texttt{sufnode}\sbra*{\edpos - \len\rbra*{v} + \len\rbra*{\link\rbra*{v}} - 1}$.
        \EndIf
        \State $\edpos \gets \edpos - 1$.
        
        \end{algorithmic}
    \end{algorithm}
    
    \begin{algorithm}[!htp]
        \caption{$\texttt{pop\_front}\rbra*{}$ via occurrence recording method}
        \label{algo:occur:pop_front}
        \begin{algorithmic}[1]

        \State $s \gets \data\substr{\stpos}{\edpos-1}$.
        \State $v \gets \argmax_{u \in \texttt{prenode}\sbra*{\stpos}} \len\rbra*{u}$. \Comment{This ensures that $\str\rbra*{v} = \prepal\rbra*{s, 1}$.}
        \If {$\linkcnt\rbra*{s,v} = 0$ and $\abs*{\texttt{S}\sbra*{v}} = 1$} \Comment{By Lemma \ref{lemma:unique-by-sv}, it means that $\str\rbra*{v}$ is unique in $s$.}
            \State Delete $v$ from the eertree.
        \EndIf
        \For {every $u \in \texttt{prenode}\sbra*{\stpos}$}
            \State Delete $\stpos$ from $\texttt{S}\sbra*{u}$.
            \State Delete $u$ from $\texttt{prenode}\sbra*{\stpos}$ and $\texttt{sufnode}\sbra*{\stpos + \len\rbra*{u} - 1}$.
        \EndFor
        \If {$\len\rbra*{\link\rbra*{v}} \geq 1$}
            \State Add $\stpos + \len\rbra*{v} - \len\rbra*{\link\rbra*{v}}$ into $\texttt{S}\sbra*{\link\rbra*{v}}$.
            \State Add $\link\rbra*{v}$ into $\texttt{prenode}\sbra*{\stpos + \len\rbra*{v} - \len\rbra*{\link\rbra*{v}}}$.
            \State Add $\link\rbra*{v}$ into  $\texttt{sufnode}\sbra*{\stpos + \len\rbra*{v} - 1}$.
        \EndIf
        \State $\stpos \gets \stpos + 1$.
        \end{algorithmic}
    \end{algorithm}
    
     The correctness of the algorithms for \texttt{push\_back} and \texttt{push\_front} is based on the following facts, Lemma \ref{lemma:occurrence-recording-pop-back} and \ref{lemma:occurrence-recording-pop-front}. 
    
    \begin{lemma}[Correctness of \texttt{pop\_back} by occurrence recording] \label{lemma:occurrence-recording-pop-back}
        Let $s$ be a non-empty string. 
        Suppose $S\rbra*{s,v}$ defined for $\eertree\rbra*{s}$ satisfies Eq.~(\ref{eq:sv}). 
        Let $v' = \node\rbra*{\sufpal\rbra*{s, \abs*{s}}}$, and
        \[
            S\rbra*{s',v} = \begin{cases}
                \rbra*{ S\rbra*{s,v} \setminus \cbra*{\abs*{s} - \len\rbra*{v} + 1} } \cup \cbra*{\abs*{s} - \len\rbra*{v'} + 1}, & v = \link\rbra*{v'} \text{ and } \len\rbra*{\link\rbra*{v'}} \geq 1, \\
                S\rbra*{s,v} \setminus \cbra*{\abs*{s} - \len\rbra*{v} + 1}, & \text{otherwise},
            \end{cases}
        \]
        where $s' = s\substr{1}{\abs*{s}-1}$. 
        Then $S\rbra*{s\substr{1}{\abs*{s}-1}, v}$ satisfies Eq.~(\ref{eq:sv}) defined for $\eertree\rbra*{s\substr{1}{\abs*{s}-1}}$.
    \end{lemma}

    \begin{proof}
        Let $\occur\rbra*{s,v}$ and $\occur\rbra*{s',v}$ be the set of occurrences of $\str\rbra*{v}$ in $s$ and $s' = s\substr{1}{\abs*{s}-1}$, respectively. According to Eq.~(\ref{eq:def-occur}), it is straightforward that
        \begin{equation} \label{eq:occur'-by-occur-for-pop-back}
            \occur\rbra*{s',v} = \occur\rbra*{s,v} \setminus \cbra*{\abs*{s} - \len\rbra*{v} + 1} = \begin{cases}
                \occur\rbra*{s,v} \setminus \cbra*{\abs*{s} - \len\rbra*{v} + 1}, & v \in \link^*\rbra*{v'}, \\
                \occur\rbra*{s,v}, & \text{otherwise}. 
            \end{cases}
        \end{equation}
        Since $S\rbra*{s',v} = S\rbra*{s,v}$ and $\occur\rbra*{s',v} = \occur\rbra*{s,v}$ for every $v \notin \link^*\rbra*{v'}$, we only need to show that $S\rbra*{s',v}$ satisfies Eq.~(\ref{eq:sv}) defined for $\eertree\rbra*{s'}$ for every $v \in \link^*\rbra*{v'} \setminus \cbra*{\even, \odd}$. Now we will show that $S\rbra*{s',\link^k\rbra*{v'}}$ satisfies Eq.~(\ref{eq:sv}) defined for $\eertree\rbra*{s'}$ for every $k \geq 0$. That is, for every $k \geq 0$ such that $\link^k\rbra*{v'} \notin \cbra*{\even, \odd}$, it holds that
        \begin{equation} \label{eq:sv-for-pop-back}
            \occur\rbra*{s',\link^k\rbra*{v'}} = S\rbra*{s',\link^k\rbra*{v'}} \cup \bigcup_{\link\rbra*{u} = \link^k\rbra*{v'}} \rbra*{\occur\rbra*{s',u} \cup \overline{\occur}\rbra*{s',u}}.
        \end{equation}
        The proof is split into three cases. 
        
        \textbf{Case 1}. $k = 0$. In this case, $\link^k\rbra*{v'} = \link^0\rbra*{v'} = v'$. Let $u$ be a node such that $\link\rbra*{u} = v'$. By Eq.~(\ref{eq:occur'-by-occur-for-pop-back}), we know that $\abs*{s} - \len\rbra*{u} + 1 \notin \occur\rbra*{s,u} = \occur\rbra*{s',u}$, thus $\abs{s} - \len\rbra*{v'} + 1 \notin 
        \overline{\occur}\rbra*{s',u}
        $. Therefore, the right hand side of Eq.~(\ref{eq:sv-for-pop-back}) becomes
        \begin{align*}
            & ~~~~ S\rbra*{s',v'} \cup \bigcup_{\link\rbra*{u} = v'} \rbra*{\occur\rbra*{s',u} \cup \overline{\occur}\rbra*{s',u}} \\
            & = \rbra*{ S\rbra*{s,v'} \setminus \cbra*{\abs*{s} - \len\rbra*{v'} + 1} } \cup \bigcup_{\link\rbra*{u} = v'} 
            \rbra*{\occur\rbra*{s,u} \cup \overline{\occur}\rbra*{s,u}}
            \\
            & = \rbra*{ S\rbra*{s,v'} \cup \bigcup_{\link\rbra*{u} = v'} \rbra*{\occur\rbra*{s,u} \cup \overline{\occur}\rbra*{s,u}} } \setminus \cbra*{\abs*{s} - \len\rbra*{v'} + 1} \\
            & = \occur\rbra*{s,v'} \setminus \cbra*{\abs*{s} - \len\rbra*{v'} + 1} \\
            & = \occur\rbra*{s',v'}.
        \end{align*}
        
        \textbf{Case 2}. $k = 1$. We only need to consider the case that $\len\rbra*{\link\rbra*{v'}} \geq 1$. We note that
        \[
            S\rbra*{s',\link\rbra*{v'}} = S\rbra*{s,\link\rbra*{v'}} \setminus \cbra*{\abs*{s} - \len\rbra*{\link\rbra*{v'}} + 1} \cup \cbra*{\abs*{s} - \len\rbra*{v'} + 1}.
        \]
        Let $u$ be a node such that $\link\rbra*{u} = \link\rbra*{v'}$. If $u \neq v'$, by Eq.~(\ref{eq:occur'-by-occur-for-pop-back}), we have $\abs*{s} - \len\rbra*{u} + 1 \notin \occur\rbra*{s,u}$, thus $\abs*{s} - \len\rbra*{\link\rbra*{v'}} + 1 \notin \set{i + \len\rbra*{u} - \len\rbra*{\link\rbra*{v'}}}{i \in \occur\rbra*{s,u}} = \overline{\occur}\rbra*{s, u}$. 
        By Eq.~(\ref{eq:def-occur}), $\abs*{s} - \len\rbra*{\link\rbra*{v'}} + 1 \notin \occur\rbra*{s,u}$ because $\abs*{s} - \len\rbra*{\link\rbra*{v'}} + 1 = \abs*{s} - \len\rbra*{\link\rbra*{u}} + 1 > \abs*{s} - \len\rbra*{u} + 1$. Note that these also hold for $u = v'$.
        
        The right hand side of Eq.~(\ref{eq:sv-for-pop-back}) becomes
        \begin{align*}
            & ~~~~ S\rbra*{s',\link\rbra*{v'}} \cup \bigcup_{\link\rbra*{u} = \link\rbra*{v'}} \rbra*{\occur\rbra*{s',u} \cup
            \overline{\occur}\rbra*{s',u}
            } \\
            & = \rbra*{ S\rbra*{s,\link\rbra*{v'}} \setminus \cbra*{\abs*{s} - \len\rbra*{\link\rbra*{v'}} + 1} \cup \cbra*{\abs*{s} - \len\rbra*{v'} + 1} }  \\
            & ~~~~ \cup \bigcup_{\substack{\link\rbra*{u} = \link\rbra*{v'} \\ u \neq v'}} \rbra*{\occur\rbra*{s,u} \cup \overline{\occur}\rbra*{s,u}} \\
            & ~~~~ \cup \rbra*{\occur\rbra*{s,v'} \setminus \cbra*{\abs*{s} - \len\rbra*{v'} + 1}} \cup \rbra*{
            \overline{\occur}\rbra*{s,v'}
            \setminus \cbra*{\abs*{s} - \len\rbra*{\link\rbra*{v'}} + 1}} \\
            & = \rbra*{ S\rbra*{s,\link\rbra*{v'}} \cup \bigcup_{\link\rbra*{u} = \link\rbra{v'}} \rbra*{\occur\rbra*{s,u} \cup \overline{\occur}\rbra*{s,u}} } \setminus \cbra*{\abs*{s} - \len\rbra*{\link\rbra*{v'}} + 1} \\
            & = \occur\rbra*{s,\link\rbra*{v'}} \setminus \cbra*{\abs*{s} - \len\rbra*{\link\rbra*{v'}} + 1} \\
            & = \occur\rbra*{s',\link\rbra*{v'}}.
        \end{align*}
        
        \textbf{Case 3}. $k \geq 2$. We only need to consider the case that $\len\rbra*{\link^k\rbra*{v'}} \geq 1$. Let $u$ be a node such that $\link\rbra*{u} = \link^k\rbra*{v'}$.  The right hand side of Eq.~(\ref{eq:sv-for-pop-back}) becomes
        \[
            S\rbra*{s',\link^k\rbra*{v'}} \cup \bigcup_{\link\rbra*{u} = \link^{k}\rbra*{v'}} \rbra*{\occur\rbra*{s',u} \cup 
            \overline{\occur}\rbra*{s',u}
            }
            = A_1 \cup A_2 \cup A_3 \cup A_4,
        \]
        where
        \[
            A_1 =  S\rbra*{s,\link^k\rbra*{v'}} \setminus  \cbra*{ \abs*{s} - \len\rbra*{\link^k\rbra*{v'}} + 1 } ,
        \]
        \[
            A_2 = \bigcup_{\substack{\link\rbra*{u} = \link^k\rbra*{v'} \\ u \neq \link^{k-1}\rbra*{v'}}} \rbra*{ 
            \occur\rbra*{s,u} \cup 
            \overline{\occur}\rbra*{s,u}
            },
        \]
        \[
            A_3 = \occur\rbra*{s,\link^{k-1}\rbra*{v'}} \setminus \cbra*{\abs*{s} - \len\rbra*{\link^{k-1}\rbra*{v'}} + 1},
        \]
        \[
            A_4 =  
            \overline{\occur}\rbra*{s,\link^{k-1}\rbra*{v'}}
            \setminus \cbra*{\abs*{s} - \len\rbra*{\link^k\rbra*{v'}} + 1} .
        \]
        In the sets $A_1$ and $A_4$, the element $\abs*{s} - \len\rbra*{\link^k\rbra*{v'}} + 1$ is removed. Next, we will show that this special element is not in either $A_2$ or $A_3$.
        
        By Lemma \ref{lemma:tail-notin-A2cupA3} (which will be shown later), we have
        \begin{equation} \label{eq:A1cupA2cupA3cupA3}
            A_1 \cup A_2 \cup A_3 \cup A_4 = \rbra*{ S\rbra*{s,\link^{k}\rbra*{v'}} \cup A_2 \cup A_3 \cup A_4' } \setminus \cbra*{\abs*{s} - \len\rbra*{\link^k\rbra*{v'}} + 1},
        \end{equation}
        where $A_4' =
            \overline{\occur}\rbra*{s,\link^{k-1}\rbra*{v'}}
            $.
        
        There is another special element $\abs*{s} - \len\rbra*{\link^{k-1}\rbra*{v'}} + 1$ being removed in $A_3$. Next, we will show that this special element is in either $A_2$ or $A_4'$. 
        
        By Lemma \ref{lemma:sp-in-A2cupA4} (which will be shown later) and Eq.~(\ref{eq:A1cupA2cupA3cupA3}), the right hand side of Eq.~(\ref{eq:sv-for-pop-back}) becomes
        \begin{align*}
            & ~~~~ A_1 \cup A_2 \cup A_3 \cup A_4 \\
            & = \rbra*{ S\rbra*{s,\link^{k}\rbra*{v'}} \cup A_2 \cup \occur\rbra*{s,\link^{k-1}\rbra*{v'}} \cup A_4' } \setminus \cbra*{\abs*{s} - \len\rbra*{\link^k\rbra*{v'}} + 1} \\
            & = \rbra*{ S\rbra*{s,\link^k\rbra*{v'}} \cup \bigcup_{\link\rbra*{u} = \link^k\rbra*{v'}} \rbra*{\occur\rbra*{s,u} \cup 
            \overline{\occur}\rbra*{s,u}
            } }
            \setminus \cbra*{\abs*{s} - \len\rbra*{\link^k\rbra*{v'}} + 1} \\
            & = \occur\rbra*{s,\link^k\rbra*{v'}} \setminus \cbra*{\abs*{s} - \len\rbra*{\link^k\rbra*{v'}} + 1} \\
            & = \occur\rbra*{s',\link^k\rbra*{v'}}.
        \end{align*}
    \end{proof}
    
    \begin{lemma}[Correctness of \texttt{pop\_front} by occurrence recording] \label{lemma:occurrence-recording-pop-front}
        Let $s$ be a non-empty string. 
        Suppose $S\rbra*{s,v}$ defined for $\eertree\rbra*{s}$ satisfies Eq.~(\ref{eq:sv}). 
        Let $v' = \node\rbra*{\prepal\rbra*{s, 1}}$, and
        \[
            S\rbra*{s',v} = \begin{cases}
                \rbra*{ \rbra*{ S\rbra*{s,v} \setminus \cbra*{1} } \cup \cbra*{\len\rbra*{v'} - \len\rbra*{\link\rbra*{v'}} + 1} } - 1, & v = \link\rbra*{v'} \text{ and } \len\rbra*{\link\rbra*{v'}} \geq 1, \\
                \rbra*{S\rbra*{s,v} \setminus \cbra*{1}} - 1, & \text{otherwise}, 
            \end{cases}
        \]
        where $s' = s\substr{2}{\abs*{s}}$, and $A - a = \set{x - a}{x \in A}$ for a set $A$ and an element $a$.
        Then $S\rbra*{s\substr{2}{\abs*{s}},v}$ satisfies Eq.~(\ref{eq:sv}) defined for $\eertree\rbra*{s\substr{2}{\abs*{s}}}$.
    \end{lemma}
    \begin{proof}
        Because of symmetry, the proof is similar to that of Lemma \ref{lemma:occurrence-recording-pop-back}.
    \end{proof}

    \subsection{Complexity analysis}
    
    \begin{theorem} [Double-ended eertree by occurrence recording]
        Double-ended eertrees can be implemented with amortized time complexity per operation $O\rbra*{\log\rbra*{n}+\log\rbra*{\sigma}}$ and worst-case space complexity per operation $O\rbra*{\log\rbra*{\sigma}}$, where $\sigma$ is the size of the alphabet and $n$ is the length of the current string. More precisely,
        \begin{itemize}
            \item A \textup{\texttt{push\_back}} or \textup{\texttt{push\_front}} operation requires worst-case time complexity $O\rbra*{\log\rbra*{n}+\log\rbra*{\sigma}}$ and worst-case space complexity $O\rbra*{\log\rbra*{\sigma}}$.
            \item A \textup{\texttt{pop\_back}} or \textup{\texttt{pop\_front}} operation requires amortized time complexity $O\rbra*{\log\rbra*{n}}$ and worst-case space complexity $O\rbra*{1}$.
        \end{itemize}
    \end{theorem}
    
    \begin{proof}
        The correctness has already been proved right after providing the algorithms in Section \ref{sec:push-occurrence-recording} and Section \ref{sec:pop-occurrence-recording}. Next, we analyze the time and space complexity of each of the deque operations. For convenience, we use balanced binary search trees to store \texttt{S}, \texttt{prenode} and \texttt{sufnode}. 
    
        For one query of \texttt{push\_back} and \texttt{push\_front} operations, it is clear that no loop exists in their implementations (see Algorithm \ref{algo:occur:push_back} and Algorithm \ref{algo:occur:push_front}). By the online construction of eertree in \cite{RS18}, we need $O\rbra*{\log\rbra*{\sigma}}$ time and space in the worst case to convert $\eertree\rbra*{s}$ to $\eertree\rbra*{sc}$ or $\eertree\rbra*{cs}$. Also, there is only one insertion operation required on each of \texttt{S}, \texttt{prenode} and \texttt{sufnode}, which can be done in $O\rbra*{\log\rbra*{n}}$ time and $O\rbra*{1}$ space by binary tree insertions in the worst case. Therefore, a \textup{\texttt{push\_back}} or \textup{\texttt{push\_front}} operation requires worst-case time complexity $O\rbra*{\log\rbra*{n}+\log\rbra*{\sigma}}$ and worst-case space complexity $O\rbra*{\log\rbra*{\sigma}}$.
        
        For one query of \texttt{pop\_back} (resp. \texttt{pop\_front}) operations, the only loop is to delete the start position of $\str\rbra*{u}$ for every node $u$ in $\texttt{sufnode}\sbra*{\edpos-1}$ (resp. $\texttt{prenode}\sbra*{\stpos}$). According to the correctness of the algorithm, we conclude that $\texttt{S}\sbra*{u} = S\rbra*{s, u}$, $\texttt{sufnode}\sbra*{\edpos-1} = \sufnode\rbra*{s, \abs*{s}}$ and $\texttt{prenode}\sbra*{\stpos} = \prenode\rbra*{s, 1}$ at this moment. Next, we only consider the case of \texttt{pop\_back} because of symmetry. Note that for every node $u$ in $\sufnode\rbra*{s, \abs*{s}}$, the start position of the prefix $\str\rbra*{u}$ of $s$ must be in $S\rbra*{s, u}$. This follows that for every node $u$ in $\texttt{sufnode}\sbra*{\edpos-1}$, its start position $\edpos - \len\rbra*{u}$ is in $\texttt{S}\sbra*{u}$ before deleting it. Now we consider the total number of elements to be deleted. Note that every element to be deleted must have been added before. For all the four deque operations, there will be at most one element added per operation. Therefore, there will be amortized $O\rbra*{1}$ elements to be deleted per operation. So the time complexity per operation is amortized $O\rbra*{\log\rbra*{n}}$ and the space complexity is $O\rbra*{1}$ per operation in the worst case.
    \end{proof}

    \subsection{Technical lemmas}

    \begin{lemma} \label{lemma:tail-notin-A2cupA3}
        In the proof of Lemma \ref{lemma:occurrence-recording-pop-back}, if $\len\rbra*{\link^k\rbra*{v'}} \geq 1$, then $\abs*{s} - \len\rbra*{\link^k\rbra*{v'}} + 1 \notin A_2 \cup A_3$. 
        \end{lemma}
        \begin{proof}
            Let $u$ be a node such that $\link\rbra*{u} = \link^k\rbra*{v'}$, and then $\len\rbra*{u} > \len\rbra*{\link^k\rbra*{v'}}$. Then for every $i \in \occur\rbra*{u}$, we have $1 \leq i \leq \abs*{s} - \len\rbra*{u} + 1 < \abs*{s} - \len\rbra*{\link^k\rbra*{v'}} + 1$, which means $\abs*{s} - \len\rbra*{\link^k\rbra*{v'}} + 1 \notin \occur\rbra*{u}$ and therefore $\abs*{s} - \len\rbra*{\link^k\rbra*{v'}} + 1 \notin A_3$. 
            
            If $u \neq \link^{k-1}\rbra*{v'}$, by Eq.~(\ref{eq:occur'-by-occur-for-pop-back}), we know that $\abs*{s} - \len\rbra*{u} + 1 \notin \occur\rbra*{s,u}$. This immediately leads to
            \[
                \abs*{s} - \len\rbra*{\link^k\rbra*{v'}} + 1 \notin \set{i + \len\rbra*{u} - \len\rbra*{\link^k\rbra*{v'}}}{i \in \occur \rbra*{s,u}} = \overline{\occur}\rbra*{s, u}.
            \]
            This implies that $\abs*{s} - \len\rbra*{\link^k\rbra*{v'}} + 1 \notin A_2$.
        \end{proof}

        \begin{lemma} \label{lemma:sp-in-A2cupA4}
            In the proof of Lemma \ref{lemma:occurrence-recording-pop-back}, if $\len\rbra*{\link^k\rbra*{v'}} \geq 1$, then $\abs*{s} - \len\rbra*{\link^{k-1}\rbra*{v'}} + 1 \in A_2 \cup A_4'$.
        \end{lemma}

        To prove Lemma \ref{lemma:sp-in-A2cupA4}, we need some properties of eertrees.
        In the following, we show that the lengths of any node and its ancestors in the link tree are convex.
    \begin{lemma} \label{lemma:sub-arithmetic}
        Suppose $s$ is a string and $v$ is a node in $\eertree\rbra*{s}$. Let $v_1 = \link\rbra*{v}$ and $v_2 = \link\rbra*{v_1}$. If $\len\rbra*{v_2} \geq 1$, then
        $\len\rbra*{v} + \len\rbra*{v_2} \geq 2\len\rbra*{v_1}$.
    \end{lemma}
    \begin{proof}
        
        Note that $\str\rbra*{v}\substr{1}{\len\rbra*{v_1}} = \str\rbra*{v_1}$ is the longest palindromic proper prefix of $\str\rbra*{v}$.
        By Lemma \ref{lemma:period-of-pal}, $p = \len\rbra*{v} - \len\rbra*{v_1}$ is the minimal period of $\str\rbra*{v}$. Similarly, $p_1 = \len\rbra*{v_1} - \len\rbra*{v_2}$ is the minimal period of $\str\rbra*{v_1}$. Since every period of $\str\rbra*{v}$ is a period of $\str\rbra*{v_1}$, we have $\len\rbra*{v} - \len\rbra*{v_1} = p \geq p_1 = \len\rbra*{v_1} - \len\rbra*{v_2}$, i.e., $\len\rbra*{v} + \len\rbra*{v_2} \geq 2\len\rbra*{v_1}$.
    \end{proof}
    
    By Lemma \ref{lemma:sub-arithmetic}, we furthermore derive properties of complements of occurrences that will be useful in analyzing our algorithms.
    
    \begin{lemma} \label{lemma:overline-occur}
        Suppose $s$ is a string and $v$ is a node in $\eertree\rbra*{s}$. Let $i \in \occur\rbra*{s,v}$, $v_1 = \link\rbra*{v}$ and $v_2 = \link\rbra*{v_1}$. If $\len\rbra*{v_2} \geq 1$, then \[
        i + \len\rbra*{v} - \len\rbra*{v_1} \in \bigcup_{\link\rbra*{u} = v_2} \overline{\occur}\rbra*{s,u}.
        \]
    \end{lemma}
    \begin{proof}
        We first show that 
        \[
            s' = s\substr{i+\len\rbra*{v_1}-\len\rbra*{v_2}}{i+\len\rbra*{v}-\len\rbra*{v_1}+\len\rbra*{v_2}-1}
        \]
        is a palindrome. To see this, we note that $s'$ has the same center as $s\substr{i}{i+\len\rbra*{v}-1} = \str\rbra*{v}$. It remains to show that $s' \neq \epsilon$. By Lemma \ref{lemma:sub-arithmetic}, we have $\len\rbra*{v_2} \geq 2\len\rbra*{v_1}-\len\rbra*{v}$. Then,
        \[
            \abs*{s'} = \len\rbra*{v} - 2\len\rbra*{v_1} + 2 \len\rbra*{v_2} \geq \len\rbra*{v_2} \geq 1.
        \]

        Then, we have the following two cases.

        \textbf{Case 1}. $\abs{s'} = \len\rbra{v_2}$. In this case, $\len\rbra*{v_2} = 2\len\rbra*{v_1} - \len\rbra*{v} \geq 1$, and we know that $s' = \str\rbra*{v_2}$ because $s'$ is a proper prefix of $s\substr{i+\len\rbra*{v}-\len\rbra*{v_1}}{i+\len\rbra*{v}-1} = \str\rbra*{v_1}$. Since $i \in \occur\rbra*{s,v}$, by Lemma \ref{lemma:occur-link}, we have $i \in \occur\rbra*{s,v_1}$, and thus $i + \len\rbra*{v} - \len\rbra*{v_1} = i + \len\rbra*{v_1} - \len\rbra*{v_2} \in \overline{\occur}\rbra*{s,v_1}$.

        \textbf{Case 2}. $\abs{s'} > \len\rbra{v_2}$. Let $u = \node\rbra*{s'}$. Note that $s\substr{i+\len\rbra*{v_1}-\len\rbra*{v_2}}{i+\len\rbra*{v_1}-1} = \str\rbra*{v_2}$ is a proper suffix of $s\substr{i}{i+\len\rbra*{v_1}-1} = \str\rbra*{v_1}$, and also a proper prefix of $s'$. From this, we know that there exists an integer $k \geq 1$ such that $\link^k\rbra*{u} = v_2$. Let $u' = \link^{k-1}\rbra*{u}$, then $\link\rbra*{u'} = v_2$. Since $i + \len\rbra*{v_1} - \len\rbra*{v_2} \in \occur\rbra*{s,u}$, by Lemma \ref{lemma:occur-link}, we have
        \[
            i + \len\rbra*{v_1} - \len\rbra*{v_2} + \len\rbra*{u} - \len\rbra*{u'} \in \occur\rbra*{s,u'}.
        \]
        Thus $i + \len\rbra*{v_1} - \len\rbra*{v_2} + \len\rbra*{u} - \len\rbra*{u'} + \len\rbra*{u'} - \len\rbra*{v_2} = i + \len\rbra*{v} - \len\rbra*{v_1} \in \overline{\occur}\rbra*{s,u'}$.
    \end{proof}

        Now we are ready to prove Lemma \ref{lemma:sp-in-A2cupA4}.
        
        \begin{proof} [Proof of Lemma \ref{lemma:sp-in-A2cupA4}]
            Let $v = \link^{k-2}\rbra*{v'}$, and $i = \abs*{s} - \len\rbra*{v} + 1 \in \occur\rbra*{s,v}$. 
            By Lemma \ref{lemma:overline-occur}, we have
            \[
                i + \len\rbra*{v} - \len\rbra*{\link\rbra*{v}} \in \bigcup_{\link\rbra*{u} = \link^2\rbra*{v}} \overline{\occur}\rbra*{s,u},
            \]
            which follows that
            \[
                \abs*{s} - \len\rbra*{\link^{k-1}\rbra*{v'}} + 1 \in \bigcup_{\link\rbra*{u} = \link^k\rbra*{v'}} \overline{\occur}\rbra*{s,u} \subseteq A_2 \cup A_4'.
            \]
        \end{proof}

\section{Applications} \label{sec:app-app}
    
    As an application, we apply our double-ended eertree in several computational tasks (see Table \ref{tab:app} for an overview). 
    
    \begin{table}[htp]
    \centering
    \begin{threeparttable}
    \caption{Applications of double-ended eertrees.}
    \label{tab:app}
    \begin{tabular}{ccc}
    \toprule
    Computational Task  & Our Method & Known Methods                                        \\ \midrule
    \begin{tabular}[c]{@{}c@{}} \textsc{Counting Distinct} \\ \textsc{Palindromic Substrings} \end{tabular} & $O\rbra*{n\sqrt{q}}$ ${}^*$ & \begin{tabular}[c]{@{}c@{}} $O\rbra*{nq}$ ${}^\dag$ \\ $O\rbra*{\rbra*{n+q}\log\rbra*{n}}$ \cite{RS17} \end{tabular} \\ \midrule
    \begin{tabular}[c]{@{}c@{}} \textsc{Longest} \\ \textsc{Palindromic Substring} \end{tabular}  & $O\rbra*{n\sqrt{q}}$ ${}^*$  & \begin{tabular}[c]{@{}c@{}} $O\rbra*{nq}$ ${}^\dag$ \\ $O\rbra*{n \log^2\rbra*{n} + q \log\rbra*{n}}$ \cite{ACPR20} \end{tabular} \\ \midrule
    \begin{tabular}[c]{@{}c@{}} \textsc{Shortest Unique} \\ \textsc{Palindromic Substring} \end{tabular}     & $O\rbra*{n\sqrt{q}\log\rbra*{n}}$ ${}^*$  & $O\rbra*{nq}$ ${}^\dag$                                           \\ \midrule
    \textsc{Shortest Absent Palindrome}      & $O\rbra*{n\sqrt{q} + q\log\rbra*{n}}$ ${}^*$  & $O\rbra*{nq+q\log\rbra*{n}}$ ${}^\dag$                                          \\ \midrule
    \begin{tabular}[c]{@{}c@{}} \textsc{Counting Rich Strings} \\ \textsc{with Given Word} \end{tabular} & $O\rbra*{n\sigma + k\sigma^k}$   & $O\rbra*{k\sigma^k\rbra*{n+k}}$ ${}^\dag$                                             \\ \bottomrule
    \end{tabular}
    \begin{tablenotes}
      \small
      \item ${}^*$ Complexity for offline queries.
      \item ${}^\dag$ These are straightforward algorithms equipped with eertrees \cite{RS18} or other algorithms concerning palindromes (e.g., \cite{Man75,KMP77,GPR10}) as a subroutine.
    \end{tablenotes}
    \end{threeparttable}
    \end{table}
    
    \paragraph{Range queries concerning palindromes.} We studied online and offline range queries concerning palindromes on a string $s\substr{1}{n}$ of length $n$. Each query is of the form $\rbra*{l, r}$ and asks problems of different types on substring $s\substr{l}{r}$:
    \begin{itemize}
        \item \textsc{Counting Distinct Palindromic Substrings}: Find the number of distinct palindromic substring of $s\substr{l}{r}$. 
        \item \textsc{Longest Palindromic Substring}: Find the longest palindromic substring of $s\substr{l}{r}$. 
        \item \textsc{Shortest Unique Palindromic Substring}: Find the shortest unique palindromic substring of $s\substr{l}{r}$. 
        \item \textsc{Shortest Absent Palindrome}: Find the shortest absent palindrome of $s\substr{l}{r}$. 
    \end{itemize}
    
    \begin{corollary} [Corollary \ref{corollary:count-pal}, \ref{corollary:longest-pal}, \ref{corollary:sups}, \ref{corollary:sap} and \ref{corollary:online} combined] \label{corollary:range-queries}
        Online and offline range queries of type \textsc{Counting Distinct Palindromic Substrings}, \textsc{Longest Palindromic Substring}, \textsc{Shortest Unique Palindromic Substring} and \textsc{Shortest Absent Palindrome} can be answered with total time and space complexity $\tilde O\rbra*{ n \sqrt{q} }$\footnote{$\tilde O\rbra*{\cdot}$ suppresses polylogarithmic factors of $n$, $q$ and $\sigma$.}, where $n$ is the length of the string and $q \leq n^2$ is the number of queries.
        
        Moreover, the space complexity of offline queries can be reduced to $O\rbra*{n}$ when $\sigma = O\rbra*{1}$. 
    \end{corollary}
    
    The reduced space complexity for offline queries in Corollary \ref{corollary:range-queries} is significant, since the size $\sigma$ of the alphabet is usually a constant in practice as already mentioned above.
    
    As shown in Table \ref{tab:app}, double-ended eertrees bring speedups to all mentioned computational tasks with the only exceptions\footnote{The concurrent work of Mitani, Mieno, Seto, and Horiyama \cite{MMSH23} appeared on arXiv on the same day of this paper. They showed that range queries of \textsc{Longest Palindromic Substring} can be solved in linear time $O\rbra*{n+q}$ assuming $\sigma = O\rbra*{1}$.} that
    \begin{enumerate}
        \item It was shown in \cite{RS17} that range queries of \textsc{Counting Distinct Palindromic Substrings} can be answered in $O\rbra*{\rbra*{n+q}\log\rbra*{n}}$ time assuming $\sigma = O\rbra*{1}$. Our method can be still faster than their method when $q = o\rbra*{\log^2\rbra*{n}}$ or $q = \omega\rbra*{n^2/\log^2\rbra*{n}}$.
        \item It was shown in \cite{ACPR20} that range queries of \textsc{Longest Palindromic Substring} can be answered in $O\rbra*{n \log^2\rbra*{n} + q \log\rbra*{n}}$ time assuming $\sigma = O\rbra*{1}$. Our method can be still faster than their method when $q = o\rbra*{\log^4\rbra*{n}}$ or $q = \omega\rbra*{n^2/\log^2\rbra*{n}}$.
    \end{enumerate}
    
    \paragraph{Enumerating rich strings with given word.} 
    
    Palindromic rich strings have been extensively studied \cite{GJWZ09,RR09,BDLGZ09,Ves14}. Recently, the number of rich strings of length $n$ was studied \cite{RS18,GSS16,Ruk17}. 
    Using double-ended eertree, we give an algorithm for \textsc{Counting Rich Strings with Given Word} (see Problem \ref{prob:enum-rich} for the formal statement).
    
    \begin{corollary} [Corollary \ref{corollary:count-rich-string-with-given-word} restated]
        There is an algorithm for \textsc{Counting Rich Strings with Given Word}, which computes the number of palindromic rich strings of length $n+k$ with a given word of length $n$ with time complexity $O\rbra*{n\sigma + k\sigma^k}$, where $\sigma$ is the size of the alphabet.
    \end{corollary}
    
    By contrast, a na\"{i}ve algorithm that enumerates all (roughly $k\sigma^k$ in total) possible candidates and then checks each of them by \cite{GPR10} in $O\rbra*{n+k}$ time would have time complexity $O\rbra*{k\sigma^k\rbra*{n+k}}$. The strength of our algorithm is that the parameters $n$ and $\sigma^k$ in the complexity are additive, while they are multiplicative in the na\"{i}ve algorithm.

    \subsection{Range queries concerning palindromes on a string}
    
    In this subsection, we aim to design a framework for range queries on a string concerning problems about palindromes. 
    
    Suppose a string $s\substr{1}{n}$ of length $n = \abs*{s}$ is given. We consider range queries on any substrings $s\substr{l}{r}$ of $s$, where $1 \leq l \leq r \leq n$. A query $\rbra*{l, r}$ is to find
    \begin{itemize}
        \item the number of distinct palindromic substrings,
        \item the longest palindromic substring,
        \item the shortest unique palindromic substring,
        \item the shortest absent palindrome
    \end{itemize}
    of substring $s\substr{l}{r}$ of $s$. We formally state the four types of queries as follows. 
    
    \begin{problem} [\textsc{Counting Distinct Palindromic Substrings}]
        Given a string $s$ of length $n$, for each query of the form $\rbra*{l, r}$, count the number of distinct palindromes over all $s\substr{i}{j}$ for $l \leq i \leq j \leq r$. 
    \end{problem}
    
    \begin{problem} [\textsc{Longest Palindromic Substring}]
        Given a string $s$ of length $n$, for each query of the form $\rbra*{l, r}$, find the longest palindromic substring $s\substr{i}{j}$ over $l \leq i \leq j \leq r$. If there are multiple solutions, find any of them. 
    \end{problem}
    
    \begin{problem} [\textsc{Shortest Unique Palindromic Substring}]
        Given a string $s$ of length $n$, for each query of the form $\rbra*{l, r}$, find the shortest palindromic substring $s\substr{i}{j}$ over $l \leq i \leq j \leq r$ that occurs exactly once in $s\substr{l}{r}$. If there are multiple solutions, find any of them. 
    \end{problem}
    
    \begin{problem} [\textsc{Shortest Absent Palindrome}]
        Given a string $s$ of length $n$, for each query of the form $\rbra*{l, r}$, find the shortest palindrome $t$ that is not a substring of $s\substr{l}{r}$. If there are multiple solutions, find any of them. 
    \end{problem}
    
    Now suppose we have $q \leq n^2$ queries $\rbra*{l_i, r_i}$ with $1 \leq l_i \leq r_i \leq n$ for $1 \leq i \leq q$. 
    In the following, we will consider to answer these queries in the offline and online cases, respectively.
    
    \subsubsection{Offline queries}
    
    To answer offline queries efficiently, we adopt the trick in Mo's algorithm (cf. \cite{DKPW20}). The basic idea is to maintain a double-ended eertree to iterate over the eertree of each $s\substr{l_i}{r_i}$ for all $1 \leq i \leq q$. Let $B$ be a parameter to be determined. 
    We sort all queries by $\floor*{\rbra*{l_i-1}/B}$ and in case of a tie by $r_i$ (both in increasing order). Let $\mathcal{T}$ be a double-ended eertree of $s\substr{l_1}{r_1}$ which can be constructed in $O\rbra*{\abs*{l_1-r_1+1} \log\rbra*{\sigma}}$ time. Now we will iterate all eertrees needed as follows.
    \begin{itemize}
        \item For $2 \leq i \leq q$ in increasing order,
        \begin{enumerate}
            \item Let $l \gets l_{i-1}$ and $r \gets r_{i-1}$ indicate that $\mathcal{T}$ is the eertree of $s\substr{l}{r}$ currently.
            \item Repeat the following as long as $r < r_i$:
            \begin{itemize}
                \item Set $r \gets r + 1$, and then perform $\texttt{push\_back}\rbra*{s\sbra*{r}}$ on $\mathcal{T}$.
            \end{itemize}
            \item Repeat the following as long as $l > l_i$:
            \begin{itemize}
                \item Set $l \gets l - 1$, and then perform $\texttt{push\_front}\rbra*{s\sbra*{l}}$ on $\mathcal{T}$.
            \end{itemize}
            \item Repeat the following as long as $r > r_i$:
            \begin{itemize}
                \item Set $r \gets r - 1$, and then perform $\texttt{pop\_back}\rbra*{}$ on $\mathcal{T}$.
            \end{itemize}
            \item Repeat the following as long as $l < l_i$:
            \begin{itemize}
                \item Set $l \gets l - 1$, and then perform $\texttt{pop\_front}\rbra*{}$ on $\mathcal{T}$.
            \end{itemize}
            \item Now $\mathcal{T}$ stores the eertree of $s\substr{l_i}{r_i}$. 
        \end{enumerate}
    \end{itemize}
    It is clear that the time complexity of the above process is
    \[
    O\rbra*{\sum_{i=2}^q \rbra*{ \abs*{l_i-l_{i-1}}+\abs*{r_i-r_{i-1}} } \log\rbra*{\sigma} } = O\rbra*{\rbra*{Bq + \frac{n^2}{B}} \log\rbra*{\sigma}} = O\rbra*{ n \sqrt{q} \log\rbra*{\sigma}}
    \]
    by setting $B = \floor*{n/\sqrt{q}}$. Now we have access to the eertree of substring $s\substr{l_i}{r_i}$ for each $1 \leq i \leq q$. In the following, we will consider different types of queries separately. 
    
    \paragraph{Counting distinct palindromic substrings.}
    
    The number of distinct palindromic substrings, also known as the palindromic complexity, of a string, has been studied in the literature (e.g., \cite{ABCD03,AAK10,GPR10,RS17}).
    It was noted in \cite{RS18} that the number of distinct palindromic substrings of string $s$ equals to the number of nodes in the eertree of $s$ (minus $1$). Immediately, we have the following result on counting distinct palindromic substrings. 
    
    \begin{corollary} \label{corollary:count-pal}
        Offline range queries of type \textsc{Counting Distinct Palindromic Substrings} can be solved with time
        complexity $O\rbra*{n\sqrt{q}\log\rbra*{\sigma}}$.
    \end{corollary}
    
    It was shown in \cite{RS17} that range queries of type \textsc{Counting Distinct Palindromic Substrings} can be solved in $O\rbra*{\rbra*{n+q}\log\rbra*{n}}$ time assuming $\sigma = O\rbra*{1}$. Our algorithm in Corollary \ref{corollary:count-pal} can be faster than the one given in \cite{RS17} when $q = o\rbra*{\log^2\rbra*{n}}$ or $q = \omega\rbra*{n^2/\log^2\rbra*{n}}$.
    
    \paragraph{Longest palindromic substring.}
    
    Finding the longest palindromic substring has been extensively studied in the literature (e.g., \cite{Man75,ABG95,Jeu94,Gus97,BEMTSA14,AB19,GMSU19,FNI+21,CPR22,LGS22}). To answer range queries of type \textsc{Longest Palindromic Substring}, we use $n$ linked lists to store all palindromic substrings of the current string. Specifically, let $\texttt{list}\sbra*{i}$ be the double-linked list to store palindromic substrings of length $i$ for every $1 \leq i \leq n$, and let $\texttt{maxlen}$ indicate the maximum length over all these palindromic substrings. Initially, all lists are empty and $\texttt{maxlen} = 0$. We can maintain these data as follows.
    \begin{itemize}
        \item When a node $u$ is added to the double-ended eertree $\mathcal{T}$, 
        \begin{enumerate}
            \item Add $u$ to $\texttt{list}\sbra*{\len\rbra*{u}}$.
            \item Set $\texttt{maxlen} \gets \max\cbra*{\texttt{maxlen}, \len\rbra*{u}}$.
        \end{enumerate}
        \item When a node $u$ is deleted from the double-ended eertree $\mathcal{T}$, 
        \begin{enumerate}
            \item Delete $u$ from $\texttt{list}\sbra*{\len\rbra*{u}}$.
            \item Repeat the following until $\texttt{list}\sbra*{\texttt{maxlen}}$ is not empty:
            \begin{itemize}
                \item Set $\texttt{maxlen} \gets \texttt{maxlen} - 1$.
            \end{itemize}
        \end{enumerate}
    \end{itemize}
    
    To find the longest palindromic substring of the current string in the eertree $\mathcal{T}$, just return any element in $\texttt{list}\sbra*{\texttt{maxlen}}$. The correctness is trivial. 
    In the above process, there is only one loop that decrements $\texttt{maxlen}$ until $\texttt{list}\sbra*{\texttt{maxlen}}$ is not empty. 
    It can be seen that the loop repeats no more than twice by the following observation. 
    
    \begin{proposition}
        If a node $u$ with $\len\rbra*{u} > 2$ is being deleted from a double-ended eertree due to \textup{\texttt{pop\_back}} or \textup{\texttt{pop\_front}} operations, then after node $u$ is deleted, there is a node $v$ in the eertree such that $\len\rbra*{v} = \len\rbra*{u} - 2$. 
    \end{proposition}
    \begin{proof}
        Choose node $v$ such that $\str\rbra*{v} = \str\rbra*{u}\substr{1}{\len\rbra*{u}-1}$. It is trivial that $\len\rbra*{v} = \len\rbra*{u} - 2$ and $\str\rbra*{v}$ occurs at least once in the string after the \textup{\texttt{pop\_back}} or \textup{\texttt{pop\_front}} operation.
    \end{proof}
    
    Therefore, we have the following result on finding the longest palindromic substring. 
    
    \begin{corollary} \label{corollary:longest-pal}
        Offline range queries of type \textsc{Longest Palindromic Substring} can be solved with time 
        complexity $O\rbra*{n\sqrt{q}\log\rbra*{\sigma}}$.
    \end{corollary}
    
    It was shown in \cite{ACPR20} that range queries of type \textsc{Longest Palindromic Substring} can be solved in $O\rbra*{n \log^2 \rbra*{n} + q \log\rbra*{n}}$ time assuming $\sigma = O \rbra*{1}$. Our algorithm in Corollary \ref{corollary:longest-pal} can be faster than the one given in \cite{ACPR20} when $q = o\rbra*{\log^4\rbra*{n}}$ or $q = \omega\rbra*{n^2/\log^2\rbra*{n}}$.
    
    \paragraph{Shortest unique palindromic substring.}
    
    Motivated by molecular biology \cite{KYK+92,YYK+92}, algorithms about the shortest unique palindromic substring was investigated in a series of works \cite{INM+18,WNI+20,FM21,MF22}. To find the shortest unique palindromic substring with respect to an interval of a string, they introduced the notion of minimal unique palindromic substrings (MUPSs). Here, a palindromic substring $s\substr{i}{j}$ of string $s$ is called a MUPS of $s$, if $s\substr{i}{j}$ occurs exactly once in $s$ and $s\substr{i+1}{j-i}$ either is empty or occurs at least twice. The set of MUPSs can be maintained after single-character substitution \cite{FM21}, and in a sliding window \cite{MWN+22}.
    
    In our case, we are to find the shortest unique palindromic substring of a string $s$, which is actually the shortest MUPSs of $s$. To this end, we are going to maintain the set of all MUPSs of the current string. The following lemma shows that whether a palindrome is a MUPS can be reduced to uniqueness checking. 
    
    \begin{lemma} [MUPS checking via uniqueness \cite{MWN+22}]
        A palindrome $t$ is a MUPS of string $s$, if and only if $t$ is unique in $s$ and $t\substr{1}{\abs*{t}-1}$ is not unique in $s$, where the empty string $\epsilon$ is considered to be not unique in any string. 
    \end{lemma}
    
    By Lemma \ref{lemma:unique-by-cnt}, whether a palindromic substring $t$ of string $s$ is a MUPS of $s$ can be checked in $O\rbra*{1}$ time, given access to the node of $t$ in the eertree. Now we will maintain the set $\texttt{MUPS}$ to store all MUPSs of the current string. Initially, the set $\texttt{MUPS}$ is empty. For each deque operation on the double-ended eertree $\mathcal{T}$, do the following.
    
    \begin{itemize}
        \item For each $\texttt{push\_back}\rbra*{c}$ operation on double-ended eertree $\mathcal{T}$, 
        \begin{enumerate}
            \item Perform $\texttt{push\_back}\rbra*{c}$ on $\mathcal{T}$.
            \item Let $u = \node\rbra*{\sufpal\rbra*{s, \abs*{s}}}$, where $s$ is the current string. 
            \item Maintain $\texttt{MUPS}$ according to whether $u$ as well as $\prev\rbra*{u}$ is a MUPS of $s$.
        \end{enumerate}
        
        \item For each $\texttt{push\_front}\rbra*{c}$ operation on double-ended eertree $\mathcal{T}$, 
        \begin{enumerate}
            \item Perform $\texttt{push\_front}\rbra*{c}$ on $\mathcal{T}$.
            \item Let $u = \node\rbra*{\prepal\rbra*{s, 1}}$, where $s$ is the current string. 
            \item Maintain $\texttt{MUPS}$ according to whether $u$ as well as $\prev\rbra*{u}$ is a MUPS of $s$.
        \end{enumerate}
        
        \item For each \texttt{pop\_back} operation on double-ended eertree $\mathcal{T}$, 
        \begin{enumerate}
            \item Let $u = \node\rbra*{\sufpal\rbra*{s, \abs*{s}}}$, where $s$ is the current string. 
            \item Perform $\texttt{pop\_back}\rbra*{}$ on $\mathcal{T}$.
            \item Maintain $\texttt{MUPS}$ according to whether $u$ as well as $\prev\rbra*{u}$ is a MUPS of $s$.
        \end{enumerate}
        
        \item Before each \texttt{pop\_front} operation on double-ended eertree $\mathcal{T}$, 
        \begin{enumerate}
            \item Let $u = \node\rbra*{\prepal\rbra*{s, 1}}$, where $s$ is the current string. 
            \item Perform $\texttt{pop\_front}\rbra*{}$ on $\mathcal{T}$.
            \item Maintain $\texttt{MUPS}$ according to whether $u$ as well as $\prev\rbra*{u}$ is a MUPS of $s$.
        \end{enumerate}
    \end{itemize}
    
    Here, step 3 of each case can be maintained as follows. 
    \begin{enumerate}
        \item If $u$ a MUPS of $s$, add $u$ to $\texttt{MUPS}$; otherwise, remove $u$ from $\texttt{MUPS}$. 
        \item If $\len\rbra*{u} > 2$, do the same for $\prev\rbra*{u}$. That is, if $\prev\rbra*{u}$ a MUPS of $s$, add $\prev\rbra*{u}$ to $\texttt{MUPS}$; otherwise, remove $\prev\rbra*{u}$ from $\texttt{MUPS}$. 
    \end{enumerate}
    
    To answer each query of type \textsc{Shortest Unique Palindromic Substring}, just return any element in $\texttt{MUPS}$ with the minimum length. To achieve this, we can use the binary search tree to maintain $\texttt{MUPS}$, which introduces an extra $O\rbra*{\log\rbra*{n}}$ in the complexity. Then, we have the following result.
    
    \begin{corollary} \label{corollary:sups}
        Offline range queries of type \textsc{Shortest Unique Palindromic Substring} can be solved with time 
        complexity $O\rbra*{n\sqrt{q}\rbra*{\log\rbra*{n}+\log\rbra*{\sigma}}}$.
    \end{corollary}
    
    It was shown in \cite{MWN+22} that the set of MUPSs can be maintained in the sliding window model, which is actually a special case of ours that $l_i \leq l_{i+1}$ and $r_i \leq r_{i+1}$ for every $1 \leq i < q$. In this special case, range queries of type \textsc{Shortest Unique Palindromic Substring} can be solved in time $O\rbra*{n\rbra*{\log\rbra*{n} + \log\rbra*{\sigma}} + q}$ by the sliding window technique in \cite{MWN+22}. However, their technique seems not applicable in our more general case. 
    
    \paragraph{Shortest absent palindrome.}
    
    Minimal absent palindromes (MAPs) is a palindromic version of the notion of minimal absent words, which was extensively studied in the literature \cite{CMRS00,MRS02,CC12}. The set of MAPs can be maintained in the sliding window model \cite{MWN+22}. Here, a palindrome $t$ is called a MAP of a string $s$, if $t$ does not occur in $s$ but $t\substr{1}{\abs*{t}-1}$ does, where the empty string $\epsilon$ is considered to occur in any string. 
    
    In our case, we are to find the shortest absent palindrome of a string, which is actually the shortest MAP of the string. The following lemma shows an upper bound of the length of the shortest absent palindrome of a string.
    
    \begin{lemma} \label{lemma:length-sap}
        Suppose $s$ is a string of length $n$. Then, the length of the shortest absent palindrome of $s$ is $\leq \ceil*{2 \log_{\sigma} \rbra*{n}} + 1$, where $\sigma$ is the size of the alphabet.
    \end{lemma}
    \begin{proof}
        Let $k = \ceil*{2 \log_{\sigma} \rbra*{n}} + 1$. The number of palindromes of length $k$ is $\sigma^{\ceil*{k/2}} > n$. On the other hand, the number of distinct non-empty palindromic substrings of string $s$ is at most $n$ \cite{DJP01}. We conclude that there must be a palindrome of length $k$ that does not occur in $s$, and these yield the proof.
    \end{proof}
    
    Our main idea to find the shortest absent palindrome is to maintain the set of MAPs. In the implementation, we use linked-lists to store MAPs of each length indirectly. Specifically, let $\texttt{list}\sbra*{i}$ be the double-linked list to store palindromic substrings $t$ of length $i$ for each $i$ with $\next\rbra*{\node\rbra*{t}, c} = \nullptr$ for at least one character $c$. By Lemma \ref{lemma:length-sap}, we can choose the range of $i$ as $1 \leq i \leq \ceil*{2 \log_{\sigma} \rbra*{n}} + 1$, and palindromic substrings of length beyond this range are ignored. We maintain these data as follows.
    
    \begin{itemize}
        \item When a node $u$ is added to the double-ended eertree $\mathcal{T}$, 
        \begin{enumerate}
            \item Add $u$ to $\texttt{list}\sbra*{\len\rbra*{u}}$.
            \item If $\next\rbra*{\prev\rbra*{u}, c} \neq \nullptr$ for every character $c$, delete $\prev\rbra*{u}$ from $\texttt{list}\sbra*{\len\rbra*{\prev\rbra*{u}}}$.
        \end{enumerate}
        \item When a node $u$ is deleted from the double-ended eertree $\mathcal{T}$, 
        \begin{enumerate}
            \item Delete $u$ from $\texttt{list}\sbra*{\len\rbra*{u}}$.
            \item If $\prev\rbra*{u}$ is not in $\texttt{list}\sbra*{\len\rbra*{\prev\rbra*{u}}}$, add $\prev\rbra*{u}$ to $\texttt{list}\sbra*{\len\rbra*{\prev\rbra*{u}}}$.
        \end{enumerate}
    \end{itemize}
    The above procedure can be maintained in $O\rbra*{1}$ time per operation. 
    
    We can answer each query of type \textsc{Shortest Absent Palindrome} as follows with the current string $s$.
    
    \begin{enumerate}
        \item If there is a character $c$ that does not occur in $s$, i.e., $\next\rbra*{\odd, c} = \nullptr$, return $c$.
        \item If there is a character $c$ such that $cc$ does not occur in $s$, i.e., $\next\rbra*{\even, c} = \nullptr$, return $cc$.
        \item For each $1 \leq i \leq \ceil*{2 \log_{\sigma} \rbra*{n}} + 1$ in this order, if $\texttt{list}\sbra*{i}$ is not empty, do the following:
        \begin{itemize}
            \item Let $u$ be any node stored in $\texttt{list}\sbra*{i}$, and find any character $c$ such that $\next\rbra*{u, c} = \nullptr$. Return $c \str\rbra*{u} c$.
        \end{itemize}
    \end{enumerate}
    
    Finally, we have the following result.
    
    \begin{corollary} \label{corollary:sap}
         Offline range queries of type \textsc{Shortest Absent Palindrome} can be solved with time 
         complexity $O\rbra*{n\sqrt{q}\log\rbra*{\sigma} + q \log\rbra*{n} / \log\rbra*{\sigma} }$.
    \end{corollary}
    
    \begin{remark}
        The algorithms mentioned in Corollary \ref{corollary:count-pal}, \ref{corollary:longest-pal}, \ref{corollary:sups} and \ref{corollary:sap} require time and space complexity roughly $O\rbra*{n \sqrt{q} \log\rbra*{\sigma}}$ because each operation introduces incremental $O\rbra*{\log\rbra*{\sigma}}$ space due to the online construction of eertrees by Theorem \ref{thm:eertree-surface-recording}. Indeed, we have an alternative implementation with $O\rbra*{n\sqrt{q}\sigma}$ time and $O\rbra*{n\sigma}$ space. This is achieved by using a copy-based algorithm to store $\directlink\rbra*{v, c}$ for every character $c$, which requires $O\rbra*{\sigma}$ time and space per operation. Since the length of substrings of $s$ is always $\leq \abs*{s} = n$, the space used in the copy-based algorithm is $O\rbra*{n \sigma}$ independent of the parameter $q$. When $\sigma$ is small enough such that $q = \omega\rbra*{\sigma^2/\log^2\rbra*{\sigma}}$, the space used in the copy-based algorithm can be much smaller than that in the $O\rbra*{n \sqrt{q} \log\rbra*{\sigma}}$-time algorithm, with the same time complexity up to a small factor. In particular, when $\sigma = O\rbra*{1}$ is a constant justified in Section \ref{sec:main-results}, we can obtain an algorithm with $O\rbra*{n\sqrt{q}}$ time and $O\rbra*{n}$ space. 
    \end{remark}
    
    \subsubsection{Online queries}

    To handle online queries, we first study how to make our double-ended eertree fully persistent. A persistent data structure is a collection of data structures (of the same type), called versions, ordered by the time they are created. A data structure is called fully persistent if every version of it can be both accessed and modified. 
    
    \begin{theorem} [Persistent double-ended eertrees] \label{thm:persistent-double-ended-eertree}
        Fully persistent double-ended eertrees can be implemented with worst-case time and space complexity per operation $O\rbra*{\log\rbra*{n}+\log\rbra*{\sigma}}$, where $\sigma$ is the size of the alphabet and $n$ is the length of the string in the current version. More precisely,
        \begin{itemize}
            \item A \textup{\texttt{push\_back}} or \textup{\texttt{push\_front}} operation requires worst-case time and space complexity $O\rbra*{\log\rbra*{n}+\log\rbra*{\sigma}}$.
            \item A \textup{\texttt{pop\_back}} or \textup{\texttt{pop\_front}} operation requires worst-case time and space complexity $O\rbra*{\log\rbra*{n}}$.
        \end{itemize}
    \end{theorem}
    \begin{proof}
        The implementation of persistent double-ended eertree is based on Algorithm \ref{algo:surface:push_back} to \ref{algo:surface:pop_front}, wherein $\data$, $\texttt{presurf}$ and $\texttt{sufsurf}$ are easy to make persistent in $O\rbra*{\log\rbra*{n}}$ time and space per operation by binary search trees. We only need to further consider how to add and delete nodes from eertrees persistently. To this end, we maintain a persistent set $\texttt{nodes}$ (by, for example, binary search trees) to contain all nodes in the eertree of the current version. 
        
        \begin{itemize}
            \item 
            To add a node $v$, we need to
            \begin{itemize}
                \item add $v$ into $\texttt{nodes}$ with time and space complexity $O\rbra*{\log\rbra*{n}}$,
                \item create $\directlink\rbra*{v, \cdot}$ from $\directlink\rbra*{\prev\rbra*{v}, \cdot}$ with time and space complexity $O\rbra*{\log\rbra*{\sigma}}$ by Lemma \ref{lemma:dlink},
                \item maintain $\linkcnt\rbra*{s, \link\rbra*{v}}$ and $\texttt{cnt}\sbra*{v}$ with time and space complexity $O\rbra*{\log\rbra*{n}}$.
            \end{itemize}
            Thus adding a node requires $O\rbra*{\log\rbra*{n}+\log\rbra*{\sigma}}$ time and space. 
            
            \item To delete a node $v$, we need to
            \begin{itemize}
                \item maintain $\linkcnt\rbra*{s, \link\rbra*{v}}$ and $\texttt{cnt}\sbra*{v}$ with time and space complexity $O\rbra*{\log\rbra*{n}}$,
                \item delete $v$ from $\texttt{nodes}$ with time and space complexity $O\rbra*{\log\rbra*{n}}$.
            \end{itemize}
            Thus deleting a node requires $O\rbra*{\log\rbra*{n}}$ time and space.
        \end{itemize}
        
        Therefore, double-ended eertrees can be implemented fully persistently with worst-case time and space complexity per operation $O\rbra*{\log\rbra*{n}+\log\rbra*{\sigma}}$.
    \end{proof}
    
    Before holding each type of query, we first describe the main idea on how to prepare and store necessary information in order to obtain a time-space trade-off. We partition the string $s$ into blocks, each of size $B$. Specifically, for every $1 \leq i \leq n/B$, the start position of the $i$-th block is $\ell_i = B\rbra*{i-1}+1$. Now for each $1 \leq i \leq n/B$, we are going to store the double-ended eertree of $s\substr{\ell_i}{r}$ for every $\ell_i \leq r \leq n$. There are $O\rbra*{n^2/B}$ double-ended eertrees to be stored. To achieve this, we build these double-ended eertrees as follows.
    \begin{itemize}
        \item For every $1 \leq i \leq n/B$, let $\mathcal{T}$ be the persistent double-ended eertree of the empty string.
        \begin{itemize}
            \item For every $\ell_i \leq r \leq n$ in this order,
            \begin{enumerate}
                \item perform $\texttt{push\_back}\rbra*{s\sbra*{r}}$ on $\mathcal{T}$,
                \item store (the pointer to) $\mathcal{T}$ as the double-ended eertree of $s\substr{\ell_i}{r}$. 
            \end{enumerate}
        \end{itemize}
    \end{itemize}
    By Theorem \ref{thm:persistent-double-ended-eertree}, it is clear that it takes $O\rbra*{n^2\rbra*{\log\rbra*{n}+\log\rbra*{\sigma}}/B}$ time and space to prepare the $O\rbra*{n^2/B}$ (persistent) double-ended eertrees. 
    
    To answer a query on substring $s\substr{l}{r}$, if $r-l+1 \leq B$, then build the double-ended eertree of $s\substr{l}{r}$ directly; otherwise, let $\mathcal{T}$ be the persistent double-ended eertree of $s\substr{B\floor*{\rbra*{l-1}/B}+1}{r}$, then perform $\texttt{push\_front}$ with characters $s\sbra*{B\floor*{\rbra*{l-1}/B}}, \dots, s\sbra*{l}$ in this order. It is clear that it takes $O\rbra*{B\rbra*{\log\rbra*{n}+\log\rbra*{\sigma}}}$ time and space to prepare the persistent double-ended eertree of $s\substr{l}{r}$. If the number $q$ of queries is known in advance, it is optimal to set $B = \floor*{n/\sqrt{q}}$, and then the time and space complexity becomes $O\rbra*{n\sqrt{q} \rbra*{\log\rbra*{n}+\log\rbra*{\sigma}} }$. After considering each type of query with necessary persistent auxiliary data, we have the following results.
    
    \begin{corollary} \label{corollary:online}
        Online range queries of type \textsc{Counting Distinct Palindromic Substrings}, \textsc{Longest Palindromic Substring} and \textsc{Shortest Unique Palindromic Substring} can be solved with time
        complexity $O\rbra*{n\sqrt{q}\rbra*{ \log\rbra*{n} + \log\rbra*{\sigma}} }$, and those of type \textsc{Shortest Absent Palindrome} can be solved with time complexity $O\rbra*{n\sqrt{q}\rbra*{ \log\rbra*{n} + \log\rbra*{\sigma}} + q \log\rbra*{n} \log\rbra*{\log\rbra*{n}} / \log\rbra*{\sigma} }$, if the number $q$ of queries is known in advance.
    \end{corollary}
    
    \begin{proof}
        It is straightforward to obtain these results from Corollary \ref{corollary:count-pal}, \ref{corollary:longest-pal}, \ref{corollary:sups} and \ref{corollary:sap} equipped with persistent data structures such as persistent binary search trees and persistent double-ended eertrees. Here, we especially mention that the doubly-logarithmic factor of $n$ is introduced in \textsc{Shortest Absent Palindrome} due to persistent data structures used for maintaining $\texttt{list}\sbra*{i}$ of size $O\rbra*{\log\rbra*{n}}$ in Corollary \ref{corollary:sap}.
    \end{proof}
    
    \subsection{Enumerating rich strings with a given word}
    
    Palindromic rich strings have been extensively studied \cite{GJWZ09,RR09,BDLGZ09,Ves14}. Recently, the number of rich strings of length $n$ was studied \cite{RS18,GSS16,Ruk17}. Especially, the number of binary rich strings of length $n$ (cf. sequence A216264 in OEIS \cite{Sha13}) was efficiently computed by eertree in \cite{RS18}, and they thus deduced an $O\rbra*{{1.605}^n}$ upper bound (as noted in \cite{Ruk17}). Shortly after, a lower bound $\Omega\rbra*{37.6^{\sqrt{n}}}$ was given in \cite{GSS16}.
    
    We consider a computational task to enumerate rich strings with a given word, with a formal description in Problem \ref{prob:enum-rich}.
    
    \begin{problem} [\textsc{Counting Rich Strings with Given Word}] \label{prob:enum-rich}
        Given a string $t$ of length $n$ and a number $k$, count the number of palindromic rich strings $s$ of length $n + k$ such that $t$ is a substring of $s$. 
    \end{problem}
    
    Let $\sigma$ be the size of the alphabet. There are roughly $O\rbra*{k\sigma^k}$ strings of length $n + k$ with the given substring $t$ of length $n$. A simple solution to Problem \ref{prob:enum-rich} is to enumerate each of the $O\rbra*{k\sigma^k}$ candidates and check its richness by \cite{GPR10} in $O\rbra*{n+k}$ time, thereby with total time complexity at least $O\rbra*{\rbra*{n+k}k\sigma^k}$, where $n$ and $\sigma^k$ are multiplicative in the complexity. Using our double-ended eertree, we are able to improve the time complexity such that $n$ and $\sigma^k$ are additive. 
    
    \begin{corollary} \label{corollary:count-rich-string-with-given-word}
        There is an algorithm for \textsc{Counting Rich Strings with Given Word} with time complexity $O\rbra*{n\sigma + k\sigma^k}$, where $\sigma$ is the size of the alphabet.
    \end{corollary}
    
    \begin{proof}
        The basic idea is to enumerate all possible characters being added at the front and the back of the string. It is clear that there are $\rbra*{k+1} \sigma^k$ ways to add $k$ characters to both ends of a string. 
        
        Suppose we are given a string $t$ of length $n$ and want to enumerate all strings $s$ of length $n + k$ with $t$ being its substring. To remove duplicate enumerations, we only enumerate strings of the form $xty$ such that $t$ does not occur in $t\substr{2}{\abs*{t}} y$. To this end, we build the Aho-Corasick automaton \cite{AC75} of a single string $t$ \footnote{In fact, the Knuth-Morris-Pratt algorithm \cite{KMP77} for pattern matching can also achieve our goal. The Aho-Corasick automaton we use here is for its efficient state transitions.}. This can be done in $O\rbra*{n\sigma}$ time. Recall that an Aho-Corasick automaton is a trie-like structure with each of its node corresponding to a unique string. Especially, the root corresponds to the empty string. For our purpose, we only need the transitions between its nodes. Specifically, the transition $\delta\rbra*{u, c}$ is defined for every node $u$ and character $c$, which points to the node $v$ of the largest $\len\rbra*{v}$ such that $\str\rbra*{v}$ is a proper suffix of $\str\rbra*{u}$. Here, we follow the notations $\str\rbra*{\cdot}$ and $\len\rbra*{\cdot}$ as for eertrees. 
        
        With the Aho-Corasick automaton of string $t$, we can enumerate every distinct string $s$ with $t$ being its substring, and simultaneously maintain the double-ended eertree of $s$. Our algorithm consists of two parts (see Algorithm \ref{algo:counting-rich-with-given-word}). 
        \begin{enumerate}
            \item The first part is to enumerate all characters added at the back, with $t$ only occurring once in the resulting string (see Algorithm \ref{algo:enum_back}). 
            \item The second part is to enumerate all characters added at the front, and then check whether the resulting string is palindromic rich (see Algorithm \ref{algo:enum_front}). 
        \end{enumerate}
        
        \begin{algorithm}[!htp]
            \caption{An algorithm for \textsc{Counting Rich Strings with Given Word}}
            \label{algo:counting-rich-with-given-word}
            \begin{algorithmic}[1]
            \Require string $t$ of length $n$, and number $k$.
            \Ensure the number of palindromic rich strings $s$ of length $n + k$ such that $t$ is a substring of $s$.
            \State Build the double-ended eertree $\mathcal{T}$ of $t$.
            \State Build the Aho-Corasick automaton of $t$ with transitions $\delta\rbra*{\cdot, \cdot}$.
            \State Let $u_t$ be the node corresponding to $t$ in the Aho-Corasick automaton of $t$. 
            \State $\mathit{ans} \gets 0$.
            \State $\texttt{enum\_back}\rbra*{t, u_t}$.
            \State \Return $\mathit{ans}$.
            \end{algorithmic}
        \end{algorithm}
        
        \begin{algorithm}[!htp]
            \caption{$\texttt{enum\_back}\rbra*{s, u}$}
            \label{algo:enum_back}
            \begin{algorithmic}[1]
            \State $\texttt{enum\_front}\rbra*{s}$.
            \If {$\abs*{s} \geq n+k$}
                \State \Return
            \EndIf
            \For {each character $c$}
                \State $v \gets \delta\rbra*{u, c}$.
                \If {$v \neq u_t$}
                    \State $\texttt{push\_back}\rbra*{c}$ on double-ended tree $\mathcal{T}$.
                    \State $\texttt{enum\_back}\rbra*{sc, v}$.
                    \State $\texttt{pop\_back}\rbra*{c}$ on double-ended tree $\mathcal{T}$.
                \EndIf
            \EndFor
            \end{algorithmic}
        \end{algorithm}
        
        \begin{algorithm}[!htp]
            \caption{$\texttt{enum\_front}\rbra*{s}$}
            \label{algo:enum_front}
            \begin{algorithmic}[1]
            \If {$\abs*{s} = n+k$}
                \State $\mathit{num} \gets$ the number of distinct palindromic substrings in $s$ by the double-ended eertree $\mathcal{T}$.
                \If {$\mathit{num} = n + k + 1$, i.e., $s$ is palindromic rich}
                    \State $\mathit{ans} \gets \mathit{ans} + 1$.
                \EndIf
                \State \Return
            \EndIf
            \For {each character $c$}
                \State $\texttt{push\_front}\rbra*{c}$ on double-ended tree $\mathcal{T}$.
                \State $\texttt{enum\_front}\rbra*{cs}$.
                \State $\texttt{pop\_front}\rbra*{c}$ on double-ended tree $\mathcal{T}$.
            \EndFor
            \end{algorithmic}
        \end{algorithm}
        
        It is clear that the time complexity of our algorithm is $O\rbra*{n\sigma + k\sigma^k}$. In order to prove its correctness, we only need to show that the recursive function correctly enumerates every string with $t$ being it substring exactly once. 
        To see this, we first show that every string of length $n+k$ with substring $t$ will be enumerated at least once. Suppose $s = xty$, where $\abs*{x} + \abs*{y} = k$. If there are multiple representations of $s$ in the form $s = xty$, we choose the one with the shortest $y$. Then, we can see that $t$ is not a substring of $t\substr{2}{\abs*{t}} y$. This implies that $\texttt{enum\_back}\rbra*{s, u}$ in Algorithm \ref{algo:enum_back} will eventually reach the state with $s = ty$. Upon calling $\texttt{enum\_front}\rbra*{ty}$, it is trivial that $xty$ will be enumerated. 
        
        It remains to show that every string of length $n + k$ with substring $t$ will be enumerated at most once. We define a path $t \to ty \to xty$ to represent the recursive procedure, which means that the algorithm starts by calling $\texttt{enum\_back}\rbra*{t}$, then calls $\texttt{enum\_front}\rbra*{ty}$, and finally achieves $xty$ of length $n+k$. It has already been shown in the above that there is at least once such path for every string of length $n+k$ with substring $t$. If there are two different enumerations of the same string $s$, then there are two different paths $t \to ty_1 \to x_1ty_1$ and $t \to ty_2 \to x_2ty_2$ such that $s = x_1 t y_1 = x_2 t y_2$ with $\rbra*{x_1, y_1} \neq \rbra*{x_2, y_2}$. Without loss of generality, we assume that $\abs*{y_1} < \abs*{y_2}$. It can be seen that the path $t \to ty_2 \to x_2ty_2$ is impossible as follows. There is a non-empty string $w$ such that $ty_2 = wty_1$. Before reaching the state that $s = ty_2$ in $\texttt{enum\_back}\rbra*{s, u}$, it must reach the state that $s = wt$ and $u = u_t$. Since $w$ is not empty, there is no way to call $\texttt{enum\_back}\rbra*{wt, u_t}$ because of the guard $v \neq u_t$ in $\texttt{enum\_back}\rbra*{s, u}$ of Algorithm \ref{algo:enum_back}.
    \end{proof}

\end{document}